\documentclass[aps,pra,10pt,twocolumn,superscriptaddress,showkeys,letterpaper]{revtex4-2}

\pdfoutput=1

\usepackage[T1]{fontenc}
\usepackage[utf8]{inputenc}
\usepackage{color}
\usepackage{babel}
\usepackage{mathtools}
\usepackage{amsmath}
\usepackage{amsthm}
\usepackage{amssymb}
\usepackage{graphicx}
\usepackage[pdfusetitle,
 bookmarks=true,bookmarksnumbered=false,bookmarksopen=false,
 breaklinks=false,pdfborder={0 0 0},pdfborderstyle={},backref=false,colorlinks=true]
 {hyperref}
\hypersetup{
 linkcolor=magenta, urlcolor=blue, citecolor=blue}

\makeatletter
\theoremstyle{plain}
\newtheorem{thm}{\protect\theoremname}
\theoremstyle{plain}
\newtheorem{defn}[thm]{\protect\definitionname}
\theoremstyle{plain}
\newtheorem{assumption}[thm]{\protect\assumptionname}
\theoremstyle{plain}
\newtheorem{prop}[thm]{\protect\propositionname}
\theoremstyle{plain}
\newtheorem{rem}[thm]{\protect\remarkname}
\theoremstyle{plain}
\newtheorem{lyxalgorithm}[thm]{\protect\algorithmname}
\theoremstyle{plain}
\newtheorem{lem}[thm]{\protect\lemmaname}
\theoremstyle{plain}
\newtheorem{cor}[thm]{\protect\corollaryname}

\DeclareMathOperator{\Tr}{Tr}

\DeclareMathOperator{\Herm}{Herm}
\DeclareMathOperator{\id}{id}

\DeclareMathOperator{\FD}{FD}
\DeclareMathOperator{\signum}{sgn}

\allowdisplaybreaks

\makeatother

\providecommand{\algorithmname}{Algorithm}
\providecommand{\assumptionname}{Assumption}
\providecommand{\corollaryname}{Corollary}
\providecommand{\definitionname}{Definition}
\providecommand{\lemmaname}{Lemma}
\providecommand{\propositionname}{Proposition}
\providecommand{\remarkname}{Remark}
\providecommand{\theoremname}{Theorem}

\begin{document}
\title{Fermi--Dirac thermal measurements: A framework for\\quantum hypothesis testing and semidefinite optimization}

\author{Nana Liu}

\affiliation{Institute of Natural Sciences, School of Mathematical Sciences, Ministry of Education Key Laboratory in Scientific and Engineering Computing, and Global College, Shanghai Jiao Tong University, Shanghai 200240, China}

\author{Mark M. Wilde}
\affiliation{School of Electrical and Computer Engineering,
Cornell University, Ithaca, New York 14850, USA} 
\date{\today}

\begin{abstract}
Quantum measurements are the means by which we recover messages encoded into quantum
states. They are at the forefront of quantum hypothesis testing, wherein
the goal is to perform an optimal measurement for arriving at a correct
conclusion. Mathematically, a measurement operator is Hermitian with eigenvalues in $[0,1]$. By noticing
that this constraint on each eigenvalue is the same as that imposed
on fermions by the Pauli exclusion principle, we interpret every eigenmode
of a measurement operator as an independent effective  fermionic mode. Under this
perspective, various objective functions in quantum hypothesis testing
can be viewed as the total expected energy associated with these fermionic
occupation numbers. By instead fixing a temperature and minimizing
the total expected fermionic free energy, we find that optimal measurements
for these modified objective functions are Fermi--Dirac thermal measurements,
wherein their eigenvalues are specified by Fermi--Dirac distributions. In the low-temperature
limit, their performance 
closely approximates that of optimal measurements for quantum hypothesis
testing, and we show that their parameters can be learned by classical
or hybrid quantum--classical optimization algorithms. This leads to a new quantum machine-learning model, termed \textit{Fermi--Dirac machines}, consisting of parameterized Fermi--Dirac thermal measurements---an alternative to quantum Boltzmann machines based on thermal states. Beyond
hypothesis testing, we show how general semidefinite optimization
problems can be solved using this approach, leading to a novel paradigm
for semidefinite optimization on quantum computers, in which the goal
is to implement thermal measurements rather than prepare thermal states.
Finally, we propose quantum algorithms for implementing Fermi--Dirac
thermal measurements, and we also propose second-order hybrid quantum--classical optimization
algorithms.
\end{abstract}

\maketitle

\tableofcontents{}

\section{Introduction}

\subsection{Background and motivation}

Quantum measurements play a fundamental role in quantum physics. They
are the means by which we retrieve messages encoded into quantum states,
and in many cases, we are interested in determining optimal measurements
that either minimize the probability of erroneously identifying an
encoded message or maximize an information measure.

Most generally, a measurement is mathematically characterized by a
positive operator-valued measure (POVM)~\cite{Kraus1983}, which is
a tuple $\left(M_{m}\right)_{m}$ of positive semidefinite operators
such that the completeness relation $\sum_{m}M_{m}=I$ holds. As such,
quantum measurements are specified by semidefinite constraints, and
so when one wishes to find an optimal measurement for a given quantum
information processing task, often it can be described as a semidefinite
program. This connection inextricably links quantum information theory
and semidefinite optimization, and this interplay has ultimately benefited
both fields of research~\cite{Audenaert2002,Brandao2017,Wang2018,Brandao2019,Apeldoorn2019,vanApeldoorn2020quantumsdpsolvers,Siddhu2022,Skrzypczyk2023,Patel2024,liu2025qthermoSDPs,minervini2025constrained,Nie2025}.

As a key scenario in which quantum measurements play a central role,
consider the task of symmetric quantum hypothesis testing, in which
$d$-dimensional states $\rho$ or $\sigma$ are prepared with equal
probability. Early work identified that the minimum error probability
for this task is as follows~\cite{Helstrom1967,Helstrom1969,Holevo1972}:
\begin{align}
 & \min_{M\in\mathcal{M}_{d}}\frac{1}{2}\Tr[\left(I-M\right)\rho]+\frac{1}{2}\Tr[M\sigma]\nonumber \\
 & =\frac{1}{2}\left(1+\min_{M\in\mathcal{M}_{d}}\Tr\!\left[M\left(\sigma-\rho\right)\right]\right)\label{eq:sym-bin-hypo-intro}\\
 & =\frac{1}{2}\left(1+\Tr\!\left[\Pi_{\sigma<\rho}\left(\sigma-\rho\right)\right]\right),
\end{align}
where $\mathcal{M}_{d}$ denotes the set of all $d$-dimensional measurement
operators (see~\eqref{eq:meas-ops-def}) and $\Pi_{\sigma<\rho}$
is the projection onto the strictly negative eigenspace of $\sigma-\rho$. That is,
\begin{equation}
    \Pi_{\sigma<\rho} \coloneqq \sum_{i : \lambda_i < 0} |i\rangle\!\langle i|,
\end{equation}
where $\sigma - \rho = \sum_i \lambda_i |i\rangle\!\langle i|$ is a spectral decomposition of $\rho - \sigma$.
The measurement $\left(\Pi_{\sigma<\rho},I-\Pi_{\sigma<\rho}\right)$
is a sharp threshold measurement, deciding that the prepared state
is $\rho$ if the first outcome occurs and deciding $\sigma$ otherwise.
In fact, we can substitute $\Pi_{\sigma<\rho}$ with $\Pi_{\sigma<\rho}+\frac{1}{2}\Pi_{\rho=\sigma}$,
where
\begin{equation}
\Pi_{\rho=\sigma} \coloneqq \sum_{i : \lambda_i =0} |i\rangle\!\langle i|    
\end{equation}
is the projection onto the zero eigenspace
of $\sigma-\rho$, and the resulting measurement is still optimal
for~\eqref{eq:sym-bin-hypo-intro}.

\subsection{Approach}

\label{subsec:Approach-intro}In this paper, we adopt a different
perspective on quantum hypothesis testing, and on semidefinite optimization
more broadly, inspired by the mathematical developments put forward
in \cite[Section~2.4]{lindsey2023}. The key perspective that we develop
here is to interpret the objective function
\begin{equation}
\Tr\!\left[M\left(\sigma-\rho\right)\right]\label{eq:intro-expected-energy}
\end{equation}
as the total expected energy of $d$ independent effective  fermions, so that the
aim of the optimization in~\eqref{eq:sym-bin-hypo-intro} is to minimize
this total expected energy. These fermions are not physical particles
but rather an effective description of the eigenmode occupations of
the measurement operator $M$, analogous to occupation-number representations
in second quantization. This approach is sensible after reflecting
on the fact that a spectral decomposition of a measurement operator
$M$ has the following form:
\begin{equation}
M=\sum_{i=1}^{d}m_{i}|\phi_{i}\rangle\!\langle\phi_{i}|.
\end{equation}
That $M$ is a measurement operator implies that each of its eigenvalues
satisfies the constraint $0\leq m_{i}\leq1$. As such, we can interpret
this constraint in terms of the Pauli exclusion principle, with $m_{i}$
being the probability that the eigenmode $\left(m_{i},|\phi_{i}\rangle\right)$
is occupied, and it not being possible for $m_{i}$ to exceed one,
consistent with the Pauli exclusion principle. So then each eigenmode
$\left(m_{i},|\phi_{i}\rangle\right)$ can be thought of as a fermion,
with expected energy given by $m_{i}\langle\phi_{i}|\left(\sigma-\rho\right)|\phi_{i}\rangle$.
Summing up the total expected energy of all $d$ eigenmodes leads
to the expression in~\eqref{eq:intro-expected-energy}.

Inspired by thermodynamics and related to the perspective of~\cite{liu2025qthermoSDPs}
connecting thermodynamics and optimization, consider that physical systems
operate at a strictly positive temperature $T>0$. As such, we can
alternatively minimize the sum of the free energies of the $d$ independent fermions,
which is given by
\begin{multline}
\sum_{i=1}^{d}\left[m_{i}\langle\phi_{i}|\left(\sigma-\rho\right)|\phi_{i}\rangle-T\cdot s(m_{i})\right]\\
=\Tr\!\left[M\left(\sigma-\rho\right)\right]-T\cdot S_{\FD}(M),\label{eq:intro-free-energy-sym-hypo-test}
\end{multline}
where
\begin{equation}
s(m_{i})\coloneqq-m_{i}\ln m_{i}-\left(1-m_{i}\right)\ln\!\left(1-m_{i}\right)
\end{equation}
is the entropy of the $i$th fermion and
\begin{equation}
S_{\FD}(M)\coloneqq-\Tr[M\ln M]-\Tr[(I-M)\ln(I-M)]
\end{equation}
is the Fermi--Dirac entropy of the measurement operator~$M$. 
The objective function in \eqref{eq:intro-free-energy-sym-hypo-test} is thus a modification of the objective function in \eqref{eq:intro-expected-energy}, with \eqref{eq:intro-free-energy-sym-hypo-test}  incorporating the negative term $-T\cdot S_{\FD}(M)$.

For
all $T>0$, the measurement operator achieving the minimum of the
right-hand side of~\eqref{eq:intro-free-energy-sym-hypo-test} is
a Fermi--Dirac thermal measurement operator $M_{T}$ of the form
\begin{equation}
M_{T}\coloneqq\left(e^{\frac{1}{T}\left(\sigma-\rho\right)}+I\right)^{-1},
\end{equation}
which converges to $\Pi_{\sigma<\rho}+\frac{1}{2}\Pi_{\rho=\sigma}$
in the zero-temperature limit (see Sections~\ref{subsec:Free-energy-minimization} and~\ref{subsec:Fermi-Dirac-thermal-measurements}
for details). Thus, the Fermi--Dirac measurement operator $M_{T}$
represents a smoothening of the sharp threshold measurement $\Pi_{\sigma<\rho}+\frac{1}{2}\Pi_{\rho=\sigma}$,
much like the Fermi--Dirac distribution of statistical mechanics or the logistic/sigmoid
function in neural-network and machine learning applications represents a smoothening
of a sharp threshold function. We can thus view $M_{T}$ as a matrix generalization of the sigmoid/logistic function. More generally, the framework outlined here is best suited for decision problems, in which the goal is to decide between two different possibilities, as is the case in binary classification or binary hypothesis testing problems.

We also note here that optimizing an objective function over measurement operators is rather distinct from optimizing over density matrices, the latter being considered in our related work \cite{liu2025qthermoSDPs,minervini2025constrained}. Indeed, as we mentioned above, the only constraint on a measurement operator is for all of its eigenvalues to be bounded between zero and one. These are thus independent constraints on the eigenvalues, so that the eigenmodes are in correspondence with $d$ independent effective fermionic modes. In contrast, the constraints on the eigenvalues of a density matrix are not independent:  not only do they need to take values between zero and one, but they also are constrained to sum to one. Thus, the form of the free energy in \cite{liu2025qthermoSDPs,minervini2025constrained} is different from that in~\eqref{eq:intro-free-energy-sym-hypo-test}, and related, minimizing \eqref{eq:intro-expected-energy} over the set of density matrices leads to a very different solution than that found by minimizing over measurement operators.

\subsection{Summary of contributions}

In this paper, we extend the approach and interpretation outlined
in Section~\ref{subsec:Approach-intro} to other scenarios of quantum
hypothesis testing and, more broadly, to semidefinite optimization.
Interestingly, these methods can be realized not only on classical
computers but also as hybrid quantum--classical algorithms, thus
representing a novel paradigm for semidefinite optimization on quantum
computers, in which the goal is to implement Fermi--Dirac thermal
measurements rather than prepare thermal states.

In more detail, the main contributions of our paper are as follows:
\begin{itemize}
\item We observe that the constraints on the eigenvalues of a measurement
operator are the same as those from the Pauli exclusion principle,
leading to an interpretation of the eigenmodes of a measurement operator
as fermions. This leads to a thermodynamic model of quantum measurements
in which the eigenmodes of a measurement operator behave as fermionic
occupation modes obeying the Pauli exclusion principle. We then reframe
quantum hypothesis testing problems as energy and free-energy minimization
problems, with Fermi--Dirac thermal measurements being optimal for
the latter minimization task. We thus provide a finite-temperature
generalization of optimal threshold measurements, replacing sharp
spectral thresholding by smooth Fermi--Dirac occupation rules.
\item We identify a positive-temperature chemical potential maximization
problem that is dual and equal to the free-energy minimization problem
(Proposition~\ref{prop:dual-FD-free-energy}). This dual optimization
problem is concave and has a smoothness parameter inversely proportional
to the temperature and proportional to the sum of the norms of the
operators involved in the constraints (Remark~\ref{rem:dual-concave-smoothness-param}).
It is thus amenable to gradient-based optimization algorithms, which
can be executed either on classical or quantum computers.
\item More broadly, this approach can be applied to general measurement
optimization problems, which are a special class of semidefinite optimization
problems and are essentially equivalent to them.
\item We provide a detailed error analysis comparing the original minimization
problem to the free-energy minimization problem, showing how to choose
the temperature of Fermi--Dirac thermal measurements in order to
achieve a desired error. In more nuanced error bounds, the temperature
depends on the dimension of the ground space and the spectral gap
of a Hamiltonian related to the optimization problem (Propositions
\ref{prop:approx-error-spectral-gap} and~\ref{prop:approx-err-no-FD-obj}).
\item As mentioned above, our approach opens up a new paradigm for solving
semidefinite programs on quantum computers, alternative to the conventional
approach~\cite{Brandao2017,Brandao2019,Apeldoorn2019,vanApeldoorn2020quantumsdpsolvers}
relying on thermal state preparation. In support of this, we propose
quantum algorithms for realizing Fermi--Dirac thermal measurements
when given access to samples of states needed to form them (Algorithms~\ref{alg:FD-thermal-alg} and~\ref{alg:DME-alg}).

\item We apply all of our findings to quantum hypothesis testing and binary classification, thus providing
a number of new insights for this domain of quantum information theory. In particular, we show how near-optimal measurements for hypothesis testing, of the Fermi--Dirac thermal form, can be learned via gradient ascent algorithms.
In addition to the above, our approach introduces a novel analytical
method for performing calculations associated with quantum hypothesis
testing, and the smoothness of Fermi--Dirac thermal measurements
could be beneficial for this purpose.

\item From another perspective, our paper introduces a novel model for quantum machine learning, called \textit{Fermi--Dirac machines}, which consists of parameterized Fermi--Dirac thermal measurements whose parameters can be learned by means of gradient-based, hybrid quantum--classical algorithms. This presents an alternative paradigm to quantum Boltzmann machines~\cite{Amin2018,Benedetti2017,Kieferova2017}, the latter being based on parameterized thermal states. The strong connection of Fermi--Dirac thermal measurements to sigmoid/logistic functions thus represents an alternative way for generalizing concepts from neural networks to quantum machine learning.
\end{itemize}

\subsection{Paper organization}

The rest of our paper is organized as follows:

In Section~\ref{sec:Notation-and-preliminaries}, we briefly review
some notation and concepts used throughout our paper.

Section~\ref{sec:General-measurement-optimization} contains some
of our main theoretical contributions. First, in Section~\ref{subsec:General-measurement-optimization}, we formulate a general
measurement optimization problem,
which is a particular kind of semidefinite optimization problem in
standard form \cite[Eq.~(4.51)]{Boyd2004}. We also derive its dual
and comment on the form of optimal measurements. We then provide a
fermionic interpretation of this measurement optimization problem
(Section~\ref{subsec:Interpretation-fermionic}), along the lines
presented in Section~\ref{subsec:Approach-intro}. We recast the measurement
optimization problem as a measurement free-energy optimization problem
and derive its dual, finding that Fermi--Dirac thermal measurements
are optimal for this modified optimization problem (Section~\ref{subsec:Free-energy-minimization}).
We define Fermi--Dirac thermal measurements in a general way and
discuss the zero- and infinite-temperature limits of them (Section~\ref{subsec:Fermi-Dirac-thermal-measurements}). In Section~\ref{subsec:Approximation-error}, we establish various bounds on
the approximation error between the measurement free-energy optimization
problem and the original optimization problem, as a function of the
temperature and dimension; more nuanced bounds involve the spectral
gap and dimension of the ground space of the Hamiltonian relevant
for the optimization problem. In Section~\ref{subsec:Gradient-and-Hessian},
we derive analytical expressions for the gradient and Hessian of the
dual of the measurement free-energy optimization problem, and in Section
\ref{subsec:Gradient-ascent-gen-meas-opt}, we establish a classical
gradient-ascent algorithm for the solving the dual problem, along
with guarantees on its runtime. Finally, Section~\ref{subsec:General-measurement-opteq-ineq}
generalizes the measurement optimization problem and all previous
developments to the case when there are inequality constraints in
addition to equality constraints, and Section~\ref{subsec:Solving-general-SDPs}
expands the developments to arbitrary semidefinite optimization problems.

In Section~\ref{sec:Quantum-algorithm-gen-meas-opt}, we present quantum
algorithms for solving the general measurement optimization problem
from Section~\ref{sec:General-measurement-optimization}, assuming
an input model in which the fixed Hermitian operators specifying the
problem are given by linear combinations of quantum states, to which
we assume sample access on a quantum computer. We sketch a quantum
algorithm for implementing a Fermi--Dirac thermal measurement (Section
\ref{subsec:Quantum-algorithm-FD-thermal}), which involves a combination
of Schr\"odingerization~\cite{Jin2024,Jin2023}, the power of one
qumode~\cite{Liu2016}, phase estimation~\cite{Kitaev1995,Shor1997,NielsenChuang2000},
and density matrix exponentiation~\cite{lloyd2014quantum,Kimmel2017,Go2025}.
As such, this could be a useful primitive for quantum hypothesis testing
going forward. We also sketch a quantum algorithm for estimating the
elements of a certain Hessian matrix (Section~\ref{subsec:Quantum-algorithm-Hessian}),
thus enabling second-order Newton steps when performing the required
optimization.

In Section~\ref{sec:Application-q-hypo-test}, we apply all of the
contributions from Sections~\ref{sec:General-measurement-optimization}
and~\ref{sec:Quantum-algorithm-gen-meas-opt} to various settings
of quantum hypothesis testing, including symmetric binary hypothesis
testing (Section~\ref{subsec:Symmetric-binary-HT}), binary classification (Section~\ref{subsec:binary-class}), asymmetric binary
hypothesis testing (Section~\ref{subsec:Asymmetric-binary-HT}), and
asymmetric composite hypothesis testing, in which the null hypothesis
is chosen from a finite set of states along with a finite number of
states, while the alternative hypothesis consists of a single state
(Section~\ref{subsec:Asymmetric-HT-composite-null}). This last setting
is somewhat general and is the one that maps rather directly onto
the general semidefinite optimization problem presented in Section~\ref{subsec:General-measurement-opteq-ineq}. The algorithms from
Section~\ref{sec:Quantum-algorithm-gen-meas-opt} can be interpreted
as methods for learning optimal measurements for quantum hypothesis
testing.

Finally, in Section~\ref{sec:Conclusion}, we conclude with a brief
summary of our findings and suggestions for future research.

\section{Notation and preliminaries}

\label{sec:Notation-and-preliminaries}

For $c\in\mathbb{N}$, we employ the shorthand:
\begin{equation}
\left[c\right]\coloneqq\left\{ 1,\ldots,c\right\} .
\end{equation}

For $d\in\mathbb{N}$, let $\mathcal{M}_{d}$ denote the set of $d\times d$
measurement operators:
\begin{align}
\mathcal{M}_{d} & \coloneqq\left\{ M\in\Herm_{d}:0\leq M\leq I\right\} ,\label{eq:meas-ops-def}
\end{align}
where $\Herm_{d}$ denotes the set of $d\times d$ Hermitian operators
and, for Hermitian operators $A$ and $B$, the notation $A\leq B$
indicates that $A-B$ is positive semidefinite. Let $\mathcal{D}_{d}$
denote the set of $d\times d$ density operators:
\begin{equation}
\mathcal{D}_{d}\coloneqq\left\{ \rho\in\Herm_{d}:\rho\geq0,\Tr[\rho]=1\right\} .\label{eq:dens-ops-def}
\end{equation}

Let $\left(A\right)_{-}$ denote the negative part of a Hermitian
matrix $A$, i.e.,
\begin{equation}
\left(A\right)_{-}\coloneqq\sum_{i:a_{i}<0}\left|a_{i}\right||i\rangle\!\langle i|,
\end{equation}
where $A=\sum_{i}a_{i}|i\rangle\!\langle i|$ is an eigendecomposition
of $A$. Note that $\left(A\right)_{-}$ is a positive semidefinite
matrix.

\section{General measurement optimization}

\label{sec:General-measurement-optimization}

Let us begin by considering a general measurement optimization problem,
in which we optimize a linear objective function over the intersection
of the set of measurement operators with an affine space. This problem
is thus a particular kind of semidefinite optimization problem. However,
as we show later in Section~\ref{subsec:Solving-general-SDPs}, this
problem can be suitably modified to capture an arbitrary semidefinite
optimization problem in its standard form \cite[Eq.~(4.51)]{Boyd2004}.
Additionally, optimization problems relevant to quantum hypothesis
testing are special cases of this problem, as discussed in Section
\ref{sec:Application-q-hypo-test}.

\subsection{General measurement optimization problem}

\label{subsec:General-measurement-optimization}Let $c,d\in\mathbb{N}$,
let
\begin{equation}
\mathcal{Q}\equiv\left(H,Q_{1},\ldots,Q_{c}\right)\label{eq:vector-Herm-ops}
\end{equation}
be a tuple of $d\times d$ Hermitian matrices, and let
\begin{equation}
q\equiv\left(q_{1},\ldots,q_{c}\right)\in\mathbb{R}^{c}.\label{eq:constraint-vector}
\end{equation}
We can think of $H$ as a Hamiltonian and each $Q_{i}$ as a non-commuting
charge, such as electric charge, particle number, angular momentum,
etc.
\begin{defn}
The corresponding measurement optimization problem is as follows:
\begin{equation}
E(\mathcal{Q},q)\coloneqq\min_{M\in\mathcal{M}_{d}}\left\{ \Tr\!\left[HM\right]:\Tr\!\left[Q_{i}M\right]=q_{i}\,\forall i\in\left[c\right]\right\} ,\label{eq:measurement-opt-gen}
\end{equation}
where $\mathcal{M}_{d}$ is defined in~\eqref{eq:meas-ops-def}.
\end{defn}

\begin{assumption}
\label{assu:dual-existence}We assume throughout our paper that there
exists at least one measurement operator $M\in\mathcal{M}_{d}$ such
that $\Tr\!\left[Q_{i}M\right]=q_{i}$ for all $i\in\left[c\right]$
(i.e., all constraints are met). If this assumption does not hold,
then the optimal value $E(\mathcal{Q},q)$ in~\eqref{eq:measurement-opt-gen}
is trivially equal to $+\infty$. 
\end{assumption}

\begin{prop}
\label{prop:dual-gen-meas-opt}If Assumption~\ref{assu:dual-existence}
holds, then the optimal value $E(\mathcal{Q},q)$ in~\eqref{eq:measurement-opt-gen}
can be expressed in terms of the following dual optimization problem:
\begin{equation}
E(\mathcal{Q},q)=\sup_{\mu\in\mathbb{R}^{c}}f(\mu),\label{eq:dual-gen-meas-opt}
\end{equation}
where the dual objective function $f(\mu)$ is defined as
\begin{align}
f(\mu) & \coloneqq\mu\cdot q-\Tr\!\left[\left(H-\mu\cdot Q\right)_{-}\right],\label{eq:dual-unreg-obj-func}
\end{align}
and we have employed the shorthands:
\begin{equation}
\mu\cdot q\equiv\sum_{i\in\left[c\right]}\mu_{i}q_{i},\qquad\mu\cdot Q\equiv\sum_{i\in\left[c\right]}\mu_{i}Q_{i}.
\end{equation}
Alternatively, we can write the dual optimization problem in~\eqref{eq:dual-gen-meas-opt}
as the following semidefinite program:
\begin{equation}
E(\mathcal{Q},q)=\sup_{\mu\in\mathbb{R}^{c},X\geq0}\left\{ \begin{array}{c}
\mu\cdot q+\Tr\!\left[H-\mu\cdot Q\right]-\Tr[X]:\\
X\geq H-\mu\cdot Q
\end{array}\right\} .\label{eq:dual-gen-meas-opt-SDP}
\end{equation}
Furthermore, the dual objective function $f(\mu)$ is concave in $\mu$.
\end{prop}

\begin{proof}
See Appendix~\ref{sec:Proof-of-dual-gen-meas-opt-exp} for a proof
of~\eqref{eq:dual-gen-meas-opt},~\eqref{eq:dual-gen-meas-opt-SDP},
and the fact that $f(\mu)$ is concave in $\mu$.
\end{proof}
The proof of~\eqref{eq:dual-gen-meas-opt} also demonstrates that
a measurement operator $M$ is optimal for~\eqref{eq:measurement-opt-gen}
if and only if it is of the form 
\begin{equation}
\Pi_{H<\mu\cdot Q}+M_{0}(\mu),\label{eq:optimal-meas-op}
\end{equation}
for some $\mu\in\mathbb{R}^{c}$ and $M_{0}(\mu)$, where $\Pi_{H<\mu\cdot Q}$
is the projection onto the strictly negative eigenspace of the operator
$H-\mu\cdot Q$ and $M_{0}(\mu)$ is a measurement operator satisfying
\begin{equation}
0\leq M_{0}(\mu)\leq\Pi_{H=\mu\cdot Q},
\end{equation}
where $\Pi_{H=\mu\cdot Q}$ is the projection onto the zero eigenspace
of $H-\mu\cdot Q$. The projection $\Pi_{H<\mu\cdot Q}$ represents
a sharp threshold, with eigenvalues equal to one for eigenprojections
of $H-\mu\cdot Q$ in its negative eigenspace and eigenvalues equal
to zero for eigenprojections of $H-\mu\cdot Q$ in its positive eigenspace.

In the forthcoming subsections, we introduce and motivate the notion
of a Fermi--Dirac thermal measurement, which smooths out this sharp
transition, in the same way that a positive-temperature Fermi--Dirac
distribution smooths out the sharp threshold function corresponding
to the zero-temperature Fermi-Dirac distribution.

\subsection{Interpreting eigenmodes of measurement operators as independent fermions}

\label{subsec:Interpretation-fermionic}Consider a particular measurement
operator $M$ with a spectral decomposition as follows:
\begin{equation}
M=\sum_{j=1}^{d}m_{j}|\phi_{j}\rangle\!\langle\phi_{j}|.
\end{equation}
Let us refer to the eigenvalue-eigenprojector pair $\left(m_{j},|\phi_{j}\rangle\!\langle\phi_{j}|\right)$
as an eigenmode.

Due to the assumption that $M$ is a measurement operator, the constraint
$0\leq m_{j}\leq1$ holds for all $j\in\left[d\right]$. As such,
we can think of each of the $d$ eigenmodes as an independent fermion, with $0\leq m_{j}\leq1$
corresponding to an occupation number constraint, in analogy with
the Pauli exclusion principle. Indeed, $m_{j}$ is the probability
that the $j$th fermion is present in the $j$th eigenmode, and $1-m_{j}$
is the probability that it is not present.

Consider furthermore that
\begin{align}
\Tr\!\left[HM\right] & =\sum_{j=1}^{d}m_{j}\langle\phi_{j}|H|\phi_{j}\rangle.\label{eq:total-expected-energy-fermions}
\end{align}
Here, we can think of $\langle\phi_{j}|H|\phi_{j}\rangle$ as the
energy of the $j$th fermion when occupied, so that the expected energy
of the $j$th fermion is $m_{j}\langle\phi_{j}|H|\phi_{j}\rangle$.
Thus, the total expected energy of all $d$ fermions is $\Tr\!\left[HM\right]$.
Furthermore, for all $i\in\left[c\right]$, we similarly have that
\begin{equation}
\Tr\!\left[Q_{i}M\right]=\sum_{j=1}^{d}m_{j}\langle\phi_{j}|Q_{i}|\phi_{j}\rangle.
\end{equation}
Here, each $\langle\phi_{j}|Q_{i}|\phi_{j}\rangle$ is the value of
the $i$th charge for the $j$th fermion, so that the expected value
of the $i$th charge for the $j$th fermion is $m_{j}\langle\phi_{j}|Q_{i}|\phi_{j}\rangle$,
and the total expected charge of all $d$ fermions is $\Tr\!\left[Q_{i}M\right]$. 

As such, with this perspective, the goal of the measurement optimization
problem in~\eqref{eq:measurement-opt-gen} is to minimize the total
expected energy of $d$ fermions subject to a constraint on the total
expected value of the $i$th charge, for each $i\in\left[c\right]$.
The freedom in the optimization problem is to adjust the eigenmodes
as desired in order for the total expected energy to be as small as
possible, while satisfying the constraints.

\subsection{Free energy minimization}

\label{subsec:Free-energy-minimization}

Following the approach from
\cite{liu2025qthermoSDPs}, we are motivated by the fact that real
physical systems operate at a strictly positive temperature $T>0$,
in which case the goal shifts to minimize the constrained free energy
rather than the constrained energy. Recall that the free energy is
equal to the expected value of the energy less the entropy scaled
by the temperature. For the $j$th eigenmode, the free energy is thus
\begin{equation}
m_{j}\langle\phi_{j}|H|\phi_{j}\rangle-T\cdot s(m_{j}),
\end{equation}
where the binary entropy $s(p)$ is defined for $p\in\left[0,1\right]$
as
\begin{equation}
s(p)\coloneqq-p\ln p-\left(1-p\right)\ln\!\left(1-p\right).\label{eq:binary-entropy-def}
\end{equation}
The total free energy is thus given by
\begin{equation}
\sum_{j=1}^{d}\left[m_{j}\langle\phi_{j}|H|\phi_{j}\rangle-Ts(m_{j})\right]=\Tr\!\left[HM\right]-T\cdot S_{\FD}(M),\label{eq:rewrite-free-energy-fermions}
\end{equation}
where the Fermi--Dirac entropy of the measurement operator $M$ is
defined as
\begin{equation}
S_{\FD}(M)\coloneqq S(M)+S(I-M),\label{eq:FD-entropy-def}
\end{equation}
with the von Neumann entropy defined as
\begin{equation}
S(M)\coloneqq-\Tr\!\left[M\ln M\right].
\end{equation}
The equality in~\eqref{eq:rewrite-free-energy-fermions} follows from
\eqref{eq:total-expected-energy-fermions} and the fact that $S_{\FD}(M)$
is unitarily invariant, depending only on the eigenvalues of $M$,
so that
\begin{equation}
S_{\FD}(M)=\sum_{j=1}^{d}s(m_{i}).\label{eq:FD-entropy-sum-binary-entropies}
\end{equation}
Observe that the function $S_{\FD}(M)$ is concave in $M$. Indeed,
suppose that $M_{0}$ and $M_{1}$ are measurement operators, suppose
that $\lambda\in\left[0,1\right]$, and define $M_{\lambda}\coloneqq\lambda M_{1}+\left(1-\lambda\right)M_{0}$.
Then
\begin{align}
S_{\FD}(M_{\lambda}) & =S(M_{\lambda})+S(I-M_{\lambda})\label{eq:concavity-FD-entr-1}\\
 & =S(\lambda M_{1}+\left(1-\lambda\right)M_{0})\nonumber \\
 & \qquad+S(\lambda(I-M_{1})+\left(1-\lambda\right)(I-M_{0}))\\
 & \leq\lambda S(M_{1})+\left(1-\lambda\right)S(M_{0})\nonumber \\
 & \qquad+\lambda S(I-M_{1})+\left(1-\lambda\right)S(I-M_{0})\\
 & =\lambda S_{\FD}(M_{1})+\left(1-\lambda\right)S_{\FD}(M_{0}),\label{eq:concavity-FD-entr-last}
\end{align}
where the sole inequality follows from the concavity of the von Neumann
entropy (see, e.g., \cite[Eq.~(7.2.106)]{khatri2024}).

We can then modify the measurement optimization problem in~\eqref{eq:measurement-opt-gen}
to become the following measurement free-energy optimization problem:
\begin{defn}
For $T>0$, and $\mathcal{Q}$ and $q$ as given in~\eqref{eq:vector-Herm-ops}
and~\eqref{eq:constraint-vector}, respectively, the measurement free-energy
optimization problem is as follows:
\begin{multline}
F_{T}(\mathcal{Q},q)\coloneqq\min_{M\in\mathcal{M}_{d}}\left\{ \begin{array}{c}
\Tr\!\left[HM\right]-T\cdot S_{\FD}(M):\\
\Tr\!\left[Q_{i}M\right]=q_{i}\,\forall i\in\left[c\right]
\end{array}\right\} .\label{eq:primal-FD-obj}
\end{multline}
\end{defn}

By solving for the dual in this case, we arrive at the following:
\begin{prop}
\label{prop:dual-FD-free-energy}Fix $T>0$. Given $\mathcal{Q}$
and $q$ as defined in~\eqref{eq:vector-Herm-ops} and~\eqref{eq:constraint-vector},
respectively, and under Assumption~\ref{assu:dual-existence}, the
following equality holds:
\begin{equation}
F_{T}(\mathcal{Q},q)=\sup_{\mu\in\mathbb{R}^{c}}f_{T}(\mu),\label{eq:dual-free-energy-meas-opt}
\end{equation}
where the temperature-$T$ dual objective function $f_{T}(\mu)$ is
defined as
\begin{equation}
f_{T}(\mu)\coloneqq\mu\cdot q-T\Tr\!\left[\ln\!\left(e^{-\frac{1}{T}\left(H-\mu\cdot Q\right)}+I\right)\right].\label{eq:dual-FD-obj}
\end{equation}
Furthermore, an optimal measurement operator $M_{T}(\mu)$ for~\eqref{eq:primal-FD-obj}
is a Fermi--Dirac measurement operator of the following form:
\begin{equation}
M_{T}(\mu)\coloneqq\left(e^{\frac{1}{T}\left(H-\mu\cdot Q\right)}+I\right)^{-1}.\label{eq:FD-meas-op}
\end{equation}
\end{prop}

\begin{proof}
We provide two different proofs. See Appendix~\ref{sec:Proof-of-Theorem-FD-thermal-meas-gen}
for a first proof that uses concavity of the Fermi--Dirac entropy
and matrix derivatives. See Appendix~\ref{sec:Second-Proof-dual-FD-meas-op}
for a second proof that introduces the Fermi--Dirac relative entropy
and uses the fact that it is non-negative for all measurement operators
and equal to zero if and only if its arguments are equal.
\end{proof}
\begin{rem}
As a consequence of the following scalar limit holding for all $x\in\mathbb{R}$,
\begin{equation}
\lim_{T\to0^{+}}-T\ln\!\left(e^{-\frac{x}{T}}+1\right)=\min\!\left\{ x,0\right\} ,
\end{equation}
we conclude that, for all $\mu\in\mathbb{R}^{c}$, the temperature-$T$
dual objective function $f_{T}(\mu)$ in~\eqref{eq:dual-FD-obj} converges
to the objective function $f(\mu)$ in~\eqref{eq:dual-unreg-obj-func},
in the zero-temperature limit:
\begin{equation}
\lim_{T\to0^{+}}f_{T}(\mu)=f(\mu).\label{eq:temp-T-obj-to-zero-temp-obj}
\end{equation}
\end{rem}

\subsection{Fermi--Dirac thermal measurements}

\label{subsec:Fermi-Dirac-thermal-measurements}

Observe that the measurement operator $M_{T}(\mu)$ in~\eqref{eq:FD-meas-op}
has the form of a Fermi--Dirac distribution (also known as a sigmoid
or logistic function), and the following quantity in~\eqref{eq:dual-FD-obj}
\begin{equation}
-T\Tr\!\left[\ln\!\left(e^{-\frac{1}{T}\left(H-\mu\cdot Q\right)}+I\right)\right]
\end{equation}
is equal to the fermionic free energy. Additionally, the measurement
operator $I-M_{T}(\mu)$ is also a Fermi--Dirac measuerement operator,
given by
\begin{equation}
I-M_{T}(\mu)=\left(e^{-\frac{1}{T}\left(H-\mu\cdot Q\right)}+I\right)^{-1},\label{eq:other-outcome-FD}
\end{equation}
as a consequence of the scalar identity
\begin{equation}
1-\frac{1}{e^{x}+1}=\frac{e^{x}}{e^{x}+1}=\frac{1}{e^{-x}+1}.\label{eq:scalar-ident-FD-meas-compl}
\end{equation}
Thus, we refer to the measurement $\left(M_{T}(\mu),I-M_{T}(\mu)\right)$
as a Fermi--Dirac thermal measurement and define it formally as follows:
\begin{defn}[Fermi--Dirac thermal measurement]
Given $d\in\mathbb{N}$, a $d\times d$ Hermitian matrix $A$, and
a temperature $T>0$, a Fermi--Dirac thermal measurement is a binary
positive operator-valued measure $\left(M_{T}(A),I-M_{T}(A)\right)$,
where
\begin{align}
M_{T}(A) & \coloneqq\left(e^{\frac{A}{T}}+I\right)^{-1},\label{eq:gen-FD-meas-op}\\
I-M_{T}(A) & =\left(e^{-\frac{A}{T}}+I\right)^{-1}.\label{eq:gen-FD-meas-op-2}
\end{align}
\end{defn}

For fixed $\mu\in\mathbb{R}^{c}$, the following equalities hold in
the zero-temperature limit:
\begin{align}
\lim_{T\to0^{+}}M_{T}(\mu) & =\Pi_{H<\mu\cdot Q}+\frac{1}{2}\Pi_{H=\mu\cdot Q},\label{eq:temp-T-meas-op-to-zero-temp-meas-op}\\
\lim_{T\to0^{+}}\left[I-M_{T}(\mu)\right] & =\Pi_{H>\mu\cdot Q}+\frac{1}{2}\Pi_{H=\mu\cdot Q},
\end{align}
which are a consequence of the following scalar zero-temperature limits:
\begin{align}
\lim_{T\to0^{+}}\frac{1}{e^{\frac{x}{T}}+1} & =u(-x)\coloneqq\begin{cases}
1 & :x<0\\
\frac{1}{2} & :x=0\\
0 & :x>0
\end{cases},\\
\lim_{T\to0^{+}}\frac{1}{e^{-\frac{x}{T}}+1} & =u(x)\coloneqq\begin{cases}
0 & :x<0\\
\frac{1}{2} & :x=0\\
1 & :x>0
\end{cases}.
\end{align}
Thus, the zero-temperature limit recovers a threshold measurement
of the form in~\eqref{eq:optimal-meas-op}.

For fixed $\mu\in\mathbb{R}^{c}$, the following equalities hold in
the infinite-temperature limit:
\begin{equation}
\lim_{T\to\infty}M_{T}(\mu)=\lim_{T\to\infty}\left[I-M_{T}(\mu)\right]=\frac{I}{2},
\end{equation}
as a consequence of the scalar limits
\begin{equation}
\lim_{T\to\infty}\frac{1}{e^{\frac{x}{T}}+1}=\lim_{T\to\infty}\frac{1}{e^{-\frac{x}{T}}+1}=\frac{1}{2}.
\end{equation}
Thus, the infinite-temperature limit recovers an uninformative measurement.

\subsection{Approximation error}

\label{subsec:Approximation-error}

In this section, we establish several different bounds for the approximation
error between the optimal value $E(\mathcal{Q},q)$ of the original
semidefinite program (SDP) in~\eqref{eq:measurement-opt-gen} and
the optimal value $F_{T}(\mathcal{Q},q)$ of the free energy optimization
in~\eqref{eq:primal-FD-obj}. The first bounds that we present are
simple but scale linearly with dimension (Section~\ref{subsec:Simple-uniform-approximation}),
while the second bounds presented are more nuanced and are related
to the spectral gap of Hamiltonians of the form $H-\mu\cdot Q$ (Section
\ref{subsec:Spectral-gap-based-approximation}).

\subsubsection{Simple uniform approximation error bounds}

\label{subsec:Simple-uniform-approximation}

We begin with the simpler
bounds for the approximation error.
\begin{prop}
\label{prop:simple-approx-bnd}Fix $\varepsilon,T>0$, and recall
$d\in\mathbb{N}$ from~\eqref{eq:vector-Herm-ops}. If $T\leq\frac{\varepsilon}{d\ln2}$,
then
\begin{equation}
E(\mathcal{Q},q)\geq F_{T}(\mathcal{Q},q)\geq E(\mathcal{Q},q)-\varepsilon.\label{eq:simple-approx-bnd}
\end{equation}
\end{prop}

\begin{proof}
See Appendix~\ref{sec:proof-simple-approx-bnd}. The proof uses the
definitions of $E(\mathcal{Q},q)$ and $F_{T}(\mathcal{Q},q)$ and
the upper bound $S_{\FD}(M)\leq d\ln2$, holding for all $M\in\mathcal{M}_{d}$.
\end{proof}
We now argue how well the following quantity
\begin{equation}
\tilde{f}_{T}(\mu)\coloneqq\mu\cdot q+\Tr\!\left[\left(H-\mu\cdot Q\right)M_{T}(\mu)\right].\label{eq:approx-without-FD-entropy}
\end{equation}
can approximate the optimal value $E(\mathcal{Q},q)$, for a particular
choice of $\mu$. This is useful for hybrid quantum--classical algorithms
that we detail later in Section~\ref{sec:Quantum-algorithm-gen-meas-opt}.
Suppose that $\mu_{T}^{\star}\in\mathbb{R}^{c}$ is an optimal choice
for $F_{T}(\mathcal{Q},q)$, so that
\begin{equation}
F_{T}(\mathcal{Q},q)=f_{T}(\mu_{T}^{\star}).\label{eq:free-energy-opt-dual}
\end{equation}
As a consequence of Lemma~\ref{lem:rewrite-free-energy-rel-ent},
the following equality holds for all $\mu\in\mathbb{R}^{c}$:
\begin{equation}
f_{T}(\mu)=\mu\cdot q+\Tr\!\left[\left(H-\mu\cdot Q\right)M_{T}(\mu)\right]-T\cdot S_{\FD}(M_{T}(\mu)).
\end{equation}
In applications, we may alternatively wish to output the value $\tilde{f}_{T}(\mu_{T}^{\star})$
as an approximation of $f_{T}(\mu_{T}^{\star})$. That is, the expression
in~\eqref{eq:approx-without-FD-entropy} omits the term $-T\cdot S_{\FD}(M_{T}(\mu))$.
\begin{prop}
\label{prop:no-FD-ent-simple-approx-err-bnd}Fix $\varepsilon,T>0$,
and recall $d\in\mathbb{N}$ from~\eqref{eq:vector-Herm-ops}. Let
$\mu_{T}^{\star}\in\mathbb{R}^{c}$ be a parameter vector that is
optimal for $F_{T}(\mathcal{Q},q)$, defined in~\eqref{eq:primal-FD-obj}.
If $T\leq\frac{\varepsilon}{d\ln2}$, then
\begin{equation}
\left|\tilde{f}_{T}(\mu)-E(\mathcal{Q},q)\right|\leq\varepsilon.\label{eq:simple-approx-bnd-no-FD-ent}
\end{equation}
\end{prop}

\begin{proof}
See Appendix~\ref{sec:no-FD-ent-simple-approx-err-bnd}.
\end{proof}

\subsubsection{Spectral-gap based approximation error bounds}

\label{subsec:Spectral-gap-based-approximation}We now present the
more nuanced appromixation error bounds, based on spectral gap and
groundspace degeneracy.
\begin{prop}
\label{prop:approx-error-spectral-gap}Let $T>0$. Let $\mu^{\star},\mu_{T}^{\star}\in\mathbb{R}^{c}$
be parameter vectors that are optimal for $E(\mathcal{Q},q)$ and
$F_{T}(\mathcal{Q},q)$ defined in~\eqref{eq:measurement-opt-gen}
and~\eqref{eq:primal-FD-obj}, respectively. Suppose that $\mu^{\star}$
and $\mu_{T}^{\star}$ are in a parameter set $\mathbb{M}$. Suppose
that the following bound holds:
\begin{equation}
\sup_{\mu\in\mathbb{M}}\dim\ker(H-\mu\cdot Q)\leq d_{0},\label{eq:def-d0}
\end{equation}
where $d_{0}\in\mathbb{N}$. Suppose that there exists $\Delta>0$
such that
\begin{equation}
\inf_{\mu\in\mathbb{M}}\min_{i:\lambda_{i}\neq0}\left|\lambda_{i}(H-\mu\cdot Q)\right|\geq\Delta,\label{eq:def-Delta}
\end{equation}
where $\lambda_{i}(H-\mu\cdot Q)$ denotes the $i$th eigenvalue of
$H-\mu\cdot Q$. Then the following bound holds:
\begin{equation}
F_{T}(\mathcal{Q},q)\geq E(\mathcal{Q},q)-T\left(d_{0}\ln2+\left(d-d_{0}\right)e^{-\frac{\Delta}{T}}\right).\label{eq:temp-depend-bnd-d0-spec-gap}
\end{equation}
Thus, for all $\varepsilon>0$, if the temperature $T$ satisfies
the following upper bound
\begin{equation}
T\leq\min\left\{ \frac{\varepsilon}{2d_{0}\ln2},\frac{\Delta}{\ln\!\left(\frac{1}{\ln2}\left(\frac{d-d_{0}}{d_{0}}\right)\right)}\right\} ,\label{eq:temp-eps-error-spec-gap}
\end{equation}
then 
\begin{equation}
E(\mathcal{Q},q)\geq F_{T}(\mathcal{Q},q)\geq E(\mathcal{Q},q)-\varepsilon.
\end{equation}
\end{prop}

\begin{proof}
See Appendix~\ref{sec:proof-approx-error-spectral-gap}.
\end{proof}
The bound in Proposition~\ref{prop:approx-error-spectral-gap} implies
that, in the case that the spectral gap $\Delta$ decays inverse polynomially
with $\ln d$ and if $d_{0}$ is polynomial in $\ln d$, then the
dependence of the temperature $T$ on the underlying dimension is
inverse polynomial in $\ln d$. If this is the case for a given choice
of $\mathcal{Q}$ in~\eqref{eq:vector-Herm-ops}, then this temperature
dependence is much more favorable when compared to that given by the
uniform bound in Proposition~\ref{prop:simple-approx-bnd}. Later
on, we see that the temperature dependence plays a direct role in
the runtime of hybrid quantum--classical algorithms for solving the
general measurement optimization problem in~\eqref{eq:measurement-opt-gen}.
For a choice of $\mathcal{Q}$ in~\eqref{eq:vector-Herm-ops} such
that the assumptions in Proposition~\ref{prop:approx-error-spectral-gap}
hold with $\Delta^{-1}$ and $d_{0}$ polynomial in $\ln d$, the
resulting hybrid quantum--classical are efficient.

We are again interested in how well the quantity $\tilde{f}_{T}(\mu)$
in~\eqref{eq:approx-without-FD-entropy} can approximate the optimal
value $E(\mathcal{Q},q)$, for a particular choice of $\mu$. As an
intermediate step, the following lemma provides an upper bound on
the quantity $T\cdot S_{\FD}(M_{T}(A))$, which is useful for this
purpose:
\begin{prop}
\label{prop:FD-entropy-up-bnd}Let $d\in\mathbb{N}$, $T>0$, and
let $A$ be a $d\times d$ Hermitian matrix such that
\begin{align}
d_{0} & \coloneqq\dim\ker(A),\label{eq:kernel-dim}\\
\Delta & \coloneqq\min_{i:\lambda_{i}\neq0}\left|\lambda_{i}(A)\right|,\label{eq:spectral-gap-def-FD-ent-bnd}
\end{align}
where $\lambda_{i}(A)$ denotes the $i$th eigenvalue of $A$. Then
\begin{equation}
T\cdot S_{\FD}(M_{T}(A))\leq T\left[d_{0}\ln2+\left(d-d_{0}\right)s\!\left(\frac{1}{e^{\Delta/T}+1}\right)\right],\label{eq:FD-entropy-up-bnd-1}
\end{equation}
where $M_{T}(A)$ is defined in~\eqref{eq:gen-FD-meas-op} and the
binary entropy $s(\cdot)$ is defined in~\eqref{eq:binary-entropy-def}.
Thus, for all $\varepsilon>0$, if
\begin{equation}
T\leq\min\left\{ T_{0},\frac{\varepsilon}{d_{0}2\ln2},\frac{\Delta}{\ln\!\left(\frac{2\left(d-d_{0}\right)\left(2T_{0}+\Delta\right)}{\varepsilon}\right)}\right\} ,\label{eq:temp-depend-eps-err-FD-ent}
\end{equation}
where $T_{0}>0$, then
\begin{equation}
T\cdot S_{\FD}(M_{T}(A))\leq\varepsilon.\label{eq:FD-entropy-up-bnd-2}
\end{equation}
\end{prop}

\begin{proof}
See Appendix~\ref{sec:FD-entropy-up-bnd-proof}.
\end{proof}
We can now combine the bounds from Propositions~\ref{prop:approx-error-spectral-gap}
and~\ref{prop:FD-entropy-up-bnd} to determine how well the value
$\tilde{f}_{T}(\mu_{T}^{\star})$ approximates $E(\mathcal{Q},q)$:
\begin{prop}
\label{prop:approx-err-no-FD-obj}Let $T>0$. Let $\mu^{\star},\mu_{T}^{\star}\in\mathbb{R}^{c}$
be parameter vectors that are optimal for $E(\mathcal{Q},q)$ and
$F_{T}(\mathcal{Q},q)$ defined in~\eqref{eq:measurement-opt-gen}
and~\eqref{eq:primal-FD-obj}, respectively. Suppose that $\mu^{\star}$
and $\mu_{T}^{\star}$ are in a parameter set $\mathbb{M}$. Suppose
that the following bound holds:
\begin{equation}
\sup_{\mu\in\mathbb{M}}\dim\ker(H-\mu\cdot Q)\leq d_{0},\label{eq:def-d0-1}
\end{equation}
where $d_{0}\in\mathbb{N}$. Suppose that there exists $\Delta>0$
such that
\begin{equation}
\inf_{\mu\in\mathbb{M}}\min_{i:\lambda_{i}\neq0}\left|\lambda_{i}(H-\mu\cdot Q)\right|\geq\Delta,\label{eq:def-Delta-1}
\end{equation}
where $\lambda_{i}(H-\mu\cdot Q)$ denotes the $i$th eigenvalue of
$H-\mu\cdot Q$. Then the following bounds hold:
\begin{multline}
E(\mathcal{Q},q)+T\left[d_{0}\ln2+\left(d-d_{0}\right)s\!\left(\frac{1}{e^{\Delta/T}+1}\right)\right]\\
\geq\tilde{f}_{T}(\mu_{T}^{\star})\geq E(\mathcal{Q},q)-T\left(d_{0}\ln2+\left(d-d_{0}\right)e^{-\frac{\Delta}{T}}\right).\label{eq:err-bnd-obj-no-FD-ent-tem-depend}
\end{multline}
Thus, for all $\varepsilon>0$, if
\begin{multline}
T\leq\\
\min\left\{ T_{0},\frac{\varepsilon}{d_{0}2\ln2},\frac{\Delta}{\ln\!\left(\frac{2\left(d-d_{0}\right)\left(2T_{0}+\Delta\right)}{\varepsilon}\right)},\frac{\Delta}{\ln\!\left(\frac{d-d_{0}}{d_{0}\ln2}\right)}\right\} ,\label{eq:temp-threshold-obj-no-FD-ent}
\end{multline}
where $T_{0}>0$, then
\begin{equation}
\left|\tilde{f}_{T}(\mu_{T}^{\star})-E(\mathcal{Q},q)\right|\leq\varepsilon.\label{eq:err-bnd-obj-no-FD-eps}
\end{equation}
\end{prop}

\begin{proof}
See Appendix~\ref{sec:approx-err-no-FD-obj}.
\end{proof}

\subsection{Gradient and Hessian of dual objective function}

\label{subsec:Gradient-and-Hessian}

In this section, we derive analytical expressions for the gradient
and the Hessian of the objective function $f_{T}(\mu)$ in~\eqref{eq:dual-FD-obj}.
These expressions are essential in developing gradient-ascent algorithms
for optimizing $f_{T}(\mu)$ and analyzing the performance of such
algorithms. We also prove that the Hessian of $f_{T}(\mu)$ is negative
semidefinite and has bounded spectral norm, thus guaranteeing that
gradient ascent is guaranteed to converge. The expressions for the
Hessian matrix elements can be viewed as being analogous to those
reported in~\cite{liu2025qthermoSDPs} and in~\cite{Patel2025a} for
the Kubo--Mori information matrix of parameterized thermal states.
Indeed, as we show in Appendix~\ref{sec:Continuity-bounds-FD-rel-ent},
the Hessian matrix elements are proportional to those for the information
matrix resulting from the Fermi--Dirac relative entropy. 
\begin{prop}
\label{prop:gradient-FD-gen-obj-func}For $T>0$, the $i$th partial
derivative of the objective function $f_{T}(\mu)$ in~\eqref{eq:dual-FD-obj}
is given by
\begin{equation}
\frac{\partial}{\partial\mu_{i}}f_{T}(\mu)=q_{i}-\Tr\!\left[M_{T}(\mu)Q_{i}\right],
\end{equation}
where the measurement operator $M_{T}(\mu)$ is defined in~\eqref{eq:FD-meas-op}.
\end{prop}

\begin{proof}
See Appendix~\ref{sec:Proof-of-gradient-FD-gen}.
\end{proof}
\begin{prop}
\label{prop:hessian-FD-gen-obj-func}For all $c\in\mathbb{N}$, $T>0$,
and $i,j\in\left[c\right]$, the elements of the Hessian matrix $\nabla^{2}f_{T}(\mu)$
for the objective function $f_{T}(\mu)$ in~\eqref{eq:dual-FD-obj}
are given by
\begin{align}
 & \frac{\partial^{2}}{\partial\mu_{i}\partial\mu_{j}}f_{T}(\mu)\nonumber \\
 & =-\frac{1}{T}\int_{0}^{1}ds\,\Tr\!\left[M_{T}(\mu,s)Q_{i}M_{T}(\mu,1-s)Q_{j}\right],\label{eq:1st-hessian-exp}\\
 & =-\frac{1}{T}\Re\!\left[\Tr\!\left[M_{T}(\mu)\Phi_{\mu}(Q_{i})\left(I-M_{T}(\mu)\right)Q_{j}\right]\right],\label{eq:2nd-hessian-exp}
\end{align}
where the measurement operator $M_{T}(\mu,s)$ is defined for all
$s\in\left[0,1\right]$ as
\begin{equation}
M_{T}(\mu,s)\coloneqq\frac{e^{\frac{s}{T}\left(H-\mu\cdot Q\right)}}{e^{\frac{1}{T}\left(H-\mu\cdot Q\right)}+I},\label{eq:meas-op-t-depend}
\end{equation}
and the quantum channel $\Phi_{\mu}(X)$ and the high-peak tent probability
density $\gamma(t)$ are respectively defined as
\begin{align}
\Phi_{\mu}(X) & \coloneqq\int_{-\infty}^{\infty}dt\,\gamma(t)e^{-i\left(H-\mu\cdot Q\right)t/T}Xe^{i\left(H-\mu\cdot Q\right)t/T},\label{eq:Phi-channel}\\
\gamma(t) & \coloneqq\frac{2}{\pi}\ln\left|\coth\!\left(\frac{\pi t}{2}\right)\right|.\label{eq:high-peak-tent-prob-dens}
\end{align}
\end{prop}

\begin{proof}
See Appendix~\ref{sec:Proof-of-Hessian-FD-gen-obj-func}.
\end{proof}
\begin{prop}
\label{prop:hessian-NSD-spec-up-bnd}For all $T>0$, the Hessian matrix
$\nabla^{2}f_{T}(\mu)$ is negative semidefinite; i.e., for all $v\in\mathbb{R}^{c}$,
\begin{equation}
v^{T}\nabla^{2}f_{T}(\mu)v\leq0,
\end{equation}
and its spectral norm is bounded for all $\mu\in\mathbb{R}^{c}$ as
follows:
\begin{equation}
\left\Vert \nabla^{2}f_{T}(\mu)\right\Vert \leq L_{T}\coloneqq\frac{1}{T}\sum_{i\in\left[c\right]}\left\Vert Q_{i}\right\Vert _{1}\left\Vert Q_{i}\right\Vert .\label{eq:smoothness-parameter}
\end{equation}
\end{prop}

\begin{proof}
See Appendix~\ref{sec:hessian-NSD-spec-up-bnd}.
\end{proof}
\begin{rem}
\label{rem:dual-concave-smoothness-param}As a consequence of Proposition
\ref{prop:hessian-NSD-spec-up-bnd}, the function $f_{T}(\mu)$ is
concave in $\mu$ with smoothness parameter $L_{T}$, as defined in
\eqref{eq:smoothness-parameter}.
\end{rem}

\subsection{Gradient ascent for optimizing Fermi--Dirac thermal measurements}

\label{subsec:Gradient-ascent-gen-meas-opt}

An optimal $\mu_{T}^{\star}$ for~\eqref{eq:dual-free-energy-meas-opt}
depends on the constraint vector $q$ in~\eqref{eq:constraint-vector},
but it does not have a closed form in general. However, as observed
in Remark~\ref{rem:dual-concave-smoothness-param}, the function $f_{T}(\mu)$
is concave in $\mu$. As such, one can search for a value of $\mu$
that maximizes~\eqref{eq:dual-free-energy-meas-opt} by gradient ascent,
which is guaranteed to converge to a point $\delta$-close to the
global maximum of $f_{T}(\mu)$ in $O(1/\delta)$ steps. For the step
size $\eta$ of gradient ascent, it suffices to set $\eta\in\left(0,1/L_{T}\right]$,
where $L_{T}$ is defined in~\eqref{eq:smoothness-parameter}.

This leads to the following algorithm for performing the optimization
in~\eqref{eq:dual-free-energy-meas-opt}:
\begin{lyxalgorithm}
\label{alg:classical-gradient-ascent}A gradient-ascent algorithm
for measurement free-energy optimization consists of the following
steps:
\begin{enumerate}
\item Set $\delta>0$ to be the desired error, set $T\leftarrow\frac{\delta}{4d\ln2}$
to be the temperature, initialize $\mu^{0}\leftarrow\left(0,\ldots,0\right)$,
fix the learning rate $\eta\in\left(0,1/L_{T}\right]$, where $L_{T}$
is given in~\eqref{eq:smoothness-parameter}, and set the number of
steps, $J$, to satisfy $J\geq\frac{L_{T}\left\Vert \mu_{T}^{\star}\right\Vert ^{2}}{\delta}$,
where $\mu_{T}^{\star}$ is an optimal solution to~\eqref{eq:dual-free-energy-meas-opt}.
\item For all $j\in\left[J\right]$ and $i\in\left[c\right]$, set
\begin{equation}
\mu_{i}^{j}\leftarrow\mu_{i}^{j-1}+\eta\left(q_{i}-\Tr\!\left[M_{T}(\mu^{j-1})Q_{i}\right]\right).
\end{equation}
\item Output $\tilde{f}_{T}(\mu^{J})$, as defined in~\eqref{eq:approx-without-FD-entropy}.
\end{enumerate}
\end{lyxalgorithm}

To understand the performance of Algorithm~\ref{alg:classical-gradient-ascent},
consider that there are three sources of error: the error from approximating
$E(\mathcal{Q},q)$ with $F_{T}(\mathcal{Q},q)$, the error from gradient
ascent arriving at an approximate global minimum of $F_{T}(\mathcal{Q},q)$,
and the error from outputting $\tilde{f}_{T}(\mu^{J})$ instead of
$f_{T}(\mu^{J})$. Applying the triangle inequality, we find that
the total error of Algorithm~\ref{alg:classical-gradient-ascent}
is bounded as follows:
\begin{align}
 & \left|\tilde{f}_{T}(\mu^{J})-E(\mathcal{Q},q)\right|\nonumber \\
 & \overset{(a)}{\leq}\left|\tilde{f}_{T}(\mu^{J})-f_{T}(\mu^{J})\right|+\left|f_{T}(\mu^{J})-F_{T}(\mathcal{Q},q)\right|\nonumber \\
 & \qquad+\left|F_{T}(\mathcal{Q},q)-E(\mathcal{Q},q)\right|\\
 & \overset{(b)}{\leq}T\cdot S_{\FD}(M_{T}(\mu^{J}))+\frac{\delta}{2}+\frac{\delta}{4}\\
 & \leq\frac{\delta}{4}+\frac{\delta}{2}+\frac{\delta}{4}=\delta.
\end{align}
The inequality (a) follows from the triangle inequality. The inequality
(b) follows from the definitions of $\tilde{f}_{T}(\mu^{J})$ and
$f_{T}(\mu^{J})$ for bounding the first term, from \cite[Corollary~3.5]{Garrigos2024}
(regarding the convergence of gradient ascent) for bounding the second
term, and from Proposition~\ref{prop:simple-approx-bnd} for bounding
the third term.

Due to the choice of temperature $T$, the runtime of Algorithm~\ref{alg:classical-gradient-ascent}
is linear in the dimension $d$, or equivalently, exponential in $\ln d$,
with $\ln d$ corresponding to the number of qubits needed to specify
the matrices in $\mathcal{Q}$. Suppose instead that information about
the spectral gap $\Delta$ and ground space degeneracy $d_{0}$ of
$H-\mu\cdot Q$ is available for $\mu$ in a parameter set $\mathbb{M}$.
If $\Delta^{-1}$ and $d_{0}$ are polynomial in $\ln d$, then Proposition
\ref{prop:approx-err-no-FD-obj} can be invoked to conclude that the
number of steps needed for Algorithm~\ref{alg:classical-gradient-ascent}
(with the temperature $T$ adjusted appropriately) to converge is
polynomial in $\ln d$. Thus, if these conditions hold, then this
speedup in runtime can be significant, and this plays a role in determining
when we should expect our approach here to lead to efficient hybrid
quantum--classical optimization algorithms for solving the general
measurement optimization problem in~\eqref{eq:measurement-opt-gen}.

\subsection{General measurement optimization with equality and inequality constraints}

\label{subsec:General-measurement-opteq-ineq}

In this section, we briefly highlight a generalization of the measurement
optimization problem that involves inequality constraints in addition
to equality constraints. The main modification is that the inequality
constraints become associated with non-negative Lagrange multipliers.
This generalization to involve inequality constraints is needed to
capture composite quantum hypothesis testing problems of the form
considered in Section~\ref{subsec:Asymmetric-HT-composite-null}. 

Let $b,c,d\in\mathbb{N}$, let
\begin{equation}
\mathcal{Q}\equiv\left(H,Q_{1},\ldots,Q_{c}\right),\qquad\mathcal{P}\equiv\left(P_{1},\ldots,P_{b}\right),\label{eq:vector-Herm-ops-ineq}
\end{equation}
be tuples of $d\times d$ Hermitian matrices, and let
\begin{equation}
q\equiv\left(q_{1},\ldots,q_{c}\right)\in\mathbb{R}^{c},\qquad p\equiv\left(p_{1},\ldots,p_{b}\right)\in\mathbb{R}^{b},\label{eq:constraint-vector-ineq}
\end{equation}

The corresponding measurement optimization problem is as follows:
\begin{equation}
E(\mathcal{Q},\mathcal{P},q,p)\coloneqq\min_{M\in\mathcal{M}_{d}}\left\{ \begin{array}{c}
\Tr\!\left[HM\right]:\\
\Tr\!\left[Q_{i}M\right]=q_{i}\,\forall i\in\left[c\right],\\
\Tr\!\left[P_{j}M\right]\geq p_{j}\,\forall j\in\left[b\right]
\end{array}\right\} .\label{eq:measurement-opt-gen-ineq}
\end{equation}
The dual is given by
\begin{equation}
E(\mathcal{Q},\mathcal{P},q,p)=\sup_{\mu\in\mathbb{R}^{c},\nu\in\mathbb{R}_{+}^{b}}f(\mu,\nu),\label{eq:dual-gen-meas-opt-1}
\end{equation}
where the dual objective function $f(\mu)$ is defined as
\begin{align}
f(\mu,\nu) & \coloneqq\mu\cdot q+\nu\cdot p-\Tr\!\left[\left(H-\mu\cdot Q-\nu\cdot P\right)_{-}\right],\label{eq:dual-unreg-obj-func-1}
\end{align}
For $T>0$, the measurement free-energy minimization problem is given
by
\begin{equation}
F_{T}(\mathcal{Q},\mathcal{P},q,p)\coloneqq\min_{M\in\mathcal{M}_{d}}\left\{ \begin{array}{c}
\Tr\!\left[HM\right]-T\cdot S_{\FD}(M):\\
\Tr\!\left[Q_{i}M\right]=q_{i}\,\forall i\in\left[c\right],\\
\Tr\!\left[P_{j}M\right]\geq p_{j}\,\forall j\in\left[b\right]
\end{array}\right\} ,\label{eq:free-energy-eq-ineq-constraints}
\end{equation}
and the dual is given by
\begin{equation}
\sup_{\mu\in\mathbb{R}^{c},\nu\in\mathbb{R}_{+}^{b}}f_{T}(\mu,\nu),\label{eq:dual-free-energy-eq-ineq-constraints}
\end{equation}
where
\begin{equation}
f_{T}(\mu,\nu)\coloneqq\mu\cdot q+\nu\cdot p-T\Tr\!\left[\ln\!\left(e^{-\frac{1}{T}\left(H-\mu\cdot Q-\nu\cdot P\right)}+I\right)\right].
\end{equation}
The associated Fermi--Dirac thermal measurement operators, which
lead to an optimal form for~\eqref{eq:free-energy-eq-ineq-constraints},
have the following form:
\begin{equation}
M_{T}(\mu,\nu)\coloneqq\left(e^{\frac{1}{T}\left(H-\mu\cdot Q-\nu\cdot P\right)}+I\right)^{-1}.
\end{equation}

Propositions~\ref{prop:simple-approx-bnd},~\ref{prop:no-FD-ent-simple-approx-err-bnd},
\ref{prop:approx-error-spectral-gap}, and~\ref{prop:approx-err-no-FD-obj},
regarding the approximation error between measurement optimization
and measurement free-energy optimization, apply to $F_{T}(\mathcal{Q},\mathcal{P},q,p)$
and $E(\mathcal{Q},\mathcal{P},q,p)$.

Finally, the gradient of $f_{T}(\mu,\nu)$ has the following elements:
\begin{align}
\forall i\in\left[c\right],\quad\frac{\partial}{\partial\mu_{i}}f_{T}(\mu,\nu) & =q_{i}-\Tr\!\left[M_{T}(\mu,\nu)Q_{i}\right],\label{eq:gradient-ineq-eq-constraints-1}\\
\forall j\in\left[b\right],\quad\frac{\partial}{\partial\nu_{j}}f_{T}(\mu,\nu) & =p_{j}-\Tr\!\left[M_{T}(\mu,\nu)P_{j}\right],\label{eq:gradient-ineq-eq-constraints-2}
\end{align}
and the Hessian has the following elements for all $i_{1},i_{2}\in\left[c\right]$
and $j_{1},j_{2}\in\left[b\right]$:
\begin{align}
 & \frac{\partial^{2}}{\partial\mu_{i_{1}}\partial\mu_{i_{2}}}f_{T}(\mu,\nu)\nonumber \\
 & =-\frac{1}{T}\int_{0}^{1}ds\,\Tr\!\left[M_{T}(\mu,\nu,s)Q_{i_{1}}M_{T}(\mu,\nu,1-s)Q_{i_{2}}\right]\\
 & =-\frac{1}{T}\Re\!\left[\Tr\!\left[M_{T}(\mu,\nu)\Phi_{\mu,\nu}(Q_{i_{1}})\left(I-M_{T}(\mu,\nu)\right)Q_{i_{2}}\right]\right],\label{eq:2nd-hessian-exp-1}\\
 & \frac{\partial^{2}}{\partial\mu_{i_{1}}\partial\nu_{j_{1}}}f_{T}(\mu,\nu)\nonumber \\
 & =-\frac{1}{T}\int_{0}^{1}ds\,\Tr\!\left[M_{T}(\mu,\nu,s)Q_{i_{1}}M_{T}(\mu,\nu,1-s)P_{j_{1}}\right]\\
 & =-\frac{1}{T}\Re\!\left[\Tr\!\left[M_{T}(\mu,\nu)\Phi_{\mu,\nu}(Q_{i_{1}})\left(I-M_{T}(\mu,\nu)\right)P_{j_{1}}\right]\right],\\
 & \frac{\partial^{2}}{\partial\nu_{j_{1}}\partial\nu_{j_{2}}}f_{T}(\mu,\nu)\nonumber \\
 & =-\frac{1}{T}\int_{0}^{1}ds\,\Tr\!\left[M_{T}(\mu,\nu,s)P_{j_{1}}M_{T}(\mu,\nu,1-s)P_{j_{2}}\right]\\
 & =-\frac{1}{T}\Re\!\left[\Tr\!\left[M_{T}(\mu,\nu)\Phi_{\mu,\nu}(P_{j_{1}})\left(I-M_{T}(\mu,\nu)\right)P_{j_{2}}\right]\right],
\end{align}
where the measurement operator $M_{T}(\mu,\nu,s)$ is defined for
all $s\in\left[0,1\right]$ as
\begin{equation}
M_{T}(\mu,\nu,s)\coloneqq\frac{e^{\frac{s}{T}\left(H-\mu\cdot Q-\nu\cdot P\right)}}{e^{\frac{1}{T}\left(H-\mu\cdot Q-\nu\cdot P\right)}+I},\label{eq:meas-op-t-depend-1}
\end{equation}
and the quantum channel $\Phi_{\mu,\nu}(X)$ is defined as
\begin{multline}
\Phi_{\mu,\nu}(X)\coloneqq\\
\int_{-\infty}^{\infty}dt\,\gamma(t)e^{-i\left(H-\mu\cdot Q-\nu\cdot P\right)t/T}Xe^{i\left(H-\mu\cdot Q-\nu\cdot P\right)t/T}.
\end{multline}
Also, the high-peak tent probability density $\gamma(t)$ is defined
in~\eqref{eq:high-peak-tent-prob-dens}. The Hessian matrix $\nabla^{2}f_{T}(\mu,\nu)$
is negative semidefinite, and its spectral norm is bounded for all
$\mu\in\mathbb{R}^{c}$ and $\nu\in\mathbb{R}^{b}$ as follows:
\begin{multline}
\left\Vert \nabla^{2}f_{T}(\mu,\nu)\right\Vert \leq\\
\frac{1}{T}\left(\sum_{i\in\left[c\right]}\left\Vert Q_{i}\right\Vert _{1}\left\Vert Q_{i}\right\Vert +\sum_{j\in\left[b\right]}\left\Vert P_{i}\right\Vert _{1}\left\Vert P_{i}\right\Vert \right).\label{eq:hessian-up-bnd-eq-ineq}
\end{multline}
All proofs of the statements in~\eqref{eq:gradient-ineq-eq-constraints-1}--\eqref{eq:hessian-up-bnd-eq-ineq}
mirror those given for Propositions~\ref{prop:gradient-FD-gen-obj-func}--\ref{prop:hessian-NSD-spec-up-bnd}.

\subsection{Solving general semidefinite optimization problems using measurement
optimization}

\label{subsec:Solving-general-SDPs}

A general semidefinite optimization problem in standard form is as
follows \cite[Eq.~(4.51)]{Boyd2004}:
\begin{equation}
\alpha\coloneqq\min_{M\geq0}\left\{ \Tr\!\left[HM\right]:\Tr\!\left[Q_{i}M\right]=q_{i}\,\forall i\in\left[c\right]\right\} ,\label{eq:general-SDP}
\end{equation}
where $H$, $Q_{i}$, and $q_{i}$ are defined as in~\eqref{eq:vector-Herm-ops}--\eqref{eq:constraint-vector}.
The main difference between~\eqref{eq:general-SDP} and~\eqref{eq:measurement-opt-gen}
is the constraint on $M$ (i.e., $M\geq0$ in~\eqref{eq:general-SDP}
while $M\in\mathcal{M}_{d}$ in~\eqref{eq:measurement-opt-gen}).
Similar to what was done previously \cite[Lemma~1]{liu2025qthermoSDPs},
one can reduce a general SDP to a measurement optimization problem
of the form in~\eqref{eq:measurement-opt-gen}. As such, our approach
here can be used to solve general SDPs in standard form.

Let us begin by defining the following SDP as a modification of~\eqref{eq:general-SDP}:
\begin{equation}
\alpha_{R}\coloneqq\min_{M:0\leq M\leq RI}\left\{ \Tr\!\left[HM\right]:\Tr\!\left[Q_{i}M\right]=q_{i}\,\forall i\in\left[c\right]\right\} ,\label{eq:alpha-R-SDP}
\end{equation}
where $R>0$. Since the SDP for $\alpha_{R}$ involves an extra constraint,
the following inequality holds:
\begin{equation}
\alpha_{R}\geq\alpha,
\end{equation}
and it is saturated in the limit as $R\to\infty$:
\begin{equation}
\lim_{R\to\infty}\alpha_{R}=\alpha.
\end{equation}
We can think of $R$ as representing a guess on the spectral norm
of an optimal solution to~\eqref{eq:general-SDP}. If an optimal choice
of $M$ in~\eqref{eq:general-SDP} is such that $M\leq RI$, then
$\alpha_{R}=\alpha$.

By the simple observation that $M'\coloneqq\frac{M}{R}$ is a measurement
operator, we conclude the following reduction of the SDP in~\eqref{eq:alpha-R-SDP}
for $\alpha_{R}$ to a measurement optimization problem of the form
in~\eqref{eq:measurement-opt-gen}.
\begin{lem}
The following equality holds:
\begin{equation}
\alpha_{R}=R\cdot\min_{M\in\mathcal{M}_{d}}\left\{ \Tr\!\left[HM\right]:\Tr\!\left[Q_{i}M\right]=\frac{q_{i}}{R}\,\forall i\in\left[c\right]\right\} .
\end{equation}
\end{lem}

\section{Quantum algorithms for general measurement optimization problems}

\label{sec:Quantum-algorithm-gen-meas-opt}In this section, we provide
several quantum algorithms related to measurement optimization problems,
suitable for implementation on a hybrid architecture consisting of
qubits and qumodes~\cite{Liu2026}. We begin in Section~\ref{subsec:Quantum-algorithm-FD-thermal}
by developing a quantum algorithm for implementing Fermi--Dirac thermal
measurements. It relies on methods of Hamiltonian simulation~\cite{Lloyd1996,Childs2018},
Schr\"odingerization~\cite{Jin2024,Jin2023}, and the power of one
qumode~\cite{Liu2016}, and it works when either the underlying Hamiltonian
$A$ is a local or sparse Hamiltonian or if $A$ is represented as
a linear combination of quantum states, in which case we employ the
methods of density matrix exponentiation~\cite{lloyd2014quantum,Kimmel2017,Go2025,Sims2025}.
In Section~\ref{subsec:Quantum-algorithm-Hessian}, we present a quantum
algorithm for estimating the elements of the Hessian matrix, as given
in~\eqref{eq:2nd-hessian-exp}. This algorithm relies on Hamiltonian
simulation~\cite{Lloyd1996,Childs2018}, Algorithm~\ref{alg:FD-thermal-alg}
for implementing Fermi--Dirac thermal measurements, and the swap
test~\cite{Barenco1997,Buhrman2001}. Finally, in Section~\ref{subsec:HQC-for-gradient-ascent},
we present hybrid quantum--classical algorithms, using both first-
and second-order information, for performing the gradient ascent steps
needed for the optimization in~\eqref{eq:dual-free-energy-meas-opt}.

\subsection{Quantum algorithm for implementing Fermi--Dirac thermal measurements}

\label{subsec:Quantum-algorithm-FD-thermal}

Here we provide a quantum
algorithm for realizing Fermi--Dirac thermal measurements. That is,
suppose that the goal is to realize the Fermi--Dirac thermal measurement
$\left(M_{T}(A),I-M_{T}(A)\right)$, defined for a Hermitian operator
$A$ in~\eqref{eq:gen-FD-meas-op}--\eqref{eq:gen-FD-meas-op-2},
on an input state $\rho$. Then the algorithm should return the value
$0$ with probability $\Tr\!\left[\left(I-M_{T}(A)\right)\rho\right]$
and the value $1$ with probability $\Tr\!\left[M_{T}(A)\rho\right]$.
Under the assumption that one can realize Hamiltonian evolution according
to the Hamiltonian $\hat{x}\otimes A$, where $\hat{x}$ is the position-quadrature
operator, our quantum algorithm for doing so is similar to techniques
used in Schr\"odingerization~\cite{Jin2024,Jin2023} and the power
of one qumode~\cite{Liu2016}. For $T>0$, the algorithm makes use
of the following state vector $|\psi_{T}\rangle$ for a control qumode
register:
\begin{equation}
|\psi_{T}\rangle\coloneqq\int_{-\infty}^{\infty}dp\,\sqrt{\ell_{T}(p)}|p\rangle,\label{eq:control-state-alg-FD}
\end{equation}
where
\begin{equation}
\ell_{T}(p)\coloneqq\frac{e^{p/T}}{T\left(e^{p/T}+1\right)^{2}}
\end{equation}
and $\left\{ |p\rangle\right\} _{p\in\mathbb{R}}$ denotes the momentum
basis. The function $\ell_{T}(p)$ is known as a logistic probability
density function with location parameter $0$ and scale parameter
$T$. The logistic probabiity density function $\ell_{T}(p)$ is equal
to the derivative of the Fermi--Dirac function, i.e.,
\begin{equation}
\frac{e^{p/T}}{T\left(e^{p/T}+1\right)^{2}}=\frac{\partial}{\partial p}\left(\frac{1}{e^{-p/T}+1}\right),
\end{equation}
which is why it is useful in designing a quantum algorithm for realizing
Fermi--Dirac thermal measurements. The state $|\psi_{T_{1}}\rangle\!\langle\psi_{T_{1}}|$
is a non-Gaussian state, with heavier tails than a Gaussian state.
In the limit $T\to0$, the state vector $|\psi_{T}\rangle$ becomes
highly peaked at $p=0$, thus approaching a momentum eigenstate $|p=0\rangle$
in this limit. 
\begin{lyxalgorithm}
\label{alg:FD-thermal-alg}The algorithm for implementing the Fermi--Dirac
thermal measurement $\left(M_{T}(A),I-M_{T}(A)\right)$, where $T=T_{1}T_{2}$,
proceeds as follows:
\begin{enumerate}
\item For $T_{1}>0$, prepare a control qumode register in the state $|\psi_{T_{1}}\rangle\!\langle\psi_{T_{1}}|$
and a data register in the state $\rho$.
\item For $T_{2}>0$, apply the Hamiltonian evolution $e^{-i\hat{x}\otimes A/T_{2}}$
on the control and data registers.
\item Measure the control register with respect to the momentum basis $\left\{ |p\rangle\right\} _{p\in\mathbb{R}}$,
obtaining outcome $p\in\mathbb{R}$.
\item Output $0$ if $p\geq0$ and $1$ if $p<0$.
\end{enumerate}
\end{lyxalgorithm}

\begin{figure}
\begin{centering}
\includegraphics[width=3.1in]{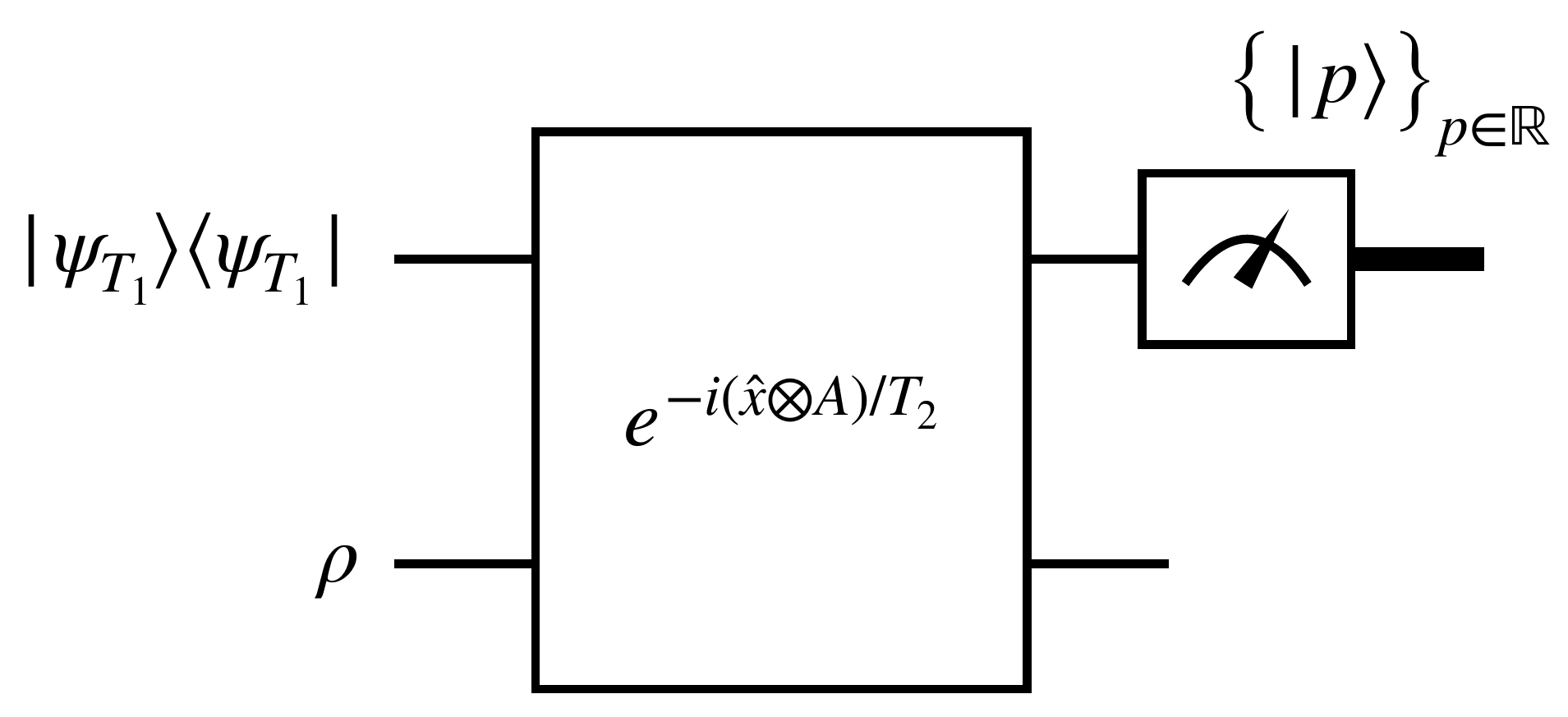}
\par\end{centering}
\caption{Quantum circuit for realizing a Fermi--Dirac thermal measurement
$\left(M_{T}(A),I-M_{T}(A)\right)$ of temperature $T=T_{1}T_{2}$,
where $T_{1},T_{2}>0$, as detailed in Algorithm~\ref{alg:FD-thermal-alg}.
The state $|\psi_{T_{1}}\rangle\!\langle\psi_{T_{1}}|$ of the control
qumode is defined in~\eqref{eq:control-state-alg-FD}, and $\rho$
is the input state on which we would like to perform the desired Fermi--Dirac
thermal measurement. The measurement of the control qumode at the
end is a momentum-quadrature measurement giving an outcome $p\in\mathbb{R}$.
The algorithm finally outputs $0$ if $p\protect\geq0$ and $1$ otherwise.
}\label{fig:fermi-dirac-circuit}

\end{figure}

See Figure~\ref{fig:fermi-dirac-circuit} for a visual depiction of
Algorithm~\ref{alg:FD-thermal-alg}. Interestingly, Algorithm~\ref{alg:FD-thermal-alg}
allows for two temperature parameters $T_{1}$ and $T_{2}$, which
allow for controlling the temperature of the resulting Fermi--Dirac
thermal measurement. In order to lower the temperature of the resulting
Fermi--Dirac thermal measurement, one can choose to keep $T_{1}$
fixed while lowering $T_{2}$, which amounts to increasing the time
needed for Hamiltonian evolution. Alternatively, one can keep $T_{2}$
fixed while lowering $T_{1}$, which amounts to keeping the Hamiltonian
evolution time fixed and instead producing a state $|\psi_{T_{1}}\rangle\!\langle\psi_{T_{1}}|$
that is more highly squeezed with respect to the momentum quadrature
(and thus more difficult to prepare). Proposition~\ref{prop:implementing-FD-measurements}
below asserts that the temperature of the implemented Fermi--Dirac
thermal measurement is equal to the product of $T_{1}$ and $T_{2}$.

\begin{prop}
\label{prop:implementing-FD-measurements}The probability that Algorithm
\ref{alg:FD-thermal-alg} outputs $0$ is equal to $\Tr\!\left[M_{T_{1}T_{2}}(A)\rho\right]$,
and the probability that it outputs $1$ is equal to $\Tr\!\left[\left(I-M_{T_{1}T_{2}}(A)\right)\rho\right]$.
\end{prop}

\begin{proof}
See Appendix~\ref{sec:implementing-FD-measurements}.
\end{proof}
We now consider a particular scenario in which the Hermitian operator
$A$ is equal to a linear combination of quantum states:
\begin{equation}
A=\sum_{x}\alpha_{x}\sigma_{x},\label{eq:linear-combo-states}
\end{equation}
where $\alpha_{x}\in\mathbb{R}$ and $\sigma_{x}$ is a quantum state,
for all $x$. We suppose that $\sum_x |\alpha_x| > 0$. Let us also suppose that we have sample access to each
state $\sigma_{x}$. This scenario is relevant for quantum hypothesis
testing applications, in which one might have sample access to each
state $\sigma_{x}$ and not necessarily a description of each state.
It is then of interest to implement the Fermi--Dirac measurement
$\left(M_{T}(A),I-M_{T}(A)\right)$, with $A$ as in~\eqref{eq:linear-combo-states}.
For this purpose, we can use the techniques of density matrix exponentiation
\cite{lloyd2014quantum,Kimmel2017,Go2025}, along with Algorithm~\ref{alg:FD-thermal-alg}.

Let us suppose that we can realize Hamiltonian evolution according
to the fixed Hamiltonian $\hat{x}\otimes F$, where $F$ is the unitary
and Hermitian swap operator that acts on the data register and another
auxiliary register:
\begin{equation}
F\coloneqq\sum_{i,j}|i\rangle\!\langle j|\otimes|j\rangle\!\langle i|.
\end{equation}
Then one can realize the Hamiltonian evolution $A$, up to an error
$\varepsilon$ and for a time $t>0$, by proceeding as follows:
\begin{lyxalgorithm}
\label{alg:DME-alg}The algorithm for realizing the Hamiltonian evolution
$e^{-i\left(\hat{x}\otimes A\right)t}$ up to an error $\varepsilon$,
where $A$ has the form in~\eqref{eq:linear-combo-states}, proceeds
according to the following steps:
\begin{enumerate}
\item Let $|\psi_{T}\rangle\!\langle\psi_{T}|$ be the state of the control
qumode register, where $T>0$, and let $\rho$ be the state of the
data register. Fix $n\in\mathbb{N}$ such that $n=O\!\left(c^{2}t^{2}/\varepsilon\right)$,
where $c\equiv\left\Vert \alpha\right\Vert _{1}\coloneqq\sum_{x}\left|\alpha_{x}\right|$.
Set $\Delta\coloneqq\frac{ct}{n}$.
\item Sample $x$ according to the probability distribution $p(x)\coloneqq\frac{\left|\alpha_{x}\right|}{c}$.
Prepare the auxiliary register in the state $\sigma_{x}$.
\item If $\alpha_{x}>0$, apply the Hamiltonian evolution $e^{-i\left(\hat{x}\otimes F\right)\Delta}$,
to all registers, where $\hat{x}$ acts on the control register and
$F$ on the data and auxiliary registers. If $\alpha_{x}<0$, apply
the Hamiltonian evolution $e^{i\left(\hat{x}\otimes F\right)\Delta}$,
to all registers. Discard the auxiliary register.
\item Repeat steps 2 and 3 $n$ times.
\end{enumerate}
\end{lyxalgorithm}

\begin{figure*}
\begin{centering}
\includegraphics[width=6.2in]{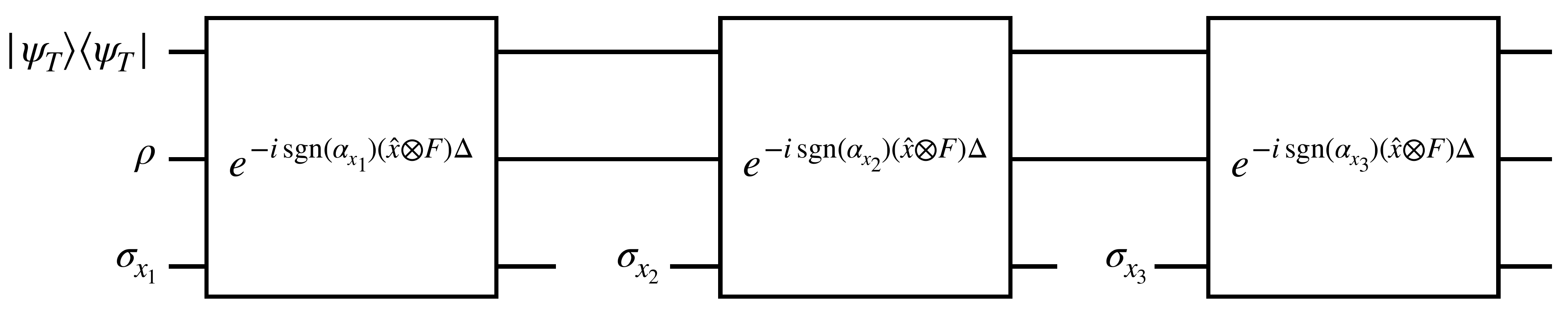}
\par\end{centering}
\caption{Quantum circuit for simulating the evolution in~\eqref{eq:simulated-evol-DME}
for the case $n=3$, as detailed in Algorithm~\ref{alg:DME-alg}.
}\label{fig:DME-sim}
\end{figure*}

See Figure~\ref{fig:DME-sim} for a visual depiction of Algorithm
\ref{alg:DME-alg} for the case $n=3$. By following an analysis similar
to that in \cite[Appendix~D]{Sims2025}, it follows that the state
resulting from Algorithm~\ref{alg:DME-alg} is $\varepsilon$-close
in trace distance to the following state:
\begin{equation}
e^{-i\left(\hat{x}\otimes A\right)t}\left(|\psi_{T}\rangle\!\langle\psi_{T}|\otimes\rho\right)e^{i\left(\hat{x}\otimes A\right)t}.\label{eq:simulated-evol-DME}
\end{equation}
To get a sense for why this holds, define
\begin{align}
\omega_{+} & \equiv\sum_{x:\alpha_{x}>0}p(x)\sigma_{x},\\
\omega_{-} & \equiv\sum_{x:\alpha_{x}<0}p(x)\sigma_{x},
\end{align}
observe that $\omega_{+}-\omega_{-}=\frac{1}{c}\sum_{x}\alpha_{x}\sigma_{x}=\frac{1}{c}A$,
and consider that the expected state after the first step of Algorithm~\ref{alg:DME-alg} has the following form:
\begin{align}
 & \Tr_{\text{aux}}\!\left[e^{-i\left(\hat{x}\otimes F\right)\Delta}\left(|\psi_{T}\rangle\!\langle\psi_{T}|\otimes\rho\otimes\omega_{+}\right)e^{i\left(\hat{x}\otimes F\right)\Delta}\right]\nonumber \\
 & \qquad+\Tr_{\text{aux}}\!\left[e^{i\left(\hat{x}\otimes F\right)\Delta}\left(|\psi_{T}\rangle\!\langle\psi_{T}|\otimes\rho\otimes\omega_{-}\right)e^{-i\left(\hat{x}\otimes F\right)\Delta}\right]\nonumber \\
 & =|\psi_{T}\rangle\!\langle\psi_{T}|\otimes\rho\nonumber \\
 & \qquad-i\Tr_{\text{aux}}\!\left[\left[\hat{x}\otimes F,|\psi_{T}\rangle\!\langle\psi_{T}|\otimes\rho\otimes\omega_{+}\right]\right]\Delta\nonumber \\
 & \qquad+i\Tr_{\text{aux}}\!\left[\left[\hat{x}\otimes F,|\psi_{T}\rangle\!\langle\psi_{T}|\otimes\rho\otimes\omega_{-}\right]\right]\Delta+O\!\left(\Delta^{2}\right)\\
 & =|\psi_{T}\rangle\!\langle\psi_{T}|\otimes\rho-i\left[\hat{x}\otimes\omega_{+},|\psi_{T}\rangle\!\langle\psi_{T}|\otimes\rho\right]\Delta\nonumber \\
 & \qquad+i\left[\hat{x}\otimes\omega_{-},|\psi_{T}\rangle\!\langle\psi_{T}|\otimes\rho\right]\Delta+O\!\left(\Delta^{2}\right)\\
 & =|\psi_{T}\rangle\!\langle\psi_{T}|\otimes\rho-i\left[\hat{x}\otimes\left(\omega_{+}-\omega_{-}\right),|\psi_{T}\rangle\!\langle\psi_{T}|\otimes\rho\right]\Delta\nonumber \\
 & \qquad+O\!\left(\Delta^{2}\right)\\
 & =|\psi_{T}\rangle\!\langle\psi_{T}|\otimes\rho-i\left[\hat{x}\otimes A,|\psi_{T}\rangle\!\langle\psi_{T}|\otimes\rho\right]\frac{\Delta}{c}+O\!\left(\Delta^{2}\right)\\
 & =e^{-i\left(\hat{x}\otimes A\right)\Delta/c}\left(|\psi_{T}\rangle\!\langle\psi_{T}|\otimes\rho\right)e^{i\left(\hat{x}\otimes A\right)\Delta/c}+O\!\left(\Delta^{2}\right).
\end{align}
In the above calculation, the first-order expansions are understood
after applying the corresponding unitaries to the states under
consideration, rather than as operator-norm expansions. This distinction
is necessary because $\hat{x}$ is an unbounded operator. More precisely, for every state vector $|\varphi\rangle$ and
for every bounded Hermitian operator $B$ acting on the finite-dimensional
data registers, the integral remainder formula gives
\begin{multline}
e^{-i(\hat{x}\otimes B)\Delta}
\left(|\psi_{T}\rangle\otimes|\varphi\rangle\right)
\\
 =
\left(I-i(\hat{x}\otimes B)\Delta\right)
\left(|\psi_{T}\rangle\otimes|\varphi\rangle\right)
+O(\Delta^{2}),
\end{multline}
where the error is in Hilbert-space norm. Indeed,
\begin{align}
&\left\|
e^{-i(\hat{x}\otimes B)\Delta}
\left(|\psi_{T}\rangle\otimes|\varphi\rangle\right)
-
\left(I-i(\hat{x}\otimes B)\Delta\right)
\left(|\psi_{T}\rangle\otimes|\varphi\rangle\right)
\right\|
\nonumber\\
&\qquad\leq
\frac{\Delta^{2}}{2}
\left\|(\hat{x}\otimes B)^{2}
\left(|\psi_{T}\rangle\otimes|\varphi\rangle\right)\right\|
\nonumber\\
&\qquad=
\frac{\Delta^{2}}{2}
\left\|\hat{x}^{2}|\psi_{T}\rangle\right\|
\left\|B^{2}|\varphi\rangle\right\|
\nonumber\\
&\qquad\leq
\frac{\Delta^{2}}{2}
\left\|\hat{x}^{2}|\psi_{T}\rangle\right\|
\left\|B\right\|^{2}.
\end{align}
The right-hand side is finite because
\begin{equation}
\left\|\hat{x}^{2}|\psi_{T}\rangle\right\|^{2}
=
\langle\psi_{T}|\hat{x}^{4}|\psi_{T}\rangle
<\infty.
\end{equation}
Thus, the first-order expansions used above are valid on the states
under consideration, even though $\hat{x}$ is unbounded. The same
argument applies with $B=F$ and $B=A/c$.

Equivalently, at the level of density operators, the integral remainder
formula gives
\begin{multline}
e^{-i(\hat{x}\otimes B)\Delta}
\Gamma
e^{i(\hat{x}\otimes B)\Delta}
\\
=
\Gamma-i[\hat{x}\otimes B,\Gamma]\Delta
+O(\Delta^{2}),
\end{multline}
where the error is in trace norm, for the states $\Gamma$ appearing
above. The finiteness of the corresponding second-order remainder
follows from
$\langle\psi_{T}|\hat{x}^{4}|\psi_{T}\rangle<\infty$ and the fact that
$B$ is bounded.

In the above, we also employed
the following identity:
\begin{equation}
\Tr_{2}[F(\tau\otimes\xi)]=\xi\tau,
\end{equation}
which holds for all states $\tau$ and $\xi$. Repeating this procedure
$n$ times leads to an error of
\begin{equation}
O\!\left(n\Delta^{2}\right)=O\!\left(n\left(\frac{ct}{n}\right)^{2}\right)=O\!\left(\frac{c^{2}t^{2}}{n}\right),
\end{equation}
so that we require $n=O\!\left(c^{2}t^{2}/\varepsilon\right)$ in
order to have $\varepsilon$ error.

Given that Algorithm~\ref{alg:DME-alg} simulates the evolution in
\eqref{eq:simulated-evol-DME}, it can be used as a subroutine in
Algorithm~\ref{alg:FD-thermal-alg} (specifically in step 2 therein)
in order to realize the Fermi--Dirac thermal measurement $\left(M_{T}(A),I-M_{T}(A)\right)$,
where $A$ is a linear combination of states as in~\eqref{eq:linear-combo-states}.

\subsection{Quantum algorithm for estimating Hessian matrix elements}

\label{subsec:Quantum-algorithm-Hessian}

Here we provide a quantum algorithm for estimating the elements of
the Hessian matrix for $f_{T}(\mu)$, as given in~\eqref{eq:2nd-hessian-exp}.
We suppose a quantum access model, called linear combinations of states
and previously considered in~\cite{Brandao2019,Apeldoorn2019,Chen2025},
in which $H$ and each $Q_{i}$ are available as a linear combination
of quantum states. That is,
\begin{align}
H & =\sum_{i}h_{i}\rho_{i},\label{eq:linear-combo-states-1}\\
Q_{i} & =\sum_{k}\alpha_{i,k}\sigma_{i,k},\label{eq:linear-combo-states-2}
\end{align}
where $h_{i}\in\mathbb{R}$ and $\rho_{i}$ is a quantum state for
all $i$ and $\alpha_{i,k}\in\mathbb{R}$ and $\sigma_{i,k}$ is a
state for all $i$ and $k$. Let us define the following norms:
\begin{align}
\left\Vert h\right\Vert _{1} & \coloneqq\sum_{i}\left|h_{i}\right|,\label{eq:h-norm}\\
\left\Vert \alpha_{i}\right\Vert _{1} & \coloneqq\sum_{k}\left|\alpha_{i,k}\right|.\label{eq:alpha_j-norm}
\end{align}

By density matrix exponentiation~\cite{lloyd2014quantum,Kimmel2017,Sims2025,Go2025},
the Hamiltonian evolution $e^{-i\left(H-\mu\cdot Q\right)t/T}$ can
be simulated, and by Algorithms~\ref{alg:FD-thermal-alg} and~\ref{alg:DME-alg},
the Fermi--Dirac thermal measurement $\left(M_{T}(\mu),I-M_{T}(\mu)\right)$
can be implemented, where $M_{T}(\mu)$ is defined in~\eqref{eq:FD-meas-op}.
There are controllable errors that occur in simulating both the Hamiltonian
evolution and Fermi--Dirac thermal measurement.

With this in place, we can now present a quantum algorithm for estimating
the element of the Hessian matrix indexed by $\left(i,j\right)$,
as given in~\eqref{eq:2nd-hessian-exp}. See Figure~\ref{fig:hessian-estimation}
for a visual depiction of the main quantum circuit used for Hessian
estimation.
\begin{lyxalgorithm}
\label{alg:hessian-est-alg}The algorithm for estimating the Hessian
element indexed by $\left(i,j\right)$ in~\eqref{eq:2nd-hessian-exp}
proceeds according to the following steps:
\begin{enumerate}
\item Set $s=1$.
\item Sample $k_{1}$ from the probability distribution $p_{i}(k)\equiv\frac{\left|\alpha_{i,k}\right|}{\left\Vert \alpha_{i}\right\Vert _{1}}$,
and sample $k_{2}$ from the probability distribution $p_{j}(k)\equiv\frac{\left|\alpha_{j,k}\right|}{\left\Vert \alpha_{j}\right\Vert _{1}}$.
\item Prepare a first data register in the state $\sigma_{i,k_{1}}$, and
prepare a second data register in the state $\sigma_{j,k_{2}}$.
\item Sample $t$ from the high-peak probability density $\gamma(t)$ in
\eqref{eq:high-peak-tent-prob-dens}, and apply the Hamiltonian evolution
$e^{-i\left(H-\mu\cdot Q\right)t/T}$ to the first data register.
\item Prepare a control qubit in the state $|-\rangle\!\langle-|$, where
$|-\rangle\coloneqq\frac{1}{\sqrt{2}}\left(|0\rangle-|1\rangle\right)$,
and perform a controlled swap $|0\rangle\!\langle0|\otimes I+|1\rangle\!\langle1|\otimes F$
on the control qubit and the two data registers.
\item Measure the control qubit in the eigenbasis of $\sigma_{X}$, and
denote the outcome by $Y_{s}\in\left\{ +1,-1\right\} $. Perform the
Fermi--Dirac thermal measurement $\left(M_{T}(\mu),I-M_{T}(\mu)\right)$
on both data registers.
\item Set
\begin{equation}
Z_{s}\leftarrow\frac{1}{T}\signum(\alpha_{i,k}\alpha_{j,k})\left\Vert \alpha_{i}\right\Vert _{1}\left\Vert \alpha_{j}\right\Vert _{1}Y_{s}
\end{equation}
if the outcome of the first Fermi--Dirac measurement is $M_{T}(\mu)$
and the outcome of the second Fermi--Dirac thermal measurement is
$I-M_{T}(\mu)$. Otherwise, set $Z_{s}=0$.
\item Set $s\leftarrow s+1$.
\item Repeat Steps 2-8 $S$ times, and set $\overline{Z_{S}}\leftarrow\frac{1}{S}\sum_{s=1}^{S}Z_{s}$.
\end{enumerate}
\end{lyxalgorithm}

\begin{figure}
\begin{centering}
\includegraphics[width=3in]{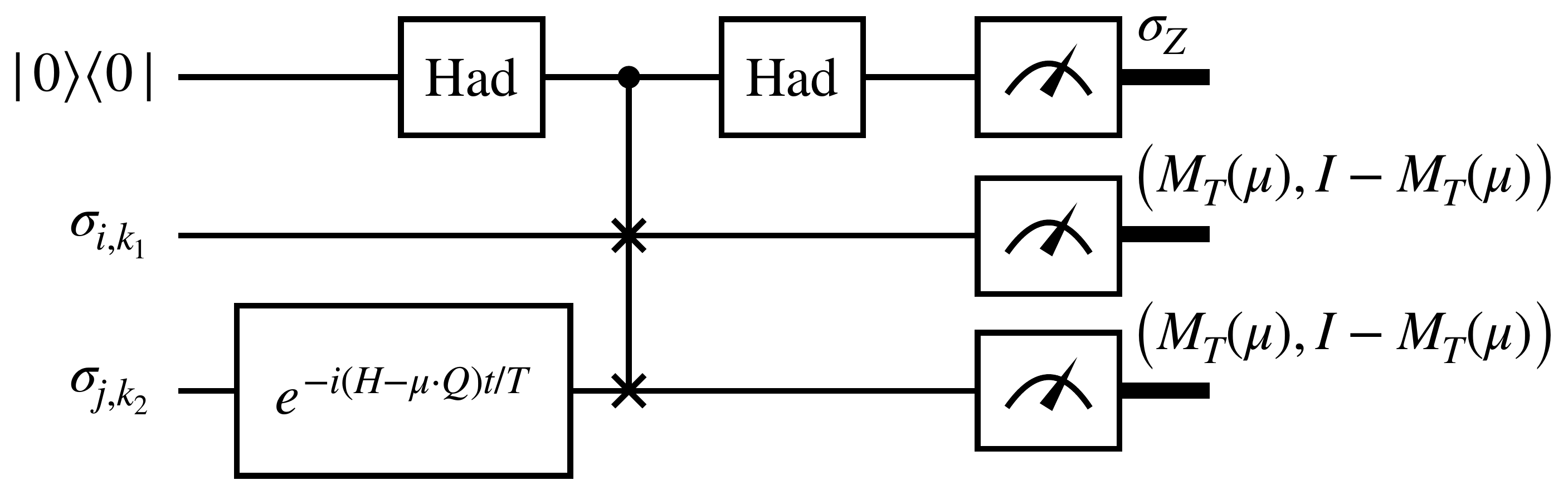}
\par\end{centering}
\caption{Quantum circuit for estimating the Hessian matrix element $\left(i,j\right)$,
as detailed in Algorithm~\ref{alg:hessian-est-alg}. \textquotedblleft Had\textquotedblright{}
stands for the qubit Hadamard gate. The classical values $k_{1}$,
$k_{2}$, and $t$ are sampled from the probability distributions
$\left|\alpha_{i,k_{1}}\right|/\left\Vert \alpha_{i}\right\Vert _{1}$,
$\left|\alpha_{j,k_{2}}\right|/\left\Vert \alpha_{j}\right\Vert _{1}$,
and $\gamma(t)$ in~\eqref{eq:high-peak-tent-prob-dens}, respectively.
The measurements at the end consist of a computational basis measurement
on the control qubit and the Fermi--Dirac thermal measurement $\left(M_{T}(\mu),I-M_{T}(\mu)\right)$
on both data registers. See Algorithm~\ref{alg:hessian-est-alg} for
more details.}\label{fig:hessian-estimation}

\end{figure}

For fixed values of $k_{1}$, $k_{2}$, and $t$, and defining the
unitary channels
\begin{align}
\mathcal{V}_{t}(\cdot) & \coloneqq e^{-i\left(H-\mu\cdot Q\right)t/T}(\cdot)e^{i\left(H-\mu\cdot Q\right)t/T},\\
\mathcal{U}(\cdot) & \coloneqq U(\cdot)U^{\dag},\\
U & \coloneqq|0\rangle\!\langle0|\otimes I+|1\rangle\!\langle1|\otimes F,
\end{align}
consider that the following equality holds:
\begin{align}
 & \Tr\!\left[\begin{array}{c}
\left(\sigma_{X}\otimes M_{T}(\mu)\otimes\left(I-M_{T}(\mu)\right)\right)\times\\
\mathcal{U}(|-\rangle\!\langle-|\otimes\mathcal{V}_{t}(\sigma_{i,k_{1}})\otimes\sigma_{j,k_{2}})
\end{array}\right]\nonumber \\
 & =-\Re\left[\Tr\!\left[\left(M_{T}(\mu)\otimes\left(I-M_{T}(\mu)\right)\right)(\mathcal{V}_{t}(\sigma_{i,k_{1}})\otimes\sigma_{j,k_{2}})F\right]\right]\\
 & =-\Re\left[\Tr\!\left[(M_{T}(\mu)\mathcal{V}_{t}(\sigma_{i,k_{1}})\otimes\left(I-M_{T}(\mu)\right)\sigma_{j,k_{2}})F\right]\right]\\
 & =-\Re\left[\Tr\!\left[M_{T}(\mu)\mathcal{V}_{t}(\sigma_{i,k_{1}})\left(I-M_{T}(\mu)\right)\sigma_{j,k_{2}}\right]\right],
\end{align}
which is proportional to the expected value of the random variable
$Z_{s}$, as defined in step 7 of Algorithm~\ref{alg:hessian-est-alg}.
If we now take the expectation of this expression over the randomness
associated with $k_{1}$, $k_{2}$, and $t$, the resulting expectation
is as follows:
\begin{equation}
-\frac{1}{\left\Vert \alpha_{i}\right\Vert _{1}\left\Vert \alpha_{j}\right\Vert _{1}}\Re\!\left[\Tr\!\left[M_{T}(\mu)\Phi_{\mu}(Q_{i})\left(I-M_{T}(\mu)\right)Q_{j}\right]\right],
\end{equation}
where the quantum channel $\Phi_{\mu}$ is defined in~\eqref{eq:Phi-channel}.
Finally multiplying this expression by $\frac{1}{T}\left\Vert \alpha_{i}\right\Vert _{1}\left\Vert \alpha_{j}\right\Vert _{1}$
then results in an expression equal to~\eqref{eq:2nd-hessian-exp}.
Thus, by applying the Hoeffding inequality~\cite{Hoeffding1963} (e.g.,
in the form of \cite[Theorem~1]{Bandyopadhyay2023}), a sufficient
number of steps, $S$, for Algorithm~\ref{alg:hessian-est-alg} to
arrive at an $\varepsilon$-accurate estimate, with success probability
$\geq1-\delta$, is given by
\begin{equation}
O\!\left(\frac{\left\Vert \alpha_{i}\right\Vert _{1}^{2}\left\Vert \alpha_{j}\right\Vert _{1}^{2}}{T^{2}\varepsilon^{2}}\ln\!\left(\frac{1}{\delta}\right)\right).
\end{equation}

\subsection{Hybrid quantum--classical algorithms for gradient ascent}

\label{subsec:HQC-for-gradient-ascent}

In this section, we sketch various algorithms that can be used for
solving the general measurement optimization problem in~\eqref{eq:measurement-opt-gen}.
Here we assume the input access model to be linear combination of
quantum states, as given previously in~\eqref{eq:linear-combo-states-1}--\eqref{eq:linear-combo-states-2}.

A first-order optimization algorithm proceeds similarly to Algorithm
\ref{alg:classical-gradient-ascent}, with the main difference being
that it is a hybrid quantum--classical algorithm, employing stochastic
gradient ascent, rather than a deterministic gradient ascent algorithm.
We provide it here for completeness:
\begin{lyxalgorithm}
\label{alg:q-stoch-gradient-ascent}The hybrid quantum--classical
algorithm for measurement optimization consists of the following steps:
\begin{enumerate}
\item Set $\delta>0$ to be the desired error, set $T\leftarrow\frac{\delta}{4d\ln2}$
to be the temperature, initialize $\mu^{0}\leftarrow\left(0,\ldots,0\right)$,
fix the learning rate $\eta>0$, and set the number of steps, $J$.
\item For all $j\in\left[J\right]$ and $i\in\left[c\right]$,
\begin{enumerate}
\item Estimate $\Tr\!\left[M_{T}(\mu^{j-1})Q_{i}\right]$ as $\tilde{Q}_{i}$,
by using Algorithm~\ref{alg:FD-thermal-alg}.
\item Set
\begin{equation}
\mu_{i}^{j}\leftarrow\mu_{i}^{j-1}+\eta\left(q_{i}-\tilde{Q}_{i}\right).
\end{equation}
\end{enumerate}
\item Estimate $\Tr\!\left[M_{T}(\mu^{J})H\right]$ as $\tilde{H}$, by
using Algorithm~\ref{alg:FD-thermal-alg}.
\item Output $\tilde{f}_{T}(\mu^{J})=\mu^{J}\cdot q+\tilde{H}-\mu^{J}\cdot\tilde{Q}$,
as defined in~\eqref{eq:approx-without-FD-entropy}.
\end{enumerate}
\end{lyxalgorithm}

We expect the runtime of Algorithm~\ref{alg:q-stoch-gradient-ascent}
to be comparable to that reported in \cite[Theorem~2]{liu2025qthermoSDPs}.
Indeed, based on standard runtime estimates for stochastic gradient
ascent (see, e.g., \cite[Theorem~6.3]{Bubeck2015}), the runtime should
be inverse quartic in the desired approximation error $\delta$, quadratic
in the search radius $r$ (which upper bounds the norm $\left\Vert \mu_{T}^{\star}\right\Vert $
of an optimal $\mu_{T}^{\star}$), and a low-order polynomial in the
norms defined in~\eqref{eq:h-norm}--\eqref{eq:alpha_j-norm}. If
there is no information about the spectral gap $\Delta$ and the ground-space
degeneracy $d_{0}$, as defined in~\eqref{eq:def-d0-1} and~\eqref{eq:def-Delta-1},
respectively, then the runtime of Algorithm~\ref{alg:q-stoch-gradient-ascent}
is linear in the dimension $d$, or equivalently, exponential in $\ln d$,
with $\ln d$ corresponding to the number of qubits. However, if $\Delta^{-1}$
and $d_{0}$ are polynomial in $\ln d$, then Algorithm~\ref{alg:q-stoch-gradient-ascent}
is an efficient hybrid quantum--classical algorithm for solving the
general measurement optimization problem in~\eqref{eq:measurement-opt-gen},
with its runtime instead proportional to a polynomial in $\ln d$.

We now sketch a second-order hybrid quantum--classical algorithm
for measurement optimization:
\begin{lyxalgorithm}
\label{alg:second-order-HQC}The second-order hybrid quantum--classical
algorithm for measurement optimization consists of the following steps:
\begin{enumerate}
\item Set $\delta>0$ to be the desired error, set $T\leftarrow\frac{\delta}{4d\ln2}$
to be the temperature, initialize $\mu^{0}\leftarrow\left(0,\ldots,0\right)$,
fix the learning rate $\eta>0$, and set the number of steps, $J$.
\item For all $j\in\left[J\right]$ and $i\in\left[c\right]$,
\begin{enumerate}
\item Estimate $\Tr\!\left[M_{T}(\mu^{j-1})Q_{i}\right]$ as $\tilde{Q}_{i}$,
by using Algorithm~\ref{alg:FD-thermal-alg}.
\item For all $k\in\left[c\right]$, estimate the Hessian matrix element
\begin{equation}
-\frac{1}{T}\Re\!\left[\Tr\!\left[M_{T}(\mu)\Phi_{\mu}(Q_{i})\left(I-M_{T}(\mu)\right)Q_{k}\right]\right]
\end{equation}
as $\widetilde{\nabla^{2}}_{ik}$, by using Algorithm~\ref{alg:hessian-est-alg}.
\end{enumerate}
\item Set $\widetilde{\nabla}$to be the estimate of the gradient vector:
$\widetilde{\nabla}\leftarrow\left(q_{i}-\tilde{Q}_{i}\right)_{i\in\left[c\right]}$.
Solve for $\Lambda$ in $\widetilde{\nabla^{2}}\Lambda=\widetilde{\nabla}$,
and set
\begin{equation}
\mu^{j}\leftarrow\mu^{j-1}-\eta\Lambda.
\end{equation}
\item Estimate $\Tr\!\left[M_{T}(\mu^{J})H\right]$ as $\tilde{H}$.
\item Output $\tilde{f}_{T}(\mu^{J})=\mu^{J}\cdot q+\tilde{H}-\mu^{J}\cdot\tilde{Q}$,
as defined in~\eqref{eq:approx-without-FD-entropy}.
\end{enumerate}
\end{lyxalgorithm}

Similar comments as above apply to Algorithm~\ref{alg:second-order-HQC}.
Without information about the spectral gap $\Delta$ and groundspace
degeneracy $d_{0}$, the expected runtime of Algorithm~\ref{alg:second-order-HQC}
is exponential in the number of qubits. However, if information is
available and similar polynomial bounds apply to $\Delta^{-1}$ and
$d_{0}$, then the runtime can be reduced to polynomial in the number
of qubits.

\subsection{Discussion on computational complexity}

For general measurement optimization problems, we should not expect
to have an efficient quantum algorithm. Indeed, if the states $\rho$
and $\sigma$ are states generated by arbitrary quantum circuits,
then the ability to implement the optimal Helstrom--Holevo measurement
implies the ability to solve computational problems that are complete
for the complexity class quantum statistical zero knowledge (QSZK)
\cite{Watrous2002LimitsPowerQSZK,Watrous2009ZeroKnowledgeQuantum,VidickWatrous2016QuantumProofs}.
It is believed that problems in QSZK are difficult for quantum computers
to solve, as it contains QMA, which in turn contains NP. This difficulty
is reflect in the fact that the best runtime our approach can guarantee
in the general case is linear in the dimension of the states, or equivalently,
exponential in the number of qubits needed to represent the states.

However, we have also outlined cases in which the runtime of our quantum
algorithms are efficient. Indeed, if the spectral gap and ground-space
degeneracy conditions outlined in Proposition~\ref{prop:approx-error-spectral-gap}
hold with $\Delta^{-1}$ and $d_{0}$ polynomial in $\ln d$, then
the resulting quantum algorithms are efficient for solving the corresponding
measurement optimization problem.

\section{Application to quantum hypothesis testing and binary classification}

\label{sec:Application-q-hypo-test}

Quantum hypothesis testing plays a prominent role in quantum information
theory, being foundational in its own right~\cite{Helstrom1967,Helstrom1969,Holevo1972}
while also being fundamentally connected to quantum communication
\cite{Hayashi2003,Hayashi2017,khatri2024} and quantum resource theories
\cite{Chitambar2019,Gour2025} more generally. One can consider a
variety of scenarios, with the most basic ones being symmetric and
asymmetric binary hypothesis testing. Additionally, one can consider binary classification, as well as
composite hypothesis testing, in which one or both hypotheses involve
a state selected from a set of states, rather than being a single
state.

In all scenarios of quantum hypothesis testing, the goal is to find
a quantum measurement that is optimal for the particular task being
considered. As such, this involves optimization, and we highlight
four scenarios below that are special cases of the measurement optimization
problems in~\eqref{eq:measurement-opt-gen} and~\eqref{eq:measurement-opt-gen-ineq}.
As such, all of our findings from Section~\ref{sec:General-measurement-optimization}
are applicable to these quantum hypothesis testing problems. Perhaps
most interestingly, our results demonstrate that Fermi--Dirac thermal
measurements are nearly optimal in all of the scenarios whenever the
temperature is sufficiently low, thus motivating this class of measurements
for practical problems of interest.

\subsection{Symmetric binary quantum hypothesis testing}

\label{subsec:Symmetric-binary-HT}

Let us begin by recalling symmetric quantum hypothesis testing. In
this setting, a state $\rho$ is selected with probability $p\in\left(0,1\right)$,
and a state $\sigma$ is selected with probability $q\equiv1-p$.
Both $\rho$ and $\sigma$ are $d\times d$ density operators, as
defined in~\eqref{eq:dens-ops-def}. The goal is to perform a measurement
to minimize the expected error probability when attempting to identify
the chosen state. Thus, the optimization problem is as follows:
\begin{align}
p_{e}(p,\rho,\sigma) & \coloneqq\min_{M\in\mathcal{M}_{d}}\left\{ p\Tr\!\left[\left(I-M\right)\rho\right]+q\Tr\!\left[M\sigma\right]\right\} \\
 & =\min_{M\in\mathcal{M}_{d}}\left\{ p+\Tr\!\left[M\left(q\sigma-p\rho\right)\right]\right\} \\
 & =p+\min_{M\in\mathcal{M}_{d}}\left\{ \Tr\!\left[M\left(q\sigma-p\rho\right)\right]\right\} .\label{eq:rewrite-sym-hypo-test}
\end{align}
By inspecting~\eqref{eq:rewrite-sym-hypo-test}, we observe that it
is a special case of the optimization problem in~\eqref{eq:measurement-opt-gen},
where $H=q\sigma-p\rho$ and there are no constraints. Applying Proposition
\ref{prop:dual-gen-meas-opt}, we conclude that
\begin{equation}
p_{e}(p,\rho,\sigma)=p-\Tr\!\left[\left(q\sigma-p\rho\right)_{-}\right],
\end{equation}
and an optimal measurement operator has the form in~\eqref{eq:optimal-meas-op},
given by
\begin{equation}
\Pi_{q\sigma<p\rho}+M_{0},
\end{equation}
where $0\leq M_{0}\leq\Pi_{q\sigma=p\rho}$. Such an optimal measurement
is known as a Helstrom--Holevo measurement~\cite{Helstrom1967,Helstrom1969,Holevo1972}.

For $T>0$, the modified free-energy optimization problem is a special
case of~\eqref{eq:primal-FD-obj}, given by:
\begin{multline}
p_{e,T}(p,\rho,\sigma)\coloneqq \\
p+\min_{M\in\mathcal{M}_{d}}\left\{ \Tr\!\left[M\left(q\sigma-p\rho\right)\right]-T\cdot S_{\FD}(M)\right\} .\label{eq:opt-prob-sym-hypo-test}
\end{multline}
Applying Proposition~\ref{prop:dual-FD-free-energy}, we conclude
that
\begin{equation}
p_{e,T}(p,\rho,\sigma)=p-T\Tr\!\left[\ln\!\left(e^{-\frac{1}{T}\left(q\sigma-p\rho\right)}+I\right)\right],
\end{equation}
and the optimal measurement operator for~\eqref{eq:opt-prob-sym-hypo-test}
is as follows:
\begin{equation}
M_{T}=\left(e^{\frac{1}{T}\left(q\sigma-p\rho\right)}+I\right)^{-1}.
\end{equation}

In the zero-temperature limit, the following equalities hold:
\begin{align}
\lim_{T\to0^{+}}p_{e,T}(p,\rho,\sigma) & =p_{e}(p,\rho,\sigma),\\
\lim_{T\to0^{+}}M_{T} & =\Pi_{q\sigma<p\rho}+\frac{1}{2}\Pi_{q\sigma=p\rho},
\end{align}
which follow from~\eqref{eq:temp-T-obj-to-zero-temp-obj} and~\eqref{eq:temp-T-meas-op-to-zero-temp-meas-op},
respectively.

In general, the following bound from Proposition~\ref{prop:simple-approx-bnd}
is applicable:
\begin{equation}
\left|p_{e}(p,\rho,\sigma)-p_{e,T}(p,\rho,\sigma)\right|\leq T\cdot d\ln2.
\end{equation}
We additionally have the following bound from Proposition~\ref{prop:approx-error-spectral-gap}:
\begin{equation}
\left|p_{e}(p,\rho,\sigma)-p_{e,T}(p,\rho,\sigma)\right|\leq T\left(d_{0}\ln2+\left(d-d_{0}\right)e^{-\frac{\Delta}{T}}\right),\label{eq:hels-hol-err-bound-FD}
\end{equation}
where $d_{0}\equiv\dim\ker(q\sigma-p\rho)$, $\Delta\equiv\min_{i:\lambda_{i}\neq0}\left|\lambda_{i}\right|$,
and $\lambda_{i}$ denotes the $i$th eigenvalue of $q\sigma-p\rho$.
By applying the error bound in~\eqref{eq:hels-hol-err-bound-FD},
the temperature $T$ needed to have $\varepsilon$ error, when approximating
$p_{e}(p,\rho,\sigma)$ by $p_{e,T}(p,\rho,\sigma)$ and employing
a Fermi--Dirac thermal measurement, is given by~\eqref{eq:temp-eps-error-spec-gap}.

\subsection{Binary classification}

\label{subsec:binary-class}

The task of binary classification is similar to a symmetric
hypothesis testing problem. In binary classification, there is a Bayes cost matrix associated
to the results of classification.

Let us first begin with a simple example of $1$-$0$ loss and show that it is equivalent to the symmetric hypothesis testing problem. Suppose we have training set $(\sigma_i, y_i)_{i=1}^m$, where $\sigma_i$ is a quantum state with corresponding label $y_i \in \{0, 1\}$. This means that the states belong to one of two classes $\mathcal{C}_0\coloneqq \{\sigma_i \vert y_i=0\}$ and $\mathcal{C}_1\coloneqq \{\sigma_i \vert y_i=1\}$. The conditional probability of outputting a predicted label $\hat{y}(\sigma)=c \in \{0, 1\}$, when given the state $\sigma_i$, can be written as $P(\hat{y}=c|\sigma_j)=\Tr (M_c \sigma_j)$, where we have the binary measurement $\left( M_{0},M_{1}\right) $ satisfying $M_{0}\in\mathcal{M}_{d}$ and $M_{1}=I-M_{0}$. 

Then the common case of $1$-$0$ loss (misclassification cost) is the average error probability 
\begin{align}
    R(M_0)& \coloneqq \frac{1}{m}\left(\sum_{\sigma_i \in C_0}\Tr (M_1 \sigma_i)+\sum_{\sigma_i \in \mathcal{C}_1}\Tr (M_0 \sigma_i)\right) \nonumber \\
   &=p_0 \Tr ((I-M_0)\rho_0)+(1-p_0) \Tr (M_0 \rho_1), \nonumber \\
   \rho_j& =\frac{1}{m_j}\sum_{\sigma_i \in \mathcal{C}_j}\sigma_i, \quad j \in \{0,1\},
\end{align}
where $p_i\coloneqq m_i/m$ and $m_i$ is the number of training states in class $\mathcal{C}_i$, for $i\in \{0,1\}$. This is of the same form as the loss function in~\eqref{eq:rewrite-sym-hypo-test}.

We can proceed to a more general case with a more general cost matrix with entries $C_{i,j}$, where $i,j\in\left\{ 0,1\right\} $. The $1$-$0$ loss example is a special case where $C_{00}=0=C_{11}$, $C_{10}=1=C_{01}$. 
Then the Bayes risk is as follows:
\begin{align}
& R(M_{0}) \notag \\
& \coloneqq\sum_{i,j\in\left\{ 0,1\right\} }C_{j,i}p_{i}\Tr\!\left[M_{j}\rho_{i}\right]\\
 & =p_{0}\left(C_{0,0}\Tr\!\left[M_{0}\rho_{0}\right]+C_{1,0}\Tr\!\left[M_{1}\rho_{0}\right]\right)\nonumber \\
 & \qquad+p_{1}\left(C_{0,1}\Tr\!\left[M_{0}\rho_{1}\right]+C_{1,1}\Tr\!\left[M_{1}\rho_{1}\right]\right)\\
 & =p_{0}\left(C_{0,0}\Tr\!\left[M_{0}\rho_{0}\right]+C_{1,0}\Tr\!\left[\left(I-M_{0}\right)\rho_{0}\right]\right)\nonumber \\
 & \qquad+p_{1}\left(C_{0,1}\Tr\!\left[M_{0}\rho_{1}\right]+C_{1,1}\Tr\!\left[\left(I-M_{0}\right)\rho_{1}\right]\right)\\
 & =p_{0}C_{1,0}+p_{1}C_{1,1}+\Tr\!\left[M_{0}\Gamma\right],
\end{align}
where
\begin{equation}
\Gamma\coloneqq p_{0}\left(C_{0,0}-C_{1,0}\right)\rho_{0}+p_{1}\left(C_{0,1}-C_{1,1}\right)\rho_{1}.
\end{equation}
Given the prior probabilities $p_{0}$ and $p_{1}$, quantum states
$\rho_{0}$ and $\rho_{1}$, and the cost matrix with entries $C_{i,j}$
for $i,j\in\left\{ 0,1\right\} $, the minimum Bayes risk is given
by
\begin{align}
R_{\min} & \coloneqq\min_{M\in\mathcal{M}_{d}}R(M)\\
 & =p_{0}C_{1,0}+p_{1}C_{1,1}+\min_{M\in\mathcal{M}_{d}}\Tr\!\left[M\Gamma\right]\\
 & =p_{0}C_{1,0}+p_{1}C_{1,1}-\Tr\!\left[\left(\Gamma\right)_{-}\right],\label{eq:bayes-cost-final}
\end{align}
where we applied Lemma \ref{lem:neg-trace-neg-part} for the last
equality.

From here, the analysis is similar to that presented in Section \ref{subsec:Symmetric-binary-HT},
and we provide it here for completeness. An optimal measurement operator
for \eqref{eq:bayes-cost-final} has the form in \eqref{eq:optimal-meas-op},
given by
\begin{equation}
\Pi_{\Gamma<0}+N_{0},
\end{equation}
where $0\leq N_{0}\leq\Pi_{\Gamma=0}$. 

For $T>0$, the modified free-energy optimization problem is a special
case of \eqref{eq:primal-FD-obj}, given by:
\begin{multline}
R_{\min,T}\coloneqq p_{0}C_{1,0}+p_{1}C_{1,1}\label{eq:opt-prob-sym-hypo-test-1}\\
+\min_{M\in\mathcal{M}_{d}}\left\{ \Tr\!\left[M\Gamma\right]-T\cdot S_{\FD}(M)\right\} .
\end{multline}
Applying Proposition \ref{prop:dual-FD-free-energy}, we conclude
that
\begin{equation}
R_{\min,T}=p_{0}C_{1,0}+p_{1}C_{1,1}-T\Tr\!\left[\ln\!\left(e^{-\frac{\Gamma}{T}}+I\right)\right],
\end{equation}
and the optimal measurement operator for \eqref{eq:opt-prob-sym-hypo-test-1}
is as follows:
\begin{equation}
M_{T}=\left(e^{\frac{\Gamma}{T}}+I\right)^{-1}.
\end{equation}

In the zero-temperature limit, the following equalities hold:
\begin{align}
\lim_{T\to0^{+}}R_{\min,T} & =R_{\min},\\
\lim_{T\to0^{+}}M_{T} & =\Pi_{\Gamma<0}+\frac{1}{2}\Pi_{\Gamma=0},
\end{align}
which follow from \eqref{eq:temp-T-obj-to-zero-temp-obj} and \eqref{eq:temp-T-meas-op-to-zero-temp-meas-op},
respectively.

In general, the following bound from Proposition \ref{prop:simple-approx-bnd}
is applicable:
\begin{equation}
\left|R_{\min}-R_{\min,T}\right|\leq T\cdot d\ln2.
\end{equation}
We additionally have the following bound from Proposition \ref{prop:approx-error-spectral-gap}:
\begin{equation}
\left|R_{\min}-R_{\min,T}\right|\leq T\left(d_{0}\ln2+\left(d-d_{0}\right)e^{-\frac{\Delta}{T}}\right),\label{eq:hels-hol-err-bound-FD-1}
\end{equation}
where $d_{0}\equiv\dim\ker(q\sigma-p\rho)$, $\Delta\equiv\min_{i:\lambda_{i}\neq0}\left|\lambda_{i}\right|$,
and $\lambda_{i}$ denotes the $i$th eigenvalue of $\Gamma$. By
applying the error bound in \eqref{eq:hels-hol-err-bound-FD-1}, the
temperature $T$ needed to have $\varepsilon$ error, when approximating
$R_{\min}$ by $R_{\min,T}$ and employing a Fermi--Dirac thermal
measurement, is given by \eqref{eq:temp-eps-error-spec-gap}.

\subsection{Asymmetric binary quantum hypothesis testing}

\label{subsec:Asymmetric-binary-HT}

For $\varepsilon\in\left[0,1\right]$, $d\in\mathbb{N}$, and $d\times d$
density matrices $\rho$ and $\sigma$, the minimum type II error
probability in asymmetric hypothesis testing is defined as follows:
\begin{align}
\beta^{\varepsilon}(\rho\|\sigma) & \coloneqq\inf_{M\in\mathcal{M}_{d}}\left\{ \Tr\!\left[M\sigma\right]:\Tr\!\left[M\rho\right]\geq1-\varepsilon\right\} .\label{eq:beta-asym-hypo-t}
\end{align}
This quantity has been considered extensively in the quantum information
theory literature, and it is fundamental for understanding quantum
communication~\cite{Hayashi2003,Buscemi2010,Wang2012,Hayashi2017,khatri2024}
and quantum resource theories~\cite{Chitambar2019,Gour2025} in the
non-asymptotic setting.

Recall from \cite[Appendix~B]{Kaur2017} that
\begin{equation}
\beta^{\varepsilon}(\rho\|\sigma)=\inf_{M\in\mathcal{M}_{d}}\left\{ \Tr\!\left[M\sigma\right]:\Tr\!\left[M\rho\right]=1-\varepsilon\right\} ;\label{eq:beta-asym-hypo-t-eq}
\end{equation}
i.e., we can always choose the measurement operator $M$ in~\eqref{eq:beta-asym-hypo-t}
to saturate the inequality constraint therein. With this, observe
that~\eqref{eq:beta-asym-hypo-t-eq} is a special case of the general
measurement optimization problem in~\eqref{eq:measurement-opt-gen},
from which Proposition~\ref{prop:dual-gen-meas-opt} implies the following
dual representation:
\begin{equation}
\beta^{\varepsilon}(\rho\|\sigma)=\sup_{\mu\geq0}\left\{ \mu\left(1-\varepsilon\right)-\Tr\!\left[\left(\sigma-\mu\rho\right)_{-}\right]\right\} ,\label{eq:dual-asym-hypo-t}
\end{equation}
as observed in \cite[Eqs.~(27)--(28)]{Buscemi2017} (see also \cite[Lemma~2]{VazquezVilar2016}).
For a temperature $T>0$, let us also define the following modified
objective function, which corresponds to a special case of the measurement
free-energy optimization problem in~\eqref{eq:primal-FD-obj}:
\begin{multline}
\beta^{\varepsilon,T}(\rho\|\sigma)\coloneqq\\
\inf_{M\in\mathcal{M}_{d}}\left\{ \Tr\!\left[M\sigma\right]-T\cdot S_{\FD}(M):\Tr\!\left[M\rho\right]=1-\varepsilon\right\} ,\label{eq:free-energy-asym-hypo-t}
\end{multline}
where the Fermi--Dirac entropy $S_{\FD}(M)$ is defined in~\eqref{eq:FD-entropy-def}.
By applying Proposition~\ref{prop:dual-FD-free-energy}, consider
that $\beta^{\varepsilon,T}(\rho\|\sigma)$ has the following dual
representation:
\begin{multline}
\beta^{\varepsilon,T}(\rho\|\sigma)=\\
\sup_{\mu\in\mathbb{R}}\left\{ \mu\left(1-\varepsilon\right)-T\Tr\!\left[\ln\!\left(e^{-\frac{1}{T}\left(\sigma-\mu\rho\right)}+I\right)\right]\right\} ,\label{eq:dual-opt-asymmetric-binary}
\end{multline}
and an optimal measurement operator for~\eqref{eq:free-energy-asym-hypo-t}
is a Fermi--Dirac thermal measurement operator of the following form:
\begin{equation}
M_{T}(\mu)\coloneqq\left(e^{\frac{1}{T}\left(\sigma-\mu\rho\right)}+I\right)^{-1}.
\end{equation}
For fixed $\mu$, the following equality holds for the zero-temperature
limit:
\begin{equation}
\lim_{T\to0^{+}}M_{T}(\mu)=\Pi_{\mu\rho>\sigma}+\frac{1}{2}\Pi_{\mu\rho=\sigma},
\end{equation}
as a special case of~\eqref{eq:temp-T-meas-op-to-zero-temp-meas-op}.

In general, the following bound from Proposition~\ref{prop:simple-approx-bnd}
is applicable:
\begin{equation}
\left|\beta^{\varepsilon}(\rho\|\sigma)-\beta^{\varepsilon,T}(\rho\|\sigma)\right|\leq T\cdot d\ln2.
\end{equation}
We additionally have the following bound from Proposition~\ref{prop:approx-error-spectral-gap}:
\begin{equation}
\left|\beta^{\varepsilon}(\rho\|\sigma)-\beta^{\varepsilon,T}(\rho\|\sigma)\right|\leq T\left(d_{0}\ln2+\left(d-d_{0}\right)e^{-\frac{\Delta}{T}}\right),
\end{equation}
where
\begin{align}
d_{0} & \geq\sup_{\mu\in\mathbb{M}}\dim\ker(\sigma-\mu\rho),\\
\Delta & \leq\inf_{\mu\in\mathbb{M}}\min_{i:\lambda_{i}\neq0}\left|\lambda_{i}(\mu)\right|,
\end{align}
$\lambda_{i}(\mu)$ denotes the $i$th eigenvalue of $\sigma-\mu\rho$,
and $\mathbb{M}$ is a set containing $\mu^{\star}$ and $\mu_{T}^{\star}$,
which are optimal choices for~\eqref{eq:dual-asym-hypo-t} and~\eqref{eq:dual-opt-asymmetric-binary},
respectively. In this case, the temperature $T$ needed to have $\delta$
error, when approximating $\beta^{\varepsilon}(\rho\|\sigma)$ by
$\beta^{\varepsilon,T}(\rho\|\sigma)$ and employing a Fermi--Dirac
thermal measurement, is given by~\eqref{eq:temp-eps-error-spec-gap}.

For optimizing~\eqref{eq:dual-opt-asymmetric-binary}, it is necessary
to search for an optimal choice of $\mu$. Since the problem we are
considering is a special case of~\eqref{eq:dual-free-energy-meas-opt},
it can be solved by gradient ascent, as outlined in Section~\ref{subsec:Gradient-ascent-gen-meas-opt}.
Applying Proposition~\ref{prop:gradient-FD-gen-obj-func}, the gradient
of the dual objective function in~\eqref{eq:dual-opt-asymmetric-binary}
has the following form:
\begin{multline}
\frac{\partial}{\partial\mu}\left(\mu\left(1-\varepsilon\right)-T\Tr\!\left[\ln\!\left(e^{-\frac{1}{T}\left(\sigma-\mu\rho\right)}+I\right)\right]\right)\\
=1-\varepsilon-\Tr\!\left[M_{T}(\mu)\rho\right].
\end{multline}
Thus we can search for the optimal measurement operator by performing
gradient ascent \cite[Chapter~3]{Bubeck2015}:
\begin{equation}
\mu_{k+1}\leftarrow\mu_{k}+\eta\left(1-\varepsilon-\Tr\!\left[M_{T}(\mu)\rho\right]\right),
\end{equation}
where $\eta\in\left(0,T\right]$. This choice of step size follows
from the bound in~\eqref{eq:smoothness-parameter} and the fact that
$\left\Vert \rho\right\Vert \left\Vert \rho\right\Vert _{1}\leq1$
for a quantum state $\rho$. This simple algorithm is guaranteed to
converge after $\frac{\left|\mu_{T}^{\star}\right|}{T\delta}$ steps,
where $\mu_{T}^{\star}$ is an optimal solution to~\eqref{eq:dual-opt-asymmetric-binary}
and $\delta>0$ is the desired error.

As discussed in Section~\ref{sec:Quantum-algorithm-gen-meas-opt},
one can also run this algorithm as a hybrid quantum--classical algorithm
if sample access to the states $\rho$ and $\sigma$ is available.
In this case, it can be viewed as a learning algorithm, in which the
optimal choice of $\mu$ is learned in an iterative manner.

\subsection{Asymmetric quantum hypothesis testing with composite null hypothesis}

\label{subsec:Asymmetric-HT-composite-null}

In asymmetric quantum hypothesis testing with composite null hypothesis,
there are multiple null hypotheses, which we denote by
\begin{equation}
\mathcal{R}\equiv\left(\rho_{1},\ldots,\rho_{c}\right),
\end{equation}
where $\rho_{i}$ is a quantum state, as well as multiple constraints
on the error probabilities, which we denote by
\begin{equation}
\mathcal{E}\equiv\left(\varepsilon_{1},\ldots,\varepsilon_{c}\right),
\end{equation}
where $\varepsilon_{i}\in\left[0,1\right]$.

Then the asymmetric composite hypothesis testing problem has the following
form:
\begin{equation}
\beta^{\varepsilon}(\mathcal{R}\|\sigma)\coloneqq\inf_{M\in\mathcal{M}_{d}}\left\{ \Tr\!\left[M\sigma\right]:\Tr\!\left[M\rho_{i}\right]\geq1-\varepsilon_{i}\,\forall i\in\left[c\right]\right\} 
\end{equation}
Applying~\eqref{eq:dual-gen-meas-opt-1}, this quantity has the following
dual representation:
\begin{equation}
\beta^{\varepsilon}(\mathcal{R}\|\sigma)=\sup_{\mu\in\mathbb{R}_{+}^{c}}\left\{ \sum_{i\in\left[c\right]}\mu_{i}\left(1-\varepsilon_{i}\right)-\Tr\!\left[\left(\sigma-\mu\cdot\rho\right)_{-}\right]\right\} .
\end{equation}
(Note that $\mu$ now denotes a vector in $\mathbb{R}_{+}^{c}$.)
For a temperature $T>0$, let us also define the following modified
objective function, which corresponds to a special case of the measurement
free-energy optimization problem in~\eqref{eq:primal-FD-obj}:
\begin{equation}
\beta^{\varepsilon,T}(\mathcal{R}\|\sigma)\coloneqq
\inf_{M\in\mathcal{M}_{d}}\left\{
\begin{array}{c}
\Tr\!\left[M\sigma\right]-T\cdot S_{\FD}(M):\\
\Tr\!\left[M\rho_{i}\right]\geq1-\varepsilon_{i}\,\forall i\in\left[c\right]
\end{array}
\right\} ,\label{eq:free-energy-asym-hypo-t-1}
\end{equation}
where the Fermi--Dirac entropy $S_{\FD}(M)$ is defined in~\eqref{eq:FD-entropy-def}.
By applying Proposition~\ref{prop:dual-FD-free-energy}, consider
that $\beta^{\varepsilon,T}(\rho\|\sigma)$ has the following dual
representation:
\begin{multline}
\beta^{\varepsilon,T}(\mathcal{R}\|\sigma)=\\
\sup_{\mu\in\mathbb{R}_{+}^{c}}\left\{ \mu\left(1-\varepsilon\right)-T\Tr\!\left[\ln\!\left(e^{-\frac{1}{T}\left(\sigma-\mu\cdot\mathcal{R}\right)}+I\right)\right]\right\} ,\label{eq:dual-opt-asymmetric-binary-1}
\end{multline}
where
\begin{equation}
\mu\cdot\mathcal{R}\equiv\sum_{i\in\left[c\right]}\mu_{i}\rho_{i},
\end{equation}
and an optimal measurement operator for~\eqref{eq:free-energy-asym-hypo-t-1}
is a Fermi--Dirac thermal measurement operator of the following form:
\begin{equation}
M_{T}(\mu)\coloneqq\left(e^{\frac{1}{T}\left(\sigma-\mu\cdot\mathcal{R}\right)}+I\right)^{-1}.
\end{equation}
For fixed $\mu$, the following equality holds for the zero-temperature
limit:
\begin{equation}
\lim_{T\to0^{+}}M_{T}(\mu)=\Pi_{\mu\cdot\mathcal{R}>\sigma}+\frac{1}{2}\Pi_{\mu\cdot\mathcal{R}=\sigma},
\end{equation}
as a special case of~\eqref{eq:temp-T-meas-op-to-zero-temp-meas-op}.

In general, the following bound from Proposition~\ref{prop:simple-approx-bnd}
is applicable:
\begin{equation}
\left|\beta^{\varepsilon}(\mathcal{R}\|\sigma)-\beta^{\varepsilon,T}(\mathcal{R}\|\sigma)\right|\leq T\cdot d\ln2.
\end{equation}
We additionally have the following bound from Proposition~\ref{prop:approx-error-spectral-gap}:
\begin{equation}
\left|\beta^{\varepsilon}(\mathcal{R}\|\sigma)-\beta^{\varepsilon,T}(\mathcal{R}\|\sigma)\right|\leq T\left(d_{0}\ln2+\left(d-d_{0}\right)e^{-\frac{\Delta}{T}}\right),
\end{equation}
where
\begin{align}
d_{0} & \geq\sup_{\mu\in\mathbb{M}}\dim\ker(\sigma-\mu\rho),\\
\Delta & \leq\inf_{\mu\in\mathbb{M}}\min_{i:\lambda_{i}\neq0}\left|\lambda_{i}(\mu)\right|,
\end{align}
$\lambda_{i}(\mu)$ denotes the $i$th eigenvalue of $\sigma-\mu\cdot\mathcal{R}$,
and $\mathbb{M}$ is a set containing $\mu^{\star}$ and $\mu_{T}^{\star}$,
which are optimal choices for~\eqref{eq:dual-asym-hypo-t} and~\eqref{eq:dual-opt-asymmetric-binary-1},
respectively. In this case, the temperature $T$ needed to have $\delta$
error, when approximating $\beta^{\varepsilon}(\mathcal{R}\|\sigma)$
by $\beta^{\varepsilon,T}(\mathcal{R}\|\sigma)$ and employing a Fermi--Dirac
thermal measurement, is given by~\eqref{eq:temp-eps-error-spec-gap}.

For optimizing~\eqref{eq:dual-opt-asymmetric-binary-1}, it is necessary
to search for an optimal choice of $\mu$. Since the problem we are
considering is a special case of~\eqref{eq:dual-free-energy-meas-opt},
it can be solved by gradient ascent, as outlined in Section~\ref{subsec:Gradient-ascent-gen-meas-opt}.
Applying Proposition~\ref{prop:gradient-FD-gen-obj-func}, the gradient
elements of the dual objective function in~\eqref{eq:dual-opt-asymmetric-binary-1}
have the following form:
\begin{multline}
\frac{\partial}{\partial\mu_{i}}\left(\sum_{i\in\left[c\right]}\mu_{i}\left(1-\varepsilon_{i}\right)-T\Tr\!\left[\ln\!\left(e^{-\frac{1}{T}\left(\sigma-\mu\cdot\mathcal{R}\right)}+I\right)\right]\right)\\
=1-\varepsilon_{i}-\Tr\!\left[M_{T}(\mu_{i})\rho_{i}\right].
\end{multline}
Thus we can search for the optimal measurement operator by performing
projected gradient ascent \cite[Chapter~3]{Bubeck2015}:
\begin{equation}
\mu_{i}^{k+1}\leftarrow\max\left\{ \mu_{i}^{k}+\eta\left(1-\varepsilon_{i}-\Tr\!\left[M_{T}(\mu_{i})\rho_{i}\right]\right),0\right\} .
\end{equation}
where $\eta\in\left(0,T/c\right]$. This choice of step size follows
from the bound in~\eqref{eq:smoothness-parameter} and the fact that
$\left\Vert \rho_{i}\right\Vert \left\Vert \rho_{i}\right\Vert _{1}\leq1$
for each quantum state $\rho_{i}$. This simple algorithm is guaranteed
to converge after $\frac{\left\Vert \mu_{T}^{\star}\right\Vert }{T\delta}$
steps, where $\mu_{T}^{\star}$ is an optimal solution to~\eqref{eq:dual-opt-asymmetric-binary-1}
and $\delta>0$ is the desired error.

As discussed in Section~\ref{sec:Quantum-algorithm-gen-meas-opt},
one can also run this algorithm as a hybrid quantum--classical algorithm
if sample access to the states $\rho_{1}$, $\ldots$, $\rho_{c}$,
$\sigma$ is available. In this case, it can be viewed as a learning
algorithm, in which the optimal choice of $\mu$ is learned in an
iterative manner.

\section{Conclusion}

\label{sec:Conclusion}

In this paper, we developed a novel paradigm for semidefinite optimization
on quantum computers, based on measurement optimization rather than
state optimization. In particular, we showed how a free-energy approximation
of a measurement optimization problem leads to Fermi--Dirac thermal
measurements being optimal for the approximation. We also proved that
the dual function $f_{T}(\mu)$ for the free-energy approximation
is a concave and smooth function of the dual variables in $\mu$,
implying that gradient ascent and its variants are guaranteed to converge
to a globally optimal solution. We also provided quantum algorithms
for implementing Fermi--Dirac thermal measurements and for estimating
Hessian matrix elements for the dual function $f_{T}(\mu)$, the first
of which is essential in hybrid quantum--classical algorithms that
perform the first-order optimization needed in a measurement optimization
problem. In our analysis of the error in the free-energy approximation,
we developed sufficient conditions for when optimization algorithms
are efficiently realizable as hybrid quantum--classical algorithms. 

Viewed in another way, our paper introduces a novel paradigm for quantum machine learning called \textit{Fermi--Dirac machines}, in which the parameters of a Fermi--Dirac thermal measurement can be learned in gradient-based, hybrid quantum--classical algorithms. The paradigm of Fermi--Dirac machines is an alternative to quantum Boltzmann machines (in particular for the case of decision problems), the latter being based on parameterized thermal states.

Going forward from here, there are several open directions to pursue. A pressing question is to determine a particular
example of a measurement optimization problem for which we can expect
a quantum computer to be efficient while a classical computer would
not be. If found, such a scenario would allow for quantum advantage
to be demonstrated in the general setting of measurement optimization.
We also think it would be worthwhile to perform numerical simulations
of the algorithms put forward here and execute them on existing quantum
computers, in order to understand their performance in realistic scenarios.
We wonder if there are other, more efficient and more natural quantum algorithms for implementing Fermi--Dirac thermal measurements beyond the approach proposed in Section~\ref{subsec:Quantum-algorithm-FD-thermal}. 
Related to this, it is interesting to speculate if  Fermi--Dirac thermal measurements could also be useful for understanding problems in quantum computational complexity theory, similar to how quantum thermal states have been helpful for this purpose \cite{Bravyi2022,gharibian2026}.

Since the original arXiv posting of our paper, we have shown, with our collaborators, how Fermi--Dirac thermal measurements lead to a novel formulation of a quantum neuron~\cite{he2026fermidirac,he2026canonical}, how they are useful for principal component analysis~\cite{yuan2026} and spectral anomaly detection~\cite{yuan2026qspade}, and how they are optimal measurements for a minimum change principle~\cite{liu2026unifying}.

\medskip{}

\textbf{Acknowledgements}. We acknowledge helpful discussions with
Zixin Huang and Michele Minervini. NL acknowledges funding from the
Science and Technology Commission of Shanghai Municipality (STCSM)
grant no.~24LZ1401200 (21JC1402900), NSFC grants no.~12471411 and
no.~12341104, the Shanghai Jiao Tong University 2030 Initiative, the Shanghai Pilot Program for Basic Research, 
and the Fundamental Research Funds for the Central Universities. MMW
acknowledges support from the National Science Foundation under grant
no.~2329662 and the Cornell School of Electrical and Computer Engineering.

\bibliography{Ref}

\appendix

\section{Dual optimization problems}

\subsection{Proof of Proposition~\ref{prop:dual-gen-meas-opt} (dual of measurement
optimization problem)}

\label{sec:Proof-of-dual-gen-meas-opt-exp}In this appendix, we prove
Theorem~\ref{prop:dual-gen-meas-opt}. Consider that
\begin{align}
 & E(\mathcal{Q},q)\nonumber \\
 & =\min_{M\in\mathcal{M}_{d}}\left\{ \Tr\!\left[HM\right]:\Tr\!\left[Q_{i}M\right]=q_{i}\,\forall i\in\left[c\right]\right\} \\
 & \overset{(a)}{=}\min_{M\in\mathcal{M}_{d}}\left\{ \Tr\!\left[HM\right]+\sup_{\mu\in\mathbb{R}^{c}}\sum_{i\in\left[c\right]}\mu_{i}\left(q_{i}-\Tr\!\left[Q_{i}M\right]\right)\right\} \\
 & =\min_{M\in\mathcal{M}_{d}}\sup_{\mu\in\mathbb{R}^{c}}\left\{ \Tr\!\left[HM\right]+\sum_{i\in\left[c\right]}\mu_{i}\left(q_{i}-\Tr\!\left[Q_{i}M\right]\right)\right\} \\
 & =\min_{M\in\mathcal{M}_{d}}\sup_{\mu\in\mathbb{R}^{c}}\left\{ \mu\cdot q+\Tr\!\left[\left(H-\mu\cdot Q\right)M\right]\right\} \\
 & \overset{(b)}{=}\sup_{\mu\in\mathbb{R}^{c}}\min_{M\in\mathcal{M}_{d}}\left\{ \mu\cdot q+\Tr\!\left[\left(H-\mu\cdot Q\right)M\right]\right\} \\
 & =\sup_{\mu\in\mathbb{R}^{c}}\left\{ \mu\cdot q+\min_{M\in\mathcal{M}_{d}}\Tr\!\left[\left(H-\mu\cdot Q\right)M\right]\right\} \label{eq:sup-and-min-over-meas}\\
 & \overset{(c)}{=}\sup_{\mu\in\mathbb{R}^{c}}\left\{ \mu\cdot q-\Tr\!\left[\left(H-\mu\cdot Q\right)_{-}\right]\right\} \\
 & \overset{(d)}{=}\sup_{\mu\in\mathbb{R}^{c},X\geq0}\left\{ \begin{array}{c}
\mu\cdot q+\Tr\!\left[H-\mu\cdot Q\right]-\Tr[X]:\\
X\geq H-\mu\cdot Q
\end{array}\right\} .
\end{align}
The equality (a) follows from introducing the Lagrange multiplier
$\mu_{i}\in\mathbb{R}$ for the $i$th constraint. The equality (b)
follows from the minimax theorem from \cite[Theorem~2.11]{BenDavid2023}.
The equalities (c) and (d) follow from Lemma~\ref{lem:neg-trace-neg-part}
below.

By picking $A=H-\mu\cdot Q$ in Lemma~\ref{lem:neg-trace-neg-part}
below, we observe that $M^{\star}$ in~\eqref{eq:optimal-measurement-operator}
is an optimal choice of measurement for the minimization in~\eqref{eq:sup-and-min-over-meas}.
The proof that an optimal measurement takes the form in~\eqref{eq:optimal-measurement-operator}
is similar to the proof of \cite[Eq.~(5.3.17)]{khatri2024}.

Finally, the function
\begin{equation}
\mu\mapsto\mu\cdot q-\Tr\!\left[\left(H-\mu\cdot Q\right)_{-}\right]
\end{equation}
is concave in $\mu$ because $\mu\cdot q$ is linear in $\mu$, and
the function $\mu\mapsto-\Tr\!\left[\left(H-\mu\cdot Q\right)_{-}\right]$
is concave, given that is equal to the minimum of functions that are
each linear in $\mu$, as established in~\eqref{eq:sup-and-min-over-meas}.
\begin{lem}
\label{lem:neg-trace-neg-part}For a $d\times d$ Hermitian matrix
$A$, where $d\in\mathbb{N}$, the following equalities hold:
\begin{align}
-\Tr\!\left[\left(A\right)_{-}\right] & =\min_{M\in\mathcal{M}_{d}}\Tr\!\left[MA\right],\label{eq:negative-eigenspace-projection}\\
 & =\sup_{X\geq0,X\geq A}\left\{ \Tr[A]-\Tr[X]\right\} .\label{eq:dual-gen-meas-opt-SDP-app}
\end{align}
\end{lem}

\begin{proof}
To see~\eqref{eq:negative-eigenspace-projection}, let $A=\left(A\right)_{+}-\left(A\right)_{-}$
be a Jordan--Hahn decomposition of $A$ into positive and negative
parts, such that $\left(A\right)_{+}\left(A\right)_{-}=0$ and $\left(A\right)_{+},\left(A\right)_{-}\geq0$.
That is,
\begin{equation}
\left(A\right)_{+}\coloneqq\sum_{i:a_{i}\geq0}a_{i}|i\rangle\!\langle i|,\qquad\left(A\right)_{-}\coloneqq\sum_{i:a_{i}<0}\left|a_{i}\right||i\rangle\!\langle i|,\label{eq:jordan-hahn-decomp}
\end{equation}
where $A=\sum_{i}a_{i}|i\rangle\!\langle i|$ is an eigendecomposition
of $A$. Let us also define the projectors:
\begin{equation}
\Pi_{+}\coloneqq\sum_{i:a_{i}\geq0}|i\rangle\!\langle i|,\qquad\Pi_{-}\coloneqq\sum_{i:a_{i}<0}|i\rangle\!\langle i|.\label{eq:jordan-hahn-projs}
\end{equation}
Then consider that
\begin{align}
\Tr\!\left[MA\right] & =\Tr\!\left[M\left(\left(A\right)_{+}-\left(A\right)_{-}\right)\right]\\
 & =\Tr\!\left[M\left(A\right)_{+}\right]-\Tr\!\left[M\left(A\right)_{-}\right]\\
 & \overset{(a)}{\geq}-\Tr\!\left[M\left(A\right)_{-}\right]\\
 & \overset{(b)}{\geq}-\Tr\!\left[\left(A\right)_{-}\right].
\end{align}
The inequality (a) follows because $\Tr\!\left[M\left(A\right)_{+}\right]\geq0$,
given that $M,\left(A\right)_{+}\geq0$. The inequality (b) follows
because $M\leq I$. By choosing $M$ to be the projection onto the
negative eigenspace of $A$, the established inequality
\begin{equation}
\Tr\!\left[MA\right]\geq-\Tr\!\left[\left(A\right)_{-}\right]
\end{equation}
is saturated. That is, we can choose $M=\Pi_{-}$, and we find that
\begin{equation}
\Tr\!\left[\Pi_{-}A\right]=-\Tr\!\left[\left(A\right)_{-}\right].
\end{equation}
This follows because $\Pi_{-}\left(A\right)_{+}=0$ and $\Pi_{-}\left(A\right)_{-}=\left(A\right)_{-}$.
In fact, we can choose a measurement operator of the following form
\begin{equation}
M^{\star}\coloneqq\Pi_{-}+M_{0},\label{eq:optimal-measurement-operator}
\end{equation}
where $0\leq M_{0}\leq\Pi_{0}$ and
\begin{equation}
\Pi_{0}\coloneqq\sum_{i:a_{i}=0}|i\rangle\!\langle i|,
\end{equation}
and the inequality is still saturated for such a choice. This follows
because $M^{\star}\left(A\right)_{+}=0$ and $M^{\star}\left(A\right)_{-}=\left(A\right)_{-}$.

To establish~\eqref{eq:dual-gen-meas-opt-SDP-app}, consider that
\begin{equation}
-\Tr\!\left[\left(A\right)_{-}\right]=\sup_{X\geq0,X\geq A}\left\{ \Tr[A]-\Tr[X]\right\} .
\end{equation}
To see this, let us use the same notation as in~\eqref{eq:jordan-hahn-decomp}
and~\eqref{eq:jordan-hahn-projs}. Consider that the constraint $X\geq0$
implies that $\Pi_{-}X\Pi_{-}\geq0$, which in turn implies that
\begin{equation}
\Tr\!\left[\Pi_{-}X\right]\geq0.\label{eq:X-constraint-1}
\end{equation}
The constraint $X\geq A$ implies that
\begin{align}
\Pi_{+}X\Pi_{+} & \geq\Pi_{+}A\Pi_{+}\\
 & =\Pi_{+}\left(\left(A\right)_{+}-\left(A\right)_{-}\right)\Pi_{+}\\
 & =\Pi_{+}\left(A\right)_{+}\Pi_{+}\\
 & =\left(A\right)_{+},
\end{align}
which in turn implies that
\begin{equation}
\Tr\!\left[\Pi_{+}X\right]\geq\Tr\!\left[\left(A\right)_{+}\right].\label{eq:X-constraint-2}
\end{equation}
Adding~\eqref{eq:X-constraint-1} and~\eqref{eq:X-constraint-2} then
gives
\begin{equation}
\Tr[X]=\Tr\!\left[\Pi_{-}X\right]+\Tr\!\left[\Pi_{+}X\right]\geq\Tr\!\left[\left(A\right)_{+}\right].
\end{equation}
Then the following inequality holds
\begin{align}
 & \Tr[A]-\Tr[X]\nonumber \\
 & =\Tr\!\left[\left(A\right)_{+}\right]-\Tr\!\left[\left(A\right)_{-}\right]-\Tr[X]\\
 & \leq-\Tr\!\left[\left(A\right)_{-}\right],
\end{align}
thus implying that
\begin{equation}
\sup_{X\geq0,X\geq A}\left\{ \Tr[A]-\Tr[X]\right\} \leq-\Tr\!\left[\left(A\right)_{-}\right].\label{eq:ineq-proof-SDP-neg-part}
\end{equation}
Equality is achieved in~\eqref{eq:ineq-proof-SDP-neg-part} by picking
$X=\left(A\right)_{+}$ and noting that it is feasible. 
\end{proof}

\subsection{First proof of Proposition~\ref{prop:dual-FD-free-energy} (dual
of measurement free-energy optimization problem)}

\label{sec:Proof-of-Theorem-FD-thermal-meas-gen}

Consider that
\begin{align}
 & F_{T}(\mathcal{Q},q)\nonumber \\
 & =\min_{M\in\mathcal{M}_{d}}\left\{
 \begin{array}{c}\Tr\!\left[HM\right]-T\cdot S_{\FD}(M):\\
 \Tr\!\left[Q_{i}M\right]=q_{i}\,\forall i\in\left[c\right]
 \end{array}
 \right\} \\
 & \overset{(a)}{=}\min_{M\in\mathcal{M}_{d}}\left\{ \begin{array}{c}
\Tr\!\left[HM\right]-T\cdot S_{\FD}(M)\\
+\sup_{\mu\in\mathbb{R}^{c}}\sum_{i\in\left[c\right]}\mu_{i}\left(q_{i}-\Tr\!\left[Q_{i}M\right]\right)
\end{array}\right\} \\
 & =\min_{M\in\mathcal{M}_{d}}\sup_{\mu\in\mathbb{R}^{c}}\left\{ \begin{array}{c}
\Tr\!\left[HM\right]-T\cdot S_{\FD}(M)\\
+\sum_{i\in\left[c\right]}\mu_{i}\left(q_{i}-\Tr\!\left[Q_{i}M\right]\right)
\end{array}\right\} \\
 & =\min_{M\in\mathcal{M}_{d}}\sup_{\mu\in\mathbb{R}^{c}}\left\{
 \begin{array}{c}
 \mu\cdot q+\Tr\!\left[\left(H-\mu\cdot Q\right)M\right]\\-T\cdot S_{\FD}(M)
 \end{array}
 \right\} \\
 & \overset{(b)}{=}\sup_{\mu\in\mathbb{R}^{c}}\min_{M\in\mathcal{M}_{d}}\left\{
 \begin{array}{c}
 \mu\cdot q+\Tr\!\left[\left(H-\mu\cdot Q\right)M\right]\\
 -T\cdot S_{\FD}(M)
 \end{array}\right\} \\
 & =\sup_{\mu\in\mathbb{R}^{c}}\left\{ 
 \mu\cdot q+\min_{M\in\mathcal{M}_{d}}\left\{
 \begin{array}{c}
 \Tr\!\left[\left(H-\mu\cdot Q\right)M\right]\\-T\cdot S_{\FD}(M)
 \end{array}
 \right\} \right\} \\
 & \overset{(c)}{=}\sup_{\mu\in\mathbb{R}^{c}}\left\{ \mu\cdot q-T\Tr\!\left[\ln\!\left(e^{-\frac{1}{T}\left(H-\mu\cdot Q\right)}+I\right)\right]\right\} .
\end{align}
The equality (a) follows by introducing the Lagrange multiplier $\mu_{i}$
for the constraint $\Tr\!\left[Q_{i}M\right]=q_{i}$, for all $i\in\left[c\right]$.
The equality (b) follows from the minimax theorem from \cite[Theorem~2.11]{BenDavid2023};
indeed, the set $\mathcal{M}_{d}$ is convex and compact, the set
$\mathbb{R}^{c}$ is convex, the objective function is linear in $\mu$,
and it is convex in $M$, due to the concavity of $S_{\FD}(M)$, as
proved in~\eqref{eq:concavity-FD-entr-1}--\eqref{eq:concavity-FD-entr-last}.
The equality (c) follows from Lemma~\ref{lem:opt-FD-free-energy}
below.
\begin{lem}
\label{lem:opt-FD-free-energy}Let $A$ be a $d\times d$ Hermitian
matrix, and let $T>0$. Then
\begin{equation}
\min_{M\in\mathcal{M}_{d}}g(A,T,M)=-T\Tr\!\left[\ln\!\left(e^{-\frac{A}{T}}+I\right)\right],\label{eq:FD-meas-opt-prob}
\end{equation}
where
\begin{equation}
g(A,T,M)\coloneqq\Tr\!\left[AM\right]-T\cdot S_{\FD}(M).
\end{equation}
Furthermore, $M=\left(e^{\frac{A}{T}}+I\right){}^{-1}$ is an optimal
choice for the left-hand side of~\eqref{eq:FD-meas-opt-prob}.
\end{lem}

\begin{proof}
Let us first note that the function $g(A,T,M)$ is convex in $M$,
due to the concavity of $M\mapsto S_{\FD}(M)$. Thus, the first-order
optimality conditions for $g(A,T,M)$ are necessary and sufficient
for optimality. To this end, consider that the matrix derivative of
this function is given by
\begin{align}
 & \frac{\partial}{\partial M}\left[\Tr\!\left[AM\right]-T\cdot S_{\FD}(M)\right]\nonumber \\
 & =A-T\frac{\partial}{\partial M}S_{\FD}(M)\\
 & =A-T\left(-\ln M+\ln\!\left(I-M\right)\right),
\end{align}
where we used the fact that the scalar derivative of the function
$m\mapsto s(m)$ is 
\begin{align}
\frac{\partial}{\partial m}s(m) & =\frac{\partial}{\partial m}\left(-m\ln m-\left(1-m\right)\ln\!\left(1-m\right)\right)\\
 & =-\ln m-1+\ln\!\left(1-m\right)+1\\
 & =-\ln m+\ln\!\left(1-m\right).
\end{align}
Then the first-order optimality condition corresponds to the following
matrix equation:
\begin{align}
0 & =A-T\left(-\ln M+\ln\!\left(I-M\right)\right)\\
\Longleftrightarrow\qquad\frac{A}{T} & =-\ln M+\ln\!\left(I-M\right)\\
\Longleftrightarrow\qquad\frac{A}{T} & =\ln\!\left(M^{-1}-I\right)\\
\Longleftrightarrow\qquad e^{\frac{A}{T}}+I & =M^{-1}\\
\Longleftrightarrow\qquad M & =\left(e^{\frac{A}{T}}+I\right){}^{-1},\label{eq:optimal-meas-op-FD}
\end{align}
so that the measurement operator $M$ in~\eqref{eq:optimal-meas-op-FD}
is optimal for the left-hand side of~\eqref{eq:FD-meas-opt-prob}.

Now plugging this choice of $M$ into $g(A,T,M)$, we find that
\begin{multline}
g(A,T,M)=\Tr\!\left[A\left(e^{\frac{A}{T}}+I\right){}^{-1}\right]\\
+T\Tr\!\left[\left(e^{\frac{A}{T}}+I\right){}^{-1}\ln\left(\left(e^{\frac{A}{T}}+I\right){}^{-1}\right)\right]\\
+T\Tr\!\left[\begin{array}{c}
\left(I-\left(e^{\frac{A}{T}}+I\right){}^{-1}\right)\times\\
\ln\left(I-\left(e^{\frac{A}{T}}+I\right){}^{-1}\right)
\end{array}\right].
\end{multline}
To simplify the above expression, it actually suffices to do so for
the scalar case. To this end, for $a\in\mathbb{R}$ and $T>0$, consider
that
\begin{align}
 & \frac{a}{e^{\frac{a}{T}}+1}+\frac{T}{e^{\frac{a}{T}}+1}\ln\!\left(\frac{1}{e^{\frac{a}{T}}+1}\right)\nonumber \\
 & \qquad+T\left(1-\frac{1}{e^{\frac{a}{T}}+1}\right)\ln\!\left(1-\frac{1}{e^{\frac{a}{T}}+1}\right)\nonumber \\
 & =\frac{a}{e^{\frac{a}{T}}+1}+\frac{T}{e^{\frac{a}{T}}+1}\ln\!\left(\frac{1}{e^{\frac{a}{T}}+1}\right)\nonumber \\
 & \qquad+T\frac{e^{\frac{a}{T}}}{e^{\frac{a}{T}}+1}\ln\!\left(\frac{e^{\frac{a}{T}}}{e^{\frac{a}{T}}+1}\right)\\
 & =\frac{a}{e^{\frac{a}{T}}+1}-\frac{T}{e^{\frac{a}{T}}+1}\ln\!\left(e^{\frac{a}{T}}+1\right)\nonumber \\
 & \qquad+T\frac{e^{\frac{a}{T}}}{e^{\frac{a}{T}}+1}\left(\frac{a}{T}-\ln\!\left(e^{\frac{a}{T}}+1\right)\right)\\
 & =a-\frac{T}{e^{\frac{a}{T}}+1}\ln\!\left(e^{\frac{a}{T}}+1\right)-T\frac{e^{\frac{a}{T}}}{e^{\frac{a}{T}}+1}\ln\!\left(e^{\frac{a}{T}}+1\right)\\
 & =a-T\ln\!\left(e^{\frac{a}{T}}+1\right)\\
 & =T\ln e^{\frac{a}{T}}-T\ln\!\left(e^{\frac{a}{T}}+1\right)\\
 & =-T\ln\!\left(\frac{e^{\frac{a}{T}}+1}{e^{\frac{a}{T}}}\right)\\
 & =-T\ln\!\left(1+e^{-\frac{a}{T}}\right).
\end{align}
Thus we conclude that
\begin{equation}
g(A,T,M)=-T\Tr\!\left[\ln\!\left(e^{-\frac{A}{T}}+I\right)\right],
\end{equation}
as claimed.
\end{proof}

\subsection{Second proof of Proposition~\ref{prop:dual-FD-free-energy} (dual
of measurement free-energy optimization problem)}

\label{sec:Second-Proof-dual-FD-meas-op}
\begin{defn}[Fermi--Dirac relative entropy]
The Fermi--Dirac relative entropy of measurement operators $M_{1},M_{2}\in\mathcal{M}_{d}$
is defined as follows:
\begin{align}
 & D_{\FD}(M_{1}\|M_{2})\nonumber \\
 & \coloneqq D(M_{1}\|M_{2})+D(I-M_{1}\|I-M_{2})\\
 & =\Tr\!\left[M_{1}\left(\ln M_{1}-\ln M_{2}\right)\right]\nonumber \\
 & \qquad+\Tr\!\left[\left(I-M_{1}\right)\left(\ln\!\left(I-M_{1}\right)-\ln\!\left(I-M_{2}\right)\right)\right].
\end{align}
\end{defn}

We can also write
\begin{multline}
D_{\FD}(M_{1}\|M_{2})=-S_{\FD}(M_{1})-\Tr\!\left[M_{1}\ln M_{2}\right]\\
-\Tr\!\left[\left(I-M_{1}\right)\left(\ln\!\left(I-M_{2}\right)\right)\right].
\end{multline}
Defining
\begin{equation}
\widetilde{M}_{i}\coloneqq|1\rangle\!\langle1|\otimes M_{i}+|0\rangle\!\langle0|\otimes\left(I-M_{i}\right),
\end{equation}
observe that
\begin{align}
D_{\FD}(M_{1}\|M_{2}) & =D(\widetilde{M}_{1}\|\widetilde{M}_{2}).\\
 & =-S(\widetilde{M}_{1})-\Tr\!\left[\widetilde{M}_{1}\ln\widetilde{M}_{2}\right].
\end{align}
We can also define density operators from $\widetilde{M}_{i}$ as
\begin{equation}
\rho_{i}\coloneqq\frac{\widetilde{M}_{i}}{d},
\end{equation}
from which we observe that
\begin{equation}
D(\rho_{1}\|\rho_{2})=d\cdot D(\widetilde{M}_{1}\|\widetilde{M}_{2})=d\cdot D_{\FD}(M_{1}\|M_{2}).\label{eq:FD-rel-ent-in-terms-of-states}
\end{equation}

\begin{lem}
The Fermi--Dirac relative entropy of two measurement operators is
proportional to the standard quantum relative entropy of the Choi
states of the corresponding measurement channels. In more detail,
let us define the measurement channel $\mathcal{M}_{i}$, for $i\in\left\{ 1,2\right\} $,
as
\begin{equation}
\mathcal{M}_{i}(\rho)\coloneqq\Tr[M_{i}\rho]|0\rangle\!\langle0|+\Tr[(I-M_{i})\rho]|1\rangle\!\langle1|.
\end{equation}
The Choi state of this channel is defined as
\begin{equation}
\left(\id\otimes\mathcal{M}_{i}\right)\left(\Phi_{d}\right),
\end{equation}
where the maximally entangled state $\Phi_{d}$ is defined as
\begin{equation}
\Phi_{d}\coloneqq\frac{1}{d}\sum_{i,j=0}^{d-1}|i\rangle\!\langle j|\otimes|i\rangle\!\langle j|.
\end{equation}
Then 
\begin{equation}
D\!\left(\left(\id\otimes\mathcal{M}_{1}\right)\left(\Phi_{d}\right)\|\left(\id\otimes\mathcal{M}_{2}\right)\left(\Phi_{d}\right)\right)=\frac{1}{d}D_{\FD}(M_{1}\|M_{2}).\label{eq:rel-ent-FD-meas}
\end{equation}
\end{lem}

\begin{proof}
Before establishing this equality, consider that the following identity
holds:
\begin{equation}
\left(\id\otimes\mathcal{M}_{i}\right)\left(\Phi_{d}\right)=\frac{1}{d}\left(M_{i}^{T}\otimes|0\rangle\!\langle0|+(I-M_{i})^{T}\otimes|1\rangle\!\langle1|\right),
\end{equation}
because
\begin{align}
 & \left(\id\otimes\mathcal{M}_{i}\right)\left(\Phi_{d}\right)\nonumber \\
 & =\frac{1}{d}\sum_{i,j=0}^{d-1}|i\rangle\!\langle j|\otimes\Tr[M_{i}|i\rangle\!\langle j|]|0\rangle\!\langle0|\nonumber \\
 & \qquad+\frac{1}{d}\sum_{i,j=0}^{d-1}|i\rangle\!\langle j|\otimes\Tr[(I-M_{i})|i\rangle\!\langle j|]|1\rangle\!\langle1|\\
 & =\frac{1}{d}\sum_{i,j=0}^{d-1}|i\rangle\!\langle j|\otimes\langle j|M_{i}|i\rangle|0\rangle\!\langle0|\nonumber \\
 & \qquad+\frac{1}{d}\sum_{i,j=0}^{d-1}|i\rangle\!\langle j|\otimes\langle j|(I-M_{i})|i\rangle|1\rangle\!\langle1|\\
 & =\frac{1}{d}\sum_{i,j=0}^{d-1}|i\rangle\langle j|M_{i}|i\rangle\langle j|\otimes|0\rangle\!\langle0|\nonumber \\
 & \qquad+\frac{1}{d}\sum_{i,j=0}^{d-1}|i\rangle\langle j|(I-M_{i})|i\rangle\langle j|\otimes|1\rangle\!\langle1|\\
 & =\frac{1}{d}\left(M_{i}^{T}\otimes|0\rangle\!\langle0|+(I-M_{i})^{T}\otimes|1\rangle\!\langle1|\right).
\end{align}
(See also \cite[Exercise~4.19]{khatri2024}). Then~\eqref{eq:rel-ent-FD-meas}
follows because
\begin{align}
 & D\!\left(\left(\id\otimes\mathcal{M}_{1}\right)\left(\Phi_{d}\right)\|\left(\id\otimes\mathcal{M}_{2}\right)\left(\Phi_{d}\right)\right)\nonumber \\
 & \overset{(a)}{=}\frac{1}{d}D\!\left(M_{1}^{T}\|M_{2}^{T}\right)+\frac{1}{d}D\!\left((I-M_{1})^{T}\|(I-M_{2})^{T}\right)\\
 & \overset{(b)}{=}\frac{1}{d}\left[D\!\left(M_{1}\|M_{2}\right)+D\!\left(I-M_{1}\|I-M_{2}\right)\right]\\
 & =\frac{1}{d}D_{\FD}(M_{1}\|M_{2}).
\end{align}
The equality (a) follows from the direct-sum property of quantum relative
entropy \cite[Proposition~7.3]{khatri2024}. The equality (b) follows
because quantum relative entropy does not increase under the action
of a positive, trace-preserving map~\cite{MuellerHermes2017}, which
implies it is invariant under the transpose map.
\end{proof}
\begin{prop}
\label{prop:FD-relative-ent-nonneg}The Fermi--Dirac relative entropy
of measurement operators $M_{1}$ and $M_{2}$ is non-negative:
\begin{equation}
D_{\FD}(M_{1}\|M_{2})\geq0,\label{eq:non-neg-FD-rel-ent}
\end{equation}
and it is faithful:
\begin{equation}
D_{\FD}(M_{1}\|M_{2})=0\qquad\Longleftrightarrow\qquad M_{1}=M_{2}.
\end{equation}
\end{prop}

\begin{proof}
To establish these properties, we simply appeal to the expression
in~\eqref{eq:FD-rel-ent-in-terms-of-states} and the fact that these
properties hold for the quantum relative entropy of states.
\end{proof}
\begin{lem}
\label{lem:rewrite-free-energy-rel-ent}For a Hermitian matrix $A$,
a measurement operator $M$, and a temperature $T>0$, the following
equality holds:
\begin{multline}
\Tr\!\left[AM\right]-T\cdot S_{\FD}(M)\\
=TD_{\FD}\!\left(M\middle\|\left(e^{\frac{A}{T}}+I\right)^{-1}\right)-T\Tr\!\left[\ln\!\left(e^{-\frac{A}{T}}+I\right)\right].
\end{multline}
\end{lem}

\begin{proof}
Consider that 
\begin{align}
 & \Tr\!\left[AM\right]-T\cdot S_{\FD}(M)\nonumber \\
 & =T\left(-S_{\FD}(M)+\Tr\!\left[M\frac{A}{T}\right]\right)\\
 & =T\left(-S_{\FD}(M)+\Tr\!\left[M\ln e^{\frac{A}{T}}\right]\right).
\end{align}
Now consider that
\begin{align}
 & -S_{\FD}(M)+\Tr\!\left[M\ln e^{\frac{A}{T}}\right]\nonumber \\
 & =-S_{\FD}(M)+\Tr\!\left[M\ln e^{\frac{A}{T}}\right]-\Tr\!\left[M\ln\left(e^{\frac{A}{T}}+I\right)^{-1}\right]\nonumber \\
 & \qquad+\Tr\!\left[M\ln\left(e^{\frac{A}{T}}+I\right)^{-1}\right]\nonumber \\
 & \qquad-\Tr\!\left[\left(I-M\right)\ln\!\left(I-\left(e^{\frac{A}{T}}+I\right)^{-1}\right)\right]\nonumber \\
 & \qquad+\Tr\!\left[\left(I-M\right)\ln\!\left(I-\left(e^{\frac{A}{T}}+I\right)^{-1}\right)\right]\\
 & =D_{\FD}\!\left(M\middle\|\left(e^{\frac{A}{T}}+I\right)^{-1}\right)+\Tr\!\left[M\ln e^{\frac{A}{T}}\right]\nonumber \\
 & \qquad+\Tr\!\left[M\ln\left(e^{\frac{A}{T}}+I\right)^{-1}\right]\nonumber \\
 & \qquad+\Tr\!\left[\left(I-M\right)\ln\!\left(I-\left(e^{\frac{A}{T}}+I\right)^{-1}\right)\right]\\
 & =D_{\FD}\!\left(M\middle\|\left(e^{\frac{A}{T}}+I\right)^{-1}\right)\notag \\
 & \qquad +\Tr\!\left[M\ln\left(e^{\frac{A}{T}}\left(e^{\frac{A}{T}}+I\right)^{-1}\right)\right]\nonumber \\
 & \qquad+\Tr\!\left[\left(I-M\right)\ln\!\left(I-\left(e^{\frac{A}{T}}+I\right)^{-1}\right)\right]\\
 & =D_{\FD}\!\left(M\middle\|\left(e^{\frac{A}{T}}+I\right)^{-1}\right)\notag \\
 & \qquad +\Tr\!\left[M\ln\left(\left(e^{-\frac{A}{T}}+I\right)^{-1}\right)\right]\nonumber \\
 & \qquad+\Tr\!\left[\left(I-M\right)\ln\!\left(\left(e^{-\frac{A}{T}}+I\right)^{-1}\right)\right]\\
 & =D_{\FD}\!\left(M\middle\|\left(e^{\frac{A}{T}}+I\right)^{-1}\right)+\Tr\!\left[\ln\!\left(\left(e^{-\frac{A}{T}}+I\right)^{-1}\right)\right]\\
 & =D_{\FD}\!\left(M\middle\|\left(e^{\frac{A}{T}}+I\right)^{-1}\right)-\Tr\!\left[\ln\!\left(e^{-\frac{A}{T}}+I\right)\right],
\end{align}
thus concluding the proof.
\end{proof}
\begin{cor}
For a Hermitian matrix $A$ and a temperature $T>0$, the following
equality holds:
\begin{equation}
\min_{M\in\mathcal{M}_{d}}\left\{ \Tr\!\left[AM\right]-T\cdot S_{\FD}(M)\right\} =-T\Tr\!\left[\ln\!\left(e^{-\frac{A}{T}}+I\right)\right].\label{eq:min-free-energy}
\end{equation}
\end{cor}

\begin{proof}
This is a direct consequence of Proposition~\ref{prop:FD-relative-ent-nonneg}
and Lemma~\ref{lem:rewrite-free-energy-rel-ent}. Indeed, we can rewrite
the objective function on the left-hand side of~\eqref{eq:min-free-energy}
as in Lemma~\ref{lem:rewrite-free-energy-rel-ent}, and then we choose
$M=\left(e^{\frac{A}{T}}+I\right)^{-1}$. Then Proposition~\ref{prop:FD-relative-ent-nonneg}
establishes that this leads to the minimum value of the Fermi--Dirac
relative entropy.
\end{proof}

\section{Continuity bounds for Fermi--Dirac relative entropy and fermionic
free energy}

\label{sec:Continuity-bounds-FD-rel-ent}In this appendix, we establish
a Lipschitz continuity bound for the Fermi--Dirac relative entropy
(Lemma~\ref{lem:cont-fermi-dirac-rel-ent}) and a continuity bound
for the fermionic free energy (Lemma~\ref{lem:cont-ferm-free-energy}).
\begin{lem}
\label{lem:cont-fermi-dirac-rel-ent}Let $c\in\mathbb{N}$, let $\mu,\nu\in\mathbb{R}^{c}$,
let $H$, $Q_{1}$, $\ldots$, $Q_{c}$ be defined as in~\eqref{eq:vector-Herm-ops},
and let $M(\mu)$ and $N(\nu)$ be the following Fermi--Dirac measurement
operators:
\begin{align}
M(\mu) & \coloneqq\left(e^{\frac{1}{T}\left(H-\mu\cdot Q\right)}+I\right)^{-1},\\
N(\nu) & \coloneqq\left(e^{\frac{1}{T}\left(H-\nu\cdot Q\right)}+I\right)^{-1}.
\end{align}
Then
\begin{equation}
D_{\FD}\!\left(M(\mu)\|N(\nu)\right)\leq\frac{2K}{T}\left\Vert \mu-\nu\right\Vert ,
\end{equation}
where
\begin{equation}
K\coloneqq\sqrt{\sum_{i\in\left[c\right]}\left\Vert Q_{i}\right\Vert _{1}^{2}}.
\end{equation}
\end{lem}

\begin{proof}
Let us abbreviate $M\equiv M(\mu)$ and $N\equiv N(\nu)$. Observe
that
\begin{align}
\ln\!\left(\frac{M}{I-M}\right) & =-\frac{1}{T}\left(H-\mu\cdot Q\right),\\
\ln\!\left(\frac{N}{I-N}\right) & =-\frac{1}{T}\left(H-\nu\cdot Q\right),
\end{align}
due to the scalar identity
\begin{align}
 & \ln\!\left(\frac{\left(e^{x}+1\right)^{-1}}{1-\left(e^{x}+1\right)^{-1}}\right)\nonumber \\
 & =\ln\!\left(\frac{1}{\left(e^{x}+1\right)\left(1-\left(e^{x}+1\right)^{-1}\right)}\right)\\
 & =\ln\!\left(\frac{1}{\left(e^{x}+1-1\right)}\right)\\
 & =\ln\!\left(\frac{1}{e^{x}}\right)\\
 & =-x.
\end{align}
Consider that
\begin{align}
 & D_{\FD}\!\left(M\|N\right)\nonumber \\
 & =\Tr\!\left[M\left(\ln M-\ln N\right)\right]\nonumber \\
 & \qquad+\Tr\!\left[\left(I-M\right)\left(\ln\!\left(I-M\right)-\ln\!\left(I-N\right)\right)\right]\\
 & =\Tr\!\left[M\left(\ln\!\left(\frac{M}{I-M}\right)-\ln\!\left(\frac{N}{I-N}\right)\right)\right]\nonumber \\
 & \qquad+\Tr\!\left[\ln\!\left(I-M\right)-\ln\!\left(I-N\right)\right]\\
 & =\Tr\!\left[M\left(-\frac{1}{T}\left(H-\mu\cdot Q\right)-\left(-\frac{1}{T}\left(H-\nu\cdot Q\right)\right)\right)\right]\nonumber \\
 & \qquad-\Tr\!\left[\ln\!\left(e^{-\frac{1}{T}\left(H-\mu\cdot Q\right)}+I\right)\right]\nonumber \\
 & \qquad+\Tr\!\left[\ln\!\left(e^{-\frac{1}{T}\left(H-\nu\cdot Q\right)}+I\right)\right]\\
 & =\Tr\!\left[M\left(\frac{1}{T}\left(\mu-\nu\right)\cdot Q\right)\right]-\Tr\!\left[\ln\!\left(e^{-\frac{1}{T}\left(H-\mu\cdot Q\right)}+I\right)\right]\nonumber \\
 & \qquad+\Tr\!\left[\ln\!\left(e^{-\frac{1}{T}\left(H-\nu\cdot Q\right)}+I\right)\right].\label{eq:alt-formula-FD-relative-entr-FD-thermal-meas}
\end{align}
Now apply Lemma~\ref{lem:cont-ferm-free-energy}, with $H_{1}=\frac{1}{T}\left(H-\mu\cdot Q\right)$
and $H_{2}=\frac{1}{T}\left(\left(\mu-\nu\right)\cdot Q\right)$,
to conclude that
\begin{align}
 & -\Tr\!\left[\ln\!\left(e^{-\frac{1}{T}\left(H-\mu\cdot Q\right)}+I\right)\right]+\Tr\!\left[\ln\!\left(e^{-\frac{1}{T}\left(H-\nu\cdot Q\right)}+I\right)\right]\nonumber \\
 & \leq\left\Vert \frac{1}{T}\left(\left(\mu-\nu\right)\cdot Q\right)\right\Vert _{1}\\
 & =\frac{1}{T}\left\Vert \sum_{i\in\left[c\right]}\left(\mu_{i}-\nu_{i}\right)Q_{i}\right\Vert _{1}\\
 & \leq\frac{1}{T}\sum_{i\in\left[c\right]}\left|\mu_{i}-\nu_{i}\right|\left\Vert Q_{i}\right\Vert _{1}\\
 & \leq\frac{1}{T}\left\Vert \mu-\nu\right\Vert \sqrt{\sum_{i\in\left[c\right]}\left\Vert Q_{i}\right\Vert _{1}^{2}}.
\end{align}
Similarly,
\begin{align}
\Tr\!\left[M\left(\frac{1}{T}\left(\mu-\nu\right)\cdot Q\right)\right] & \leq\left\Vert \frac{1}{T}\left(\left(\mu-\nu\right)\cdot Q\right)\right\Vert _{1}\\
 & \leq\frac{1}{T}\left\Vert \mu-\nu\right\Vert \sqrt{\sum_{i\in\left[c\right]}\left\Vert Q_{i}\right\Vert _{1}^{2}}.
\end{align}
So then
\begin{equation}
D_{\FD}\!\left(M\|N\right)\leq\frac{2}{T}\left\Vert \mu-\nu\right\Vert \sqrt{\sum_{i\in\left[c\right]}\left\Vert Q_{i}\right\Vert _{1}^{2}},
\end{equation}
thus concluding the proof.
\end{proof}
\begin{lem}
\label{lem:cont-ferm-free-energy}Let $H_{1}$ and $H_{2}$ be $d\times d$
Hermitian matrices. Then
\begin{equation}
\left|-\Tr\!\left[\ln\!\left(e^{-\left(H_{1}+H_{2}\right)}+I\right)\right]+\Tr\!\left[\ln\!\left(e^{-H_{1}}+I\right)\right]\right|\leq\left\Vert H_{2}\right\Vert _{1}.
\end{equation}
\end{lem}

\begin{proof}
The approach is similar to that used to prove \cite[Eq.~(55)]{Alhambra2023}.
Consider that
\begin{align}
 & \left|-\Tr\!\left[\ln\!\left(e^{-\left(H_{1}+H_{2}\right)}+I\right)\right]+\Tr\!\left[\ln\!\left(e^{-H_{1}}+I\right)\right]\right|\nonumber \\
 & =\left|\int_{0}^{1}ds\,\frac{d}{ds}\left(-\Tr\!\left[\ln\!\left(e^{-\left(H_{1}+sH_{2}\right)}+I\right)\right]\right)\right|\\
 & =\left|\int_{0}^{1}ds\,\Tr\!\left[\left(e^{-\left(H_{1}+sH_{2}\right)}+I\right)^{-1}\frac{d}{ds}\left(e^{-\left(H_{1}+sH_{2}\right)}+I\right)\right]\right|\\
 & =\left|\int_{0}^{1}ds\,\Tr\!\left[\left(e^{-\left(H_{1}+sH_{2}\right)}+I\right)^{-1}\frac{d}{ds}\left(e^{-\left(H_{1}+sH_{2}\right)}\right)\right]\right|\\
 & =\left|\int_{0}^{1}ds\,\Tr\!\left[\begin{array}{c}
\left(e^{-\left(H_{1}+sH_{2}\right)}+I\right)^{-1}\times\\
\int_{0}^{1}dt\,e^{-t\left(H_{1}+sH_{2}\right)}\frac{d}{ds}\left(-\left(H_{1}+sH_{2}\right)\right)\times\\
e^{-\left(1-t\right)\left(H_{1}+sH_{2}\right)}
\end{array}\right]\right|\\
 & =\left|\int_{0}^{1}ds\,\int_{0}^{1}dt\,\Tr\!\left[\begin{array}{c}
e^{-\left(1-t\right)\left(H_{1}+sH_{2}\right)}e^{-t\left(H_{1}+sH_{2}\right)}\times\\
\left(e^{-\left(H_{1}+sH_{2}\right)}+I\right)^{-1}\times\\
\frac{d}{ds}\left(-\left(H_{1}+sH_{2}\right)\right)
\end{array}\right]\right|\\
 & =\left|\int_{0}^{1}ds\,\Tr\!\left[e^{-\left(H_{1}+sH_{2}\right)}\left(e^{-\left(H_{1}+sH_{2}\right)}+I\right)^{-1}H_{2}\right]\right|\\
 & =\left|\int_{0}^{1}ds\,\Tr\!\left[\left(e^{\left(H_{1}+sH_{2}\right)}+I\right)^{-1}H_{2}\right]\right|\\
 & \leq\int_{0}^{1}ds\,\left|\Tr\!\left[\left(e^{\left(H_{1}+sH_{2}\right)}+I\right)^{-1}H_{2}\right]\right|\\
 & \leq\left\Vert \left(e^{\left(H_{1}+sH_{2}\right)}+I\right)^{-1}\right\Vert \left\Vert H_{2}\right\Vert _{1}\\
 & \leq\left\Vert H_{2}\right\Vert _{1},
\end{align}
thus concluding the proof.
\end{proof}
One can formulate a Fisher information matrix based on the Fermi--Dirac
relative entropy, in the following way:
\begin{align}
\left[I(\mu)\right]_{i,j} & \coloneqq\left.\frac{\partial^{2}}{\partial\varepsilon_{i}\partial\varepsilon_{j}}D_{\FD}(M_{T}(\mu)\|M_{T}(\mu+\varepsilon))\right|_{\varepsilon=0},
\end{align}
similar to how we define Fisher information matrices from parameterized
families of quantum states (see, e.g.,~\cite{Wilde2025} for this
approach). The next proposition is a Fermi--Dirac counterpart of
\cite[Lemma~6]{liu2025qthermoSDPs}.
\begin{prop}
The following equalities hold for the Fisher information matrix based
on the Fermi--Dirac relative entropy:
\begin{align}
\left[I(\mu)\right]_{i,j} & =\frac{1}{T^{2}}\int_{0}^{1}ds\,\Tr\!\left[M_{T}(\mu,s)Q_{i}M_{T}(\mu,1-s)Q_{j}\right],\nonumber \\
 & =\frac{1}{T^{2}}\Re\!\left[\Tr\!\left[M_{T}(\mu)\Phi_{\mu}(Q_{i})\left(I-M_{T}(\mu)\right)Q_{j}\right]\right].\label{eq:alt-exp-Fisher-info}
\end{align}
\end{prop}

\begin{proof}
To arrive at the expressions above, consider that
\begin{multline}
D_{\FD}(M_{T}(\mu)\|M_{T}(\mu+\varepsilon))\\
=\frac{1}{T}\Tr\!\left[M_{T}(\mu)\left(\varepsilon\cdot Q\right)\right]-\Tr\!\left[\ln\!\left(e^{-\frac{1}{T}\left(H-\mu\cdot Q\right)}+I\right)\right]\\
\qquad+\Tr\!\left[\ln\!\left(e^{-\frac{1}{T}\left(H-\left(\mu+\varepsilon\right)\cdot Q\right)}+I\right)\right],
\end{multline}
where we applied~\eqref{eq:alt-formula-FD-relative-entr-FD-thermal-meas}.
This implies that
\begin{align}
 & \frac{\partial^{2}}{\partial\varepsilon_{i}\partial\varepsilon_{j}}D_{\FD}(M_{T}(\mu)\|M_{T}(\mu+\varepsilon))\nonumber \\
 & =\frac{\partial^{2}}{\partial\varepsilon_{i}\partial\varepsilon_{j}}\Tr\!\left[\ln\!\left(e^{-\frac{1}{T}\left(H-\left(\mu+\varepsilon\right)\cdot Q\right)}+I\right)\right]\\
 & =\frac{1}{T^{2}}\int_{0}^{1}ds\,\Tr\!\left[M_{T}(\mu+\varepsilon,s)Q_{i}M_{T}(\mu+\varepsilon,1-s)Q_{j}\right],
\end{align}
where the second equality follows from~\eqref{eq:1st-hessian-exp}.
We then conclude that
\begin{multline}
\left.\frac{\partial^{2}}{\partial\varepsilon_{i}\partial\varepsilon_{j}}D_{\FD}(M_{T}(\mu)\|M_{T}(\mu+\varepsilon))\right|_{\varepsilon=0}\\
=\frac{1}{T^{2}}\int_{0}^{1}ds\,\Tr\!\left[M_{T}(\mu,s)Q_{i}M_{T}(\mu,1-s)Q_{j}\right].
\end{multline}
The equality in~\eqref{eq:alt-exp-Fisher-info} follows from applying
\eqref{eq:2nd-hessian-exp}.
\end{proof}

\section{Approximation error bounds}

\subsection{Proof of Proposition~\ref{prop:simple-approx-bnd} (simple bound
on approximation error)}

\label{sec:proof-simple-approx-bnd}

Let us first prove the inequality
\begin{equation}
E(\mathcal{Q},q)\geq F_{T}(\mathcal{Q},q)\label{eq:first-simple-ineq-en-to-free-en}
\end{equation}
in~\eqref{eq:simple-approx-bnd}. Suppose that $M\in\mathcal{M}_{d}$
is feasible for the optimization problem in~\eqref{eq:measurement-opt-gen}
(i.e., it satisfies $\Tr\!\left[Q_{i}M\right]=q_{i}$ for all $i\in\left[c\right]$).
Then
\begin{align}
\Tr\!\left[HM\right] & \geq\Tr\!\left[HM\right]-T\cdot S_{\FD}(M)\\
 & \geq F_{T}(\mathcal{Q},q).
\end{align}
The first inequality follows because $T>0$ and $S_{\FD}(M)\geq0$.
The second inequality follows from the definition of $F_{T}(\mathcal{Q},q)$
in~\eqref{eq:primal-FD-obj}. Since the inequality holds for all feasible
$M\in\mathcal{M}_{d}$, by applying the definition of $E(\mathcal{Q},q)$
in~\eqref{eq:measurement-opt-gen}, we conclude~\eqref{eq:first-simple-ineq-en-to-free-en}.

We now prove the inequality
\begin{equation}
F_{T}(\mathcal{Q},q)\geq E(\mathcal{Q},q)-T\cdot d\ln2,\label{eq:second-simple-ineq-en-to-free-en}
\end{equation}
which establishes the second inequality in~\eqref{eq:simple-approx-bnd}
by setting $T\leq\frac{\varepsilon}{d\ln2}$. Suppose that $M\in\mathcal{M}_{d}$
is feasible for the optimization problem in~\eqref{eq:primal-FD-obj}
(i.e., it satisfies $\Tr\!\left[Q_{i}M\right]=q_{i}$ for all $i\in\left[c\right]$).
Then
\begin{align}
\Tr\!\left[HM\right]-T\cdot S_{\FD}(M) & \geq\Tr\!\left[HM\right]-T\cdot d\ln2\\
 & \geq E(\mathcal{Q},q)-T\cdot d\ln2.
\end{align}
The first inequality follows because
\begin{equation}
S_{\FD}(M)\leq d\ln2,\label{eq:unif-bnd-FD-ent}
\end{equation}
which is a consequence of~\eqref{eq:FD-entropy-sum-binary-entropies}
and the fact that $s(p)\leq\ln2$ for all $p\in\left[0,1\right]$.
The second inequality follows from the definition of $E(\mathcal{Q},q)$
in~\eqref{eq:measurement-opt-gen}. Since the inequality holds for
all feasible $M\in\mathcal{M}_{d}$, by applying the definition of
$F_{T}(\mathcal{Q},q)$ in~\eqref{eq:primal-FD-obj}, we conclude
\eqref{eq:first-simple-ineq-en-to-free-en}.

\subsection{Proof of Proposition~\ref{prop:no-FD-ent-simple-approx-err-bnd}}

\label{sec:no-FD-ent-simple-approx-err-bnd}

Consider that
\begin{align}
\tilde{f}_{T}(\mu_{T}^{\star}) & \overset{(a)}{=}\mu_{T}^{\star}\cdot q+\Tr\!\left[\left(H-\mu_{T}^{\star}\cdot Q\right)M_{T}(\mu_{T}^{\star})\right]\\
 & \overset{(b)}{\geq}\mu_{T}^{\star}\cdot q+\Tr\!\left[\left(H-\mu_{T}^{\star}\cdot Q\right)M_{T}(\mu_{T}^{\star})\right]\nonumber \\
 & \qquad-T\cdot S_{\FD}(M_{T}(\mu_{T}^{\star}))\\
 & \overset{(c)}{=}F_{T}(\mathcal{Q},q)\\
 & \overset{(d)}{\geq}E(\mathcal{Q},q)-\varepsilon.
\end{align}
The equality (a) follows from the definition in~\eqref{eq:approx-without-FD-entropy}.
The inequality (b) follows because $T\cdot S_{\FD}(M_{T}(\mu_{T}^{\star}))\geq0$.
The equality (c) follows from~\eqref{eq:free-energy-opt-dual}. The
inequality (d) follows from~\eqref{eq:simple-approx-bnd}. This establishes
the inequality $\tilde{f}_{T}(\mu_{T}^{\star})\geq E(\mathcal{Q},q)-\varepsilon$.

Now consider that
\begin{align}
  \tilde{f}_{T}(\mu_{T}^{\star})
 & \overset{(a)}{=}\mu_{T}^{\star}\cdot q+\Tr\!\left[\left(H-\mu_{T}^{\star}\cdot Q\right)M_{T}(\mu_{T}^{\star})\right]\\
 & \overset{(b)}{\leq}\mu_{T}^{\star}\cdot q+\Tr\!\left[\left(H-\mu_{T}^{\star}\cdot Q\right)M_{T}(\mu_{T}^{\star})\right]\nonumber \\
 & \qquad-T\cdot S_{\FD}(M_{T}(\mu_{T}^{\star}))+\varepsilon\nonumber \\
 & \overset{(c)}{=}F_{T}(\mathcal{Q},q)+\varepsilon\\
 & \overset{(d)}{\leq}E(\mathcal{Q},q)+\varepsilon.
\end{align}
The equality (a) follows from the definition in~\eqref{eq:approx-without-FD-entropy}.
The inequality (b) follows because
\begin{equation}
T\cdot S_{\FD}(M_{T}(\mu_{T}^{\star}))\leq\varepsilon,
\end{equation}
which in turn is a consequence of the assumption $T\leq\frac{\varepsilon}{d\ln2}$
and the inequality in~\eqref{eq:unif-bnd-FD-ent}. The equality (c)
follows from~\eqref{eq:free-energy-opt-dual}. The inequality (d)
follows from~\eqref{eq:simple-approx-bnd}. This establishes the other
inequality $\tilde{f}_{T}(\mu_{T}^{\star})\leq E(\mathcal{Q},q)+\varepsilon$.

\subsection{Proof of Proposition~\ref{prop:approx-error-spectral-gap} (spectral
gap bound on approximation error)}

\label{sec:proof-approx-error-spectral-gap}

Let $H-\mu\cdot Q$ be a Hamiltonian, a $d\times d$ Hermitian matrix
with positive, zero, and negative eigenvalues, as defined from~\eqref{eq:vector-Herm-ops}.
Then $H-\mu\cdot Q$ has the following unique spectral decomposition:
\begin{equation}
H-\mu\cdot Q=\sum_{i=1}^{M}\lambda_{i}\Pi_{i},
\end{equation}
where $\lambda_{i}$ is an eigenvalue and $\Pi_{i}$ is the corresponding
eigenprojection. Here and throughout, we use the shorthand $\lambda_{i}\equiv\lambda_{i}(H-\mu\cdot Q)$
to refer to the eigenvalues of $H-\mu\cdot Q$.

Consider that the objective function $f_{T}(\mu)$ in~\eqref{eq:dual-FD-obj}
has the following form:
\begin{align}
f_{T}(\mu) & =\mu\cdot q-T\Tr\!\left[\ln\!\left(e^{-\frac{\left(H-\mu\cdot Q\right)}{T}}+I\right)\right]\\
 & =\mu\cdot q-T\Tr\!\left[\ln\!\left(e^{-\frac{\sum_{i=1}^{M}\lambda_{i}\Pi_{i}}{T}}+I\right)\right]\\
 & =\mu\cdot q-T\Tr\!\left[\sum_{i=1}^{M}\ln\!\left(e^{-\frac{\lambda_{i}}{T}}+1\right)\Pi_{i}\right]\\
 & =\mu\cdot q-T\sum_{i=1}^{M}\ln\!\left(e^{-\frac{\lambda_{i}}{T}}+1\right)\Tr\!\left[\Pi_{i}\right]\\
 & =\mu\cdot q-T\sum_{i=1}^{M}d_{i}\ln\!\left(e^{-\frac{\lambda_{i}}{T}}+1\right),\label{eq:fermionic-free-en-spec-decomp}
\end{align}
where
\begin{equation}
d_{i}\coloneqq\Tr\!\left[\Pi_{i}\right].
\end{equation}
Furthermore, observe that the objective function $f(\mu)$ in~\eqref{eq:dual-unreg-obj-func}
has the following form:
\begin{align}
f(\mu) & =\mu\cdot q-\Tr\!\left[\left(H-\mu\cdot Q\right)_{-}\right]\\
 & =\mu\cdot q+\Tr\!\left[\sum_{i:\lambda_{i}<0}\lambda_{i}\Pi_{i}\right]\\
 & =\mu\cdot q+\sum_{i:\lambda_{i}<0}\lambda_{i}\Tr\!\left[\Pi_{i}\right]\\
 & =\mu\cdot q+\sum_{i:\lambda_{i}<0}\lambda_{i}d_{i}.\label{eq:obj-func-value-spectral-decomp}
\end{align}
Consider that
\begin{align}
 & f_{T}(\mu)-f(\mu)\nonumber \\
 & =\mu\cdot q-T\Tr\!\left[\ln\!\left(e^{-\frac{1}{T}\left(H-\mu\cdot Q\right)}+I\right)\right]\nonumber \\
 & \qquad-\left(\mu\cdot q-\Tr\!\left[\left(H-\mu\cdot Q\right)_{-}\right]\right)\\
 & =-T\Tr\!\left[\ln\!\left(e^{-\frac{1}{T}\left(H-\mu\cdot Q\right)}+I\right)\right]+\Tr\!\left[\left(H-\mu\cdot Q\right)_{-}\right]\\
 & \overset{(a)}{=}-T\sum_{i}d_{i}\ln\!\left(e^{-\frac{\lambda_{i}}{T}}+1\right)-\sum_{i:\lambda_{i}<0}\lambda_{i}d_{i}\\
 & =-T\sum_{i}d_{i}\ln\!\left(e^{-\frac{\lambda_{i}}{T}}+1\right)+\sum_{i:\lambda_{i}<0}\left|\lambda_{i}\right|d_{i}\\
 & =-T\sum_{i:\lambda_{i}\geq0}d_{i}\ln\!\left(e^{-\frac{\lambda_{i}}{T}}+1\right)-T\sum_{i:\lambda_{i}<0}d_{i}\ln\!\left(e^{-\frac{\lambda_{i}}{T}}+1\right)\nonumber \\
 & \qquad+T\sum_{i:\lambda_{i}<0}d_{i}\ln\!\left(e^{\frac{\left|\lambda_{i}\right|}{T}}\right)\\
 & =-T\sum_{i:\lambda_{i}\geq0}d_{i}\ln\!\left(e^{-\frac{\lambda_{i}}{T}}+1\right)-T\sum_{i:\lambda_{i}<0}d_{i}\ln\!\left(\frac{e^{-\frac{\lambda_{i}}{T}}+1}{e^{\frac{\left|\lambda_{i}\right|}{T}}}\right)\\
 & =-T\sum_{i:\lambda_{i}\geq0}d_{i}\ln\!\left(e^{-\frac{\lambda_{i}}{T}}+1\right)-T\sum_{i:\lambda_{i}<0}d_{i}\ln\!\left(e^{-\frac{\left|\lambda_{i}\right|}{T}}+1\right)\\
 & =-T\sum_{i}d_{i}\ln\!\left(e^{-\frac{\left|\lambda_{i}\right|}{T}}+1\right).
\end{align}
Equality (a) follows from~\eqref{eq:fermionic-free-en-spec-decomp}
and~\eqref{eq:obj-func-value-spectral-decomp}. Now recall the definitions
of $d_{0}$ and $\Delta$ in~\eqref{eq:def-d0} and~\eqref{eq:def-Delta},
respectively. Continuing, consider that
\begin{align}
 & -T\sum_{i}d_{i}\ln\!\left(e^{-\frac{\left|\lambda_{i}\right|}{T}}+1\right)\nonumber \\
 & =-T\sum_{i:\lambda_{i}=0}d_{i}\ln\!\left(e^{-\frac{\left|\lambda_{i}\right|}{T}}+1\right)\notag\\
 & \qquad -T\sum_{i:\lambda_{i}\neq0}d_{i}\ln\!\left(e^{-\frac{\left|\lambda_{i}\right|}{T}}+1\right)\\
 & =-T\sum_{i:\lambda_{i}=0}d_{i}\ln2-T\sum_{i:\lambda_{i}\neq0}d_{i}\ln\!\left(e^{-\frac{\left|\lambda_{i}\right|}{T}}+1\right)\\
 & \overset{(a)}{\geq}-Td_{0}\ln2-T\sum_{i:\lambda_{i}\neq0}d_{i}\ln\!\left(e^{-\frac{\left|\lambda_{i}\right|}{T}}+1\right)\\
 & \overset{(b)}{\geq}-Td_{0}\ln2-T\sum_{i:\lambda_{i}\neq0}d_{i}\ln\!\left(e^{-\frac{\Delta}{T}}+1\right)\\
 & =-Td_{0}\ln2-T\ln\!\left(e^{-\frac{\Delta}{T}}+1\right)\sum_{i:\lambda_{i}\neq0}d_{i}\\
 & \overset{(c)}{\geq}-Td_{0}\ln2-T\left(d-d_{0}\right)\ln\!\left(e^{-\frac{\Delta}{T}}+1\right)\\
 & \overset{(d)}{\geq}-Td_{0}\ln2-T\left(d-d_{0}\right)e^{-\frac{\Delta}{T}}\\
 & =-T\left(d_{0}\ln2+\left(d-d_{0}\right)e^{-\frac{\Delta}{T}}\right).
\end{align}
The inequality (a) follows from the definition of $d_{0}$ in~\eqref{eq:def-d0}.
The inequality (b) follows from the definition of $\Delta$ in~\eqref{eq:def-Delta}.
The inequality (c) again follows from the definition of $d_{0}$ in
\eqref{eq:def-d0}. The inequality (d) follows because $-\ln(x+1)\geq-x$
for all $x\geq0$.

Thus, we have proven that
\begin{equation}
f_{T}(\mu)\geq f(\mu)-T\left(d_{0}\ln2+\left(d-d_{0}\right)e^{-\frac{\Delta}{T}}\right).
\end{equation}
Taking the supremum over all $\mu\in\mathbb{M}$, we conclude that
\begin{equation}
F_{T}(\mathcal{Q},q)\geq E(\mathcal{Q},q)-T\left(d_{0}\ln2+\left(d-d_{0}\right)e^{-\frac{\Delta}{T}}\right),
\end{equation}
thus establishing~\eqref{eq:temp-depend-bnd-d0-spec-gap}.

For the error to be $\leq\varepsilon$, we require that
\begin{equation}
T\left(d_{0}\ln2+\left(d-d_{0}\right)e^{-\frac{\Delta}{T}}\right)\leq\varepsilon.\label{eq:up-bnd-error-temp-depend-d0-spec-gap}
\end{equation}
If we pick $T$ such that
\begin{align}
\left(d-d_{0}\right)e^{-\frac{\Delta}{T}} & \leq d_{0}\ln2\\
\Longleftrightarrow\quad\frac{1}{\ln2}\left(\frac{d-d_{0}}{d_{0}}\right) & \leq e^{\frac{\Delta}{T}},\\
\Longleftrightarrow\quad\ln\!\left(\frac{1}{\ln2}\left(\frac{d-d_{0}}{d_{0}}\right)\right) & \leq\frac{\Delta}{T}\\
\Longleftrightarrow\quad T\leq & \frac{\Delta}{\ln\!\left(\frac{1}{\ln2}\left(\frac{d-d_{0}}{d_{0}}\right)\right)},
\end{align}
then we find that
\begin{equation}
T\left(d_{0}\ln2+\left(d-d_{0}\right)e^{-\frac{\Delta}{T}}\right)\leq T2d_{0}\ln2.
\end{equation}
If we furthermore require $T$ to satisfy
\begin{equation}
T\leq\frac{\varepsilon}{2d_{0}\ln2},
\end{equation}
then $T2d_{0}\ln2\leq\varepsilon$. So, if the temperature $T$ satisfies
the following bound
\begin{equation}
T\leq\min\left\{ \frac{\varepsilon}{2d_{0}\ln2},\frac{\Delta}{\ln\!\left(\frac{1}{\ln2}\left(\frac{d-d_{0}}{d_{0}}\right)\right)}\right\} ,
\end{equation}
then~\eqref{eq:up-bnd-error-temp-depend-d0-spec-gap} holds, thus
concluding the proof.

\subsection{Proof of Proposition~\ref{prop:FD-entropy-up-bnd} (upper bound on
temperature scaled Fermi--Dirac entropy)}

\label{sec:FD-entropy-up-bnd-proof}

We first prove~\eqref{eq:FD-entropy-up-bnd-1}. Consider that
\begin{align}
 & S_{\FD}(M_{T}(A))\nonumber \\
 & =\sum_{i}s\!\left(\frac{1}{e^{\lambda_{i}/T}+1}\right)\\
 & =\sum_{i:\lambda_{i}=0}s\!\left(\frac{1}{e^{\lambda_{i}/T}+1}\right)+\sum_{i:\lambda_{i}\neq0}s\!\left(\frac{1}{e^{\lambda_{i}/T}+1}\right)\\
 & \overset{(a)}{=}\sum_{i:\lambda_{i}=0}s\!\left(\frac{1}{2}\right)+\sum_{i:\lambda_{i}\neq0}s\!\left(\frac{1}{e^{\left|\lambda_{i}\right|/T}+1}\right)\\
 & \overset{(b)}{=}d_{0}\ln2+\sum_{i:\lambda_{i}\neq0}s\!\left(\frac{1}{e^{\left|\lambda_{i}\right|/T}+1}\right)\\
 & \overset{(c)}{\leq}d_{0}\ln2+\sum_{i:\lambda_{i}\neq0}s\!\left(\frac{1}{e^{\Delta/T}+1}\right)\\
 & =d_{0}\ln2+\left(d-d_{0}\right)s\!\left(\frac{1}{e^{\Delta/T}+1}\right).
\end{align}
The equality (a) follows because
\begin{equation}
s\!\left(\left(e^{x/T}+1\right)^{-1}\right)=s\!\left(\frac{1}{2}\right)
\end{equation}
when $x=0$ and because
\begin{equation}
s\!\left(\left(e^{x/T}+1\right)^{-1}\right)=s\!\left(\left(e^{\left|x\right|/T}+1\right)^{-1}\right)\label{eq:bin-ent-abs-val}
\end{equation}
for all $x\in\mathbb{R}$. Indeed,~\eqref{eq:bin-ent-abs-val} is
evident if $x\geq0$, and if $x<0$, then
\begin{align}
s\!\left(\left(e^{x/T}+1\right)^{-1}\right) & =s\!\left(1-\left(e^{x/T}+1\right)^{-1}\right)\\
 & =s\!\left(e^{x/T}\left(e^{x/T}+1\right)^{-1}\right)\\
 & =s\!\left(\left(e^{-x/T}+1\right)^{-1}\right)\\
 & =s\!\left(\left(e^{\left|x\right|/T}+1\right)^{-1}\right).
\end{align}
The equality (b) follows from the definition of $d_{0}$ in~\eqref{eq:kernel-dim}
and because $s\!\left(\frac{1}{2}\right)=\ln2$. The inequality (c)
follows because $\Delta\leq\left|\lambda_{i}\right|$ for all $\lambda_{i}\neq0$
(by the definition in~\eqref{eq:spectral-gap-def-FD-ent-bnd}), which
implies that
\begin{equation}
\left(e^{\left|\lambda_{i}\right|/T}+1\right)^{-1}\leq\left(e^{\Delta/T}+1\right)^{-1}.
\end{equation}
Since $\left(e^{\Delta/T}+1\right)^{-1}\leq\frac{1}{2}$ and the binary
entropy function $s(\cdot)$ is monotone increasing on the interval
$\left[0,\frac{1}{2}\right]$, we conclude the inequality (c). Multiplying
the final inequality by $T$, we conclude that
\begin{equation}
T\cdot S_{\FD}(M_{T}(A))\leq T\left[d_{0}\ln2+\left(d-d_{0}\right)s\!\left(\frac{1}{e^{\Delta/T}+1}\right)\right],
\end{equation}
which is~\eqref{eq:FD-entropy-up-bnd-1}.

We now prove~\eqref{eq:FD-entropy-up-bnd-2}. For $p\in\left[0,1\right]$,
consider that
\begin{equation}
s(p)\leq p\left(1-\ln p\right)
\end{equation}
because $-\left(1-p\right)\ln\!\left(1-p\right)\leq p$. Now consider
that
\begin{align}
-\ln\!\left(\frac{1}{e^{\Delta/T}+1}\right) & =\ln\!\left(e^{\Delta/T}+1\right)\\
 & \leq\ln\!\left(2e^{\Delta/T}\right)\\
 & =\frac{\Delta}{T}+\ln2\\
 & \leq\frac{\Delta}{T}+1.
\end{align}
So it follows that
\begin{equation}
s\!\left(\frac{1}{e^{\Delta/T}+1}\right)\leq\frac{\frac{\Delta}{T}+2}{e^{\Delta/T}+1}.
\end{equation}
This implies that
\begin{align}
 & T\left[d_{0}\ln2+\left(d-d_{0}\right)s\!\left(\frac{1}{e^{\Delta/T}+1}\right)\right]\nonumber \\
 & \leq T\left[d_{0}\ln2+\left(d-d_{0}\right)\left(\frac{\frac{\Delta}{T}+2}{e^{\Delta/T}+1}\right)\right]\\
 & =Td_{0}\ln2+\left(d-d_{0}\right)\left(\frac{2T+\Delta}{e^{\Delta/T}+1}\right).\label{eq:FD-up-bnd-intermed-step}
\end{align}
We would like the expression in~\eqref{eq:FD-up-bnd-intermed-step}
to be $\leq\varepsilon$. If we pick
\begin{equation}
T\leq\frac{\varepsilon}{d_{0}2\ln2},\label{eq:temp-constraint-1}
\end{equation}
then
\begin{multline}
Td_{0}\ln2+\left(d-d_{0}\right)\left(\frac{2T+\Delta}{e^{\Delta/T}+1}\right)\\
\leq\frac{\varepsilon}{2}+\left(d-d_{0}\right)\left(\frac{2T+\Delta}{e^{\Delta/T}+1}\right).\label{eq:eps-bound-from-temp-constraint-first}
\end{multline}
Also, we aim to have the following inequality hold:
\begin{equation}
\left(d-d_{0}\right)\left(\frac{2T+\Delta}{e^{\Delta/T}+1}\right)\leq\frac{\varepsilon}{2}.
\end{equation}
If we pick
\begin{equation}
T\leq T_{0},\label{eq:temp-constraint-2}
\end{equation}
then
\begin{align}
& \left(d-d_{0}\right)\left(\frac{2T+\Delta}{e^{\Delta/T}+1}\right) \notag \\
& \leq\left(d-d_{0}\right)\left(\frac{2T_{0}+\Delta}{e^{\Delta/T}+1}\right)\\
 & \leq\left(d-d_{0}\right)e^{-\Delta/T}\left(2T_{0}+\Delta\right).
\end{align}
So then it suffices to take
\begin{align}
\left(d-d_{0}\right)e^{-\Delta/T}\left(2T_{0}+\Delta\right) & \leq\frac{\varepsilon}{2}\label{eq:eps-bound-from-temp-constraint-last}\\
\Longleftrightarrow\qquad\frac{2\left(d-d_{0}\right)\left(2T_{0}+\Delta\right)}{\varepsilon} & \leq e^{\Delta/T}\\
\Longleftrightarrow\qquad\ln\left(\frac{2\left(d-d_{0}\right)\left(2T_{0}+\Delta\right)}{\varepsilon}\right) & \leq\frac{\Delta}{T}\\
\Longleftrightarrow\qquad\frac{\Delta}{\ln\left(\frac{2\left(d-d_{0}\right)\left(2T_{0}+\Delta\right)}{\varepsilon}\right)} & \geq T.\label{eq:temp-constraint-3}
\end{align}
So putting together the constraints in~\eqref{eq:temp-constraint-1},
\eqref{eq:temp-constraint-2}, and~\eqref{eq:temp-constraint-3},
along with their consequences in~\eqref{eq:eps-bound-from-temp-constraint-first}
and~\eqref{eq:eps-bound-from-temp-constraint-last}, then the constraint
\begin{equation}
T\leq\min\left\{ T_{0},\frac{\varepsilon}{d_{0}2\ln2},\frac{\Delta}{\ln\left(\frac{2\left(d-d_{0}\right)\left(2T_{0}+\Delta\right)}{\varepsilon}\right)}\right\} 
\end{equation}
suffices to guarantee that
\begin{equation}
T\cdot S_{\FD}(A)\leq\varepsilon,
\end{equation}
thus establishing~\eqref{eq:FD-entropy-up-bnd-2}.

\subsection{Proof of Proposition~\ref{prop:approx-err-no-FD-obj}}

\label{sec:approx-err-no-FD-obj}

Consider that
\begin{align}
\tilde{f}_{T}(\mu_{T}^{\star}) & \overset{(a)}{=}\mu_{T}^{\star}\cdot q+\Tr\!\left[\left(H-\mu_{T}^{\star}\cdot Q\right)M_{T}(\mu_{T}^{\star})\right]\\
 & \overset{(b)}{\geq}\mu_{T}^{\star}\cdot q+\Tr\!\left[\left(H-\mu_{T}^{\star}\cdot Q\right)M_{T}(\mu_{T}^{\star})\right]\nonumber \\
 & \qquad-T\cdot S_{\FD}(M_{T}(\mu_{T}^{\star}))\\
 & \overset{(c)}{=}F_{T}(\mathcal{Q},q)\\
 & \overset{(d)}{\geq}E(\mathcal{Q},q)-T\left(d_{0}\ln2+\left(d-d_{0}\right)e^{-\frac{\Delta}{T}}\right).
\end{align}
The equality (a) follows from the definition in~\eqref{eq:approx-without-FD-entropy}.
The inequality (b) follows because $T\cdot S_{\FD}(M_{T}(\mu_{T}^{\star}))\geq0$.
The equality (c) follows from~\eqref{eq:free-energy-opt-dual}. The
inequality (d) follows from~\eqref{eq:temp-depend-bnd-d0-spec-gap}.
This establishes the lower bound in~\eqref{eq:err-bnd-obj-no-FD-ent-tem-depend}.

Now consider that
\begin{align}
 & \tilde{f}_{T}(\mu_{T}^{\star})\nonumber \\
 & \overset{(a)}{=}\mu_{T}^{\star}\cdot q+\Tr\!\left[\left(H-\mu_{T}^{\star}\cdot Q\right)M_{T}(\mu_{T}^{\star})\right]\\
 & \overset{(b)}{\leq}\mu_{T}^{\star}\cdot q+\Tr\!\left[\left(H-\mu_{T}^{\star}\cdot Q\right)M_{T}(\mu_{T}^{\star})\right]\nonumber \\
 & \qquad-T\cdot S_{\FD}(M_{T}(\mu_{T}^{\star}))\nonumber \\
 & \qquad+T\left[d_{0}\ln2+\left(d-d_{0}\right)s\!\left(\frac{1}{e^{\Delta/T}+1}\right)\right]\\
 & \overset{(c)}{=}F_{T}(\mathcal{Q},q)+T\left[d_{0}\ln2+\left(d-d_{0}\right)s\!\left(\frac{1}{e^{\Delta/T}+1}\right)\right]\\
 & \overset{(d)}{\leq}E(\mathcal{Q},q)+T\left[d_{0}\ln2+\left(d-d_{0}\right)s\!\left(\frac{1}{e^{\Delta/T}+1}\right)\right].
\end{align}
The equality (a) follows from the definition in~\eqref{eq:approx-without-FD-entropy}.
The inequality (b) follows because
\begin{equation}
T\cdot S_{\FD}(M_{T}(\mu_{T}^{\star}))\leq T\left[d_{0}\ln2+\left(d-d_{0}\right)s\!\left(\frac{1}{e^{\Delta/T}+1}\right)\right],
\end{equation}
which in turn is a consequence of the assumptions in~\eqref{eq:def-d0-1}--\eqref{eq:def-Delta-1}
and the inequality in~\eqref{eq:FD-entropy-up-bnd-1}. The equality
(c) follows from~\eqref{eq:free-energy-opt-dual}. The inequality
(d) follows from~\eqref{eq:simple-approx-bnd}. This establishes the
upper bound in~\eqref{eq:err-bnd-obj-no-FD-ent-tem-depend}. 

By choosing $T$ as in~\eqref{eq:temp-eps-error-spec-gap}, we conclude
from~\eqref{eq:up-bnd-error-temp-depend-d0-spec-gap} that 
\begin{equation}
T\left(d_{0}\ln2+\left(d-d_{0}\right)e^{-\frac{\Delta}{T}}\right)\leq\varepsilon,\label{eq:err-bnd-obj-no-FD-eps-1}
\end{equation}
and by choosing $T$ as in~\eqref{eq:temp-depend-eps-err-FD-ent},
we conclude from~\eqref{eq:FD-entropy-up-bnd-2} that
\begin{equation}
T\left[d_{0}\ln2+\left(d-d_{0}\right)s\!\left(\frac{1}{e^{\Delta/T}+1}\right)\right]\leq\varepsilon.\label{eq:err-bnd-obj-no-FD-eps-2}
\end{equation}
Finally putting together~\eqref{eq:err-bnd-obj-no-FD-ent-tem-depend},
\eqref{eq:err-bnd-obj-no-FD-eps-1}, and~\eqref{eq:err-bnd-obj-no-FD-eps-2},
we conclude that
\begin{equation}
E(\mathcal{Q},q)+\varepsilon\geq\tilde{f}_{T}(\mu_{T}^{\star})\geq E(\mathcal{Q},q)-\varepsilon,
\end{equation}
if~\eqref{eq:temp-threshold-obj-no-FD-ent} holds, which is equivalent
to the desired inequality in~\eqref{eq:err-bnd-obj-no-FD-eps}.

\section{Gradient and Hessian calculations}

\subsection{Proof of Proposition~\ref{prop:gradient-FD-gen-obj-func} (gradient
of dual objective function)}

\label{sec:Proof-of-gradient-FD-gen}Consider that
\begin{align}
 & \frac{\partial}{\partial\mu_{i}}f_{T}(\mu)\nonumber \\
 & =\frac{\partial}{\partial\mu_{i}}\left(\mu\cdot q-T\Tr\!\left[\ln\!\left(e^{-\frac{1}{T}\left(H-\mu\cdot Q\right)}+I\right)\right]\right)\\
 & =q_{i}-T\frac{\partial}{\partial\mu_{i}}\Tr\!\left[\ln\!\left(e^{-\frac{1}{T}\left(H-\mu\cdot Q\right)}+I\right)\right].\label{eq:grad-proof}
\end{align}
Now consider that
\begin{align}
 & \frac{\partial}{\partial\mu_{i}}\Tr\!\left[\ln\!\left(e^{-\frac{1}{T}\left(H-\mu\cdot Q\right)}+I\right)\right]\nonumber \\
 & \overset{(a)}{=}\Tr\!\left[\left(e^{-\frac{1}{T}\left(H-\mu\cdot Q\right)}+I\right)^{-1}\frac{\partial}{\partial\mu_{i}}\left(e^{-\frac{1}{T}\left(H-\mu\cdot Q\right)}+I\right)\right]\\
 & =\Tr\!\left[\left(e^{-\frac{1}{T}\left(H-\mu\cdot Q\right)}+I\right)^{-1}\frac{\partial}{\partial\mu_{i}}e^{-\frac{1}{T}\left(H-\mu\cdot Q\right)}\right]\\
 & \overset{(b)}{=}\Tr\!\left[\begin{array}{c}
\left(e^{-\frac{1}{T}\left(H-\mu\cdot Q\right)}+I\right)^{-1}\int_{0}^{1}dt\,e^{-\frac{t}{T}\left(H-\mu\cdot Q\right)}\times\\
\frac{\partial}{\partial\mu_{i}}\left[-\frac{1}{T}\left(H-\mu\cdot Q\right)\right]e^{-\frac{1-t}{T}\left(H-\mu\cdot Q\right)}
\end{array}\right]\\
 & \overset{(c)}{=}\int_{0}^{1}dt\,\Tr\!\left[\begin{array}{c}
e^{-\frac{1-t}{T}\left(H-\mu\cdot Q\right)}\left(e^{-\frac{1}{T}\left(H-\mu\cdot Q\right)}+I\right)^{-1}\times\\
e^{-\frac{t}{T}\left(H-\mu\cdot Q\right)}\frac{Q_{i}}{T}
\end{array}\right]\\
 & \overset{(d)}{=}\frac{1}{T}\int_{0}^{1}dt\,\Tr\!\left[\begin{array}{c}
e^{-\frac{1-t}{T}\left(H-\mu\cdot Q\right)}e^{-\frac{t}{T}\left(H-\mu\cdot Q\right)}\times\\
\left(e^{-\frac{1}{T}\left(H-\mu\cdot Q\right)}+I\right)^{-1}Q_{i}
\end{array}\right]\\
 & =\frac{1}{T}\Tr\!\left[e^{-\frac{1}{T}\left(H-\mu\cdot Q\right)}\left(e^{-\frac{1}{T}\left(H-\mu\cdot Q\right)}+I\right)^{-1}Q_{i}\right]\\
 & =\frac{1}{T}\Tr\!\left[\left(e^{\frac{1}{T}\left(H-\mu\cdot Q\right)}+I\right)^{-1}Q_{i}\right]\\
 & \overset{(e)}{=}\frac{1}{T}\Tr\!\left[M_{T}(\mu)Q_{i}\right].\label{eq:final-line-grad-proof}
\end{align}
The equality (a) follows because
\begin{equation}
\frac{\partial}{\partial x}\Tr[f(A(x))]=\Tr\!\left[f'(A(x))\frac{\partial}{\partial x}A(x)\right]
\end{equation}
for a differentiable function $f$ of a matrix $A(x)$. The equality
(b) follows from Duhamel's formula for the derivative of matrix exponentials
(see, e.g., \cite[Proposition~47]{Wilde2025}):
\begin{equation}
\frac{\partial}{\partial x}e^{A(x)}=\int_{0}^{1}dt\,e^{tA(x)}\left(\frac{\partial}{\partial x}A(x)\right)e^{\left(1-t\right)A(x)}.
\end{equation}
The equality (c) follows from cyclicity of trace and the fact that
\begin{equation}
\frac{\partial}{\partial\mu_{i}}\left[-\frac{1}{T}\left(H-\mu\cdot Q\right)\right]=\frac{Q_{i}}{T}.
\end{equation}
The equality (d) follows because $\left(e^{-\frac{1}{T}\left(H-\mu\cdot Q\right)}+I\right)^{-1}$
commutes with $e^{-\frac{t}{T}\left(H-\mu\cdot Q\right)}$. The final
equality (e) follows by applying the definition of the measurement
operator $M_{T}(\mu)$ in~\eqref{eq:FD-meas-op}. Combining~\eqref{eq:final-line-grad-proof}
with~\eqref{eq:grad-proof}, we conclude that
\begin{equation}
\frac{\partial}{\partial\mu_{i}}f(\mu)=q_{i}-\Tr\!\left[M_{T}(\mu)Q_{i}\right].
\end{equation}

\subsection{Proof of Proposition~\ref{prop:hessian-FD-gen-obj-func} (Hessian
of dual objective function)}

\label{sec:Proof-of-Hessian-FD-gen-obj-func}

We first prove~\eqref{eq:1st-hessian-exp}. Consider that
\begin{align}
 & \frac{\partial^{2}}{\partial\mu_{i}\partial\mu_{j}}f_{T}(\mu)\nonumber \\
 & =\frac{\partial}{\partial\mu_{i}}\left(q_{j}-\Tr\!\left[M_{T}(\mu)Q_{j}\right]\right)\\
 & =-\frac{\partial}{\partial\mu_{i}}\Tr\!\left[M_{T}(\mu)Q_{j}\right]\\
 & =-\Tr\!\left[\frac{\partial}{\partial\mu_{i}}\left[\left(e^{\frac{1}{T}\left(H-\mu\cdot Q\right)}+I\right)^{-1}\right]Q_{j}\right]\\
 & =\Tr\!\left[\begin{array}{c}
\left(e^{\frac{1}{T}\left(H-\mu\cdot Q\right)}+I\right)^{-1}\left(\frac{\partial}{\partial\mu_{i}}\left[e^{\frac{1}{T}\left(H-\mu\cdot Q\right)}+I\right]\right)\times\\
\left(e^{\frac{1}{T}\left(H-\mu\cdot Q\right)}+I\right)^{-1}Q_{j}
\end{array}\right]\\
 & =\Tr\!\left[M_{T}(\mu)\left[\frac{\partial}{\partial\mu_{i}}e^{\frac{1}{T}\left(H-\mu\cdot Q\right)}\right]M_{T}(\mu)Q_{j}\right]\label{eq:hessian-mid-proof}\\
 & =-\frac{1}{T}\int_{0}^{1}dt\,\Tr\!\left[
 \begin{array}{c}
 M_{T}(\mu)e^{\frac{t}{T}\left(H-\mu\cdot Q\right)}Q_{i}\times \\
 e^{\frac{1-t}{T}\left(H-\mu\cdot Q\right)}M_{T}(\mu)Q_{j}
 \end{array}
 \right]\\
 & =-\frac{1}{T}\int_{0}^{1}dt\,\Tr\!\left[
 \begin{array}{c}
 M_{T}(\mu)e^{\frac{t}{T}\left(H-\mu\cdot Q\right)}Q_{i}\times \\
 e^{\frac{1-t}{T}\left(H-\mu\cdot Q\right)}M_{T}(\mu)Q_{j}
 \end{array}
 \right]\\
 & =-\frac{1}{T}\int_{0}^{1}dt\,\Tr\!\left[Q_{i}M_{T}(\mu,1-s)Q_{j}M_{T}(\mu,s)\right],\label{eq:hessian-elements}
\end{align}
where we used the facts that
\begin{align}
\frac{\partial}{\partial x}\left[A(x)^{-1}\right] & =-A(x)^{-1}\left(\frac{\partial}{\partial x}A(x)\right)A(x)^{-1},\\
\frac{\partial}{\partial x}e^{A(x)} & =\int_{0}^{1}ds\,e^{sA(x)}\left(\frac{\partial}{\partial x}A(x)\right)e^{\left(1-s\right)A(x)},
\end{align}
and the definition of $M_{T}(\mu,s)$ in~\eqref{eq:meas-op-t-depend}.

We now prove~\eqref{eq:2nd-hessian-exp}. Recall from \cite[Lemmas~10 and 12]{Patel2025}
that
\begin{equation}
\frac{\partial}{\partial x}e^{A(x)}=\frac{1}{2}\left\{ \Phi_{A(x)}\!\left(\frac{\partial}{\partial x}A(x)\right),e^{A(x)}\right\} ,\label{eq:deriv-exp-fourier-trans}
\end{equation}
where the quantum channel $\Phi_{A(x)}$ is defined as
\begin{equation}
\Phi_{A(x)}(Y)\coloneqq\int_{-\infty}^{\infty}dt\,\gamma(t)e^{-iA(x)t}Ye^{iA(x)t},
\end{equation}
and the high-peak tent probability density $\gamma(t)$ is defined
in~\eqref{eq:high-peak-tent-prob-dens} (see also~\cite{Hastings2007,Kim2012,Ejima2019,Kato2019,Anshu2021}).
Starting from~\eqref{eq:hessian-mid-proof} and applying~\eqref{eq:deriv-exp-fourier-trans},
while noting that
\begin{equation}
\frac{\partial}{\partial\mu_{i}}\left(\frac{1}{T}\left(H-\mu\cdot Q\right)\right)=-\frac{Q_{i}}{T},
\end{equation}
we find that
\begin{align}
 & \frac{\partial^{2}}{\partial\mu_{i}\partial\mu_{j}}f_{T}(\mu)\nonumber \\
 & =\Tr\!\left[M_{T}(\mu)\left[\frac{\partial}{\partial\mu_{i}}e^{\frac{1}{T}\left(H-\mu\cdot Q\right)}\right]M_{T}(\mu)Q_{j}\right]\\
 & =-\frac{1}{2T}\Tr\!\left[M_{T}(\mu)\left\{ \Phi_{\mu}\!\left(Q_{i}\right),e^{\frac{1}{T}\left(H-\mu\cdot Q\right)}\right\} M_{T}(\mu)Q_{j}\right]\\
 & =-\frac{1}{2T}\Tr\!\left[M_{T}(\mu)\Phi_{\mu}\!\left(Q_{i}\right)e^{\frac{1}{T}\left(H-\mu\cdot Q\right)}M_{T}(\mu)Q_{j}\right]\nonumber \\
 & \qquad-\frac{1}{2T}\Tr\!\left[M_{T}(\mu)e^{\frac{1}{T}\left(H-\mu\cdot Q\right)}\Phi_{\mu}\!\left(Q_{i}\right)M_{T}(\mu)Q_{j}\right]\\
 & \overset{(a)}{=}-\frac{1}{2T}\Tr\!\left[M_{T}(\mu)\Phi_{\mu}\!\left(Q_{i}\right)\left(I-M_{T}(\mu)\right)Q_{j}\right]\nonumber \\
 & \qquad-\frac{1}{2T}\Tr\!\left[\left(I-M_{T}(\mu)\right)\Phi_{\mu}\!\left(Q_{i}\right)M_{T}(\mu)Q_{j}\right]\\
 & =-\frac{1}{T}\Re\!\left[\Tr\!\left[M_{T}(\mu)\Phi_{\mu}(Q_{i})\left(I-M_{T}(\mu)\right)Q_{j}\right]\right],
\end{align}
where equality (a) follows from~\eqref{eq:scalar-ident-FD-meas-compl}.

\subsection{Proof of Proposition~\ref{prop:hessian-NSD-spec-up-bnd} (bounds
on Hessian of dual objective function)}

\label{sec:hessian-NSD-spec-up-bnd}We first prove that the Hessian
is negative semidefinite. Let $v\in\mathbb{R}^{c}$ be arbitrary,
and let $\nabla^{2}f_{T}(\mu)$ denote the Hessian matrix for $f_{T}(\mu)$,
with elements as given in~\eqref{eq:hessian-elements}. Consider that
\begin{align}
 & -v^{T}\nabla^{2}f_{T}(\mu)v\nonumber \\
 & =-\sum_{i,j\in\left[c\right]}v_{i}\left[\frac{\partial^{2}}{\partial\mu_{i}\partial\mu_{j}}f(\mu)\right]v_{j}\\
 & =\frac{1}{T}\sum_{i,j\in\left[c\right]}v_{i}\int_{0}^{1}ds\,\Tr\!\left[Q_{i}M_{T}(\mu,1-s)Q_{j}M_{T}(\mu,s)\right]v_{j}\\
 & =\frac{1}{T}\int_{0}^{1}ds\,\Tr\!\left[\begin{array}{c}
\left(\sum_{i\in\left[c\right]}v_{i}Q_{i}\right)M_{T}(\mu,1-s)\times\\
\left(\sum_{j\in\left[c\right]}v_{j}Q_{j}\right)M_{T}(\mu,s)
\end{array}\right]\\
 & =\frac{1}{T}\int_{0}^{1}ds\,\Tr\!\left[WM_{T}(\mu,1-s)WM_{T}(\mu,s)\right]\label{eq:Hessian-concave-last-step}\\
 & \geq0,
\end{align}
where $W\coloneqq\sum_{i\in\left[c\right]}v_{i}Q_{i}$. The last inequality
follows because the matrix $M_{T}(\mu,1-s)$ is positive semidefinite,
which implies that $WM_{T}(\mu,1-s)W$ is positive semidefinite given
that $W$ is Hermitian. Also, the matrix $M_{T}(\mu,s)$ is positive
semidefinite. As such, the trace expression in~\eqref{eq:Hessian-concave-last-step}
is non-negative for all $s\in\left[0,1\right]$, so that the integral
is non-negative also.

Consider that
\begin{align}
 & \left|\frac{\partial^{2}}{\partial\mu_{i}\partial\mu_{j}}f(\mu)\right|\nonumber \\
 & =\frac{1}{T}\left|\int_{0}^{1}ds\,\Tr\!\left[Q_{i}M_{T}(\mu,1-s)Q_{j}M_{T}(\mu,s)\right]\right|\\
 & \leq\frac{1}{T}\int_{0}^{1}ds\,\left|\Tr\!\left[M_{T}(\mu,s)Q_{i}M_{T}(\mu,1-s)Q_{j}\right]\right|\\
 & \overset{(a)}{\leq}\frac{1}{T}\int_{0}^{1}ds\,\left\Vert M_{T}(\mu,s)\right\Vert \left\Vert Q_{i}\right\Vert _{1}\left\Vert M_{T}(\mu,1-s)\right\Vert \left\Vert Q_{j}\right\Vert \\
 & \leq\frac{1}{T}\left\Vert Q_{i}\right\Vert _{1}\left\Vert Q_{j}\right\Vert .\label{eq:hessian-elements-bound}
\end{align}
The inequality (a) follows from the H\"older inequality, as well
as the upper bound
\begin{equation}
\frac{e^{sx}}{e^{x}+1}\leq1,\label{eq:scalar-ineq-meas-ops-interp}
\end{equation}
holding for all $s\in\left[0,1\right]$ and $x\in\mathbb{R}$, which
implies that the following inequality holds for all $s\in\left[0,1\right]$:
\begin{equation}
\max\left\{ \left\Vert M_{T}(\mu,s)\right\Vert ,\left\Vert M_{T}(\mu,1-s)\right\Vert \right\} \leq1.
\end{equation}
To see~\eqref{eq:scalar-ineq-meas-ops-interp}, consider that
\begin{equation}
\frac{e^{tx}}{e^{x}+1}\leq1\qquad\Longleftrightarrow\qquad e^{tx}\leq e^{x}+1.
\end{equation}
If $x\geq0$, then $e^{tx}\leq e^{x}$, which in turn implies that
$e^{tx}\leq e^{x}+1$. If $x\leq0$, then $e^{tx}\leq1$, which in
turn implies that $e^{tx}\leq e^{x}+1$.

An alternative proof of the upper bound in~\eqref{eq:hessian-elements-bound}
is as follows, making use of~\eqref{eq:2nd-hessian-exp}:
\begin{align}
 & \left|\frac{\partial^{2}}{\partial\mu_{i}\partial\mu_{j}}f(\mu)\right|\nonumber \\
 & =\left|-\frac{1}{T}\Re\!\left[\Tr\!\left[M_{T}(\mu)\Phi_{\mu}(Q_{i})\left(I-M_{T}(\mu)\right)Q_{j}\right]\right]\right|\\
 & \leq\frac{1}{T}\left|\Tr\!\left[M_{T}(\mu)\Phi_{\mu}(Q_{i})\left(I-M_{T}(\mu)\right)Q_{j}\right]\right|\\
 & \leq\frac{1}{T}\left\Vert M_{T}(\mu)\right\Vert \left\Vert \Phi_{\mu}(Q_{i})\right\Vert _{1}\left\Vert I-M_{T}(\mu)\right\Vert \left\Vert Q_{j}\right\Vert \\
 & \leq\frac{1}{T}\left\Vert Q_{i}\right\Vert _{1}\left\Vert Q_{j}\right\Vert .\label{eq:hessian-elements-bound-1}
\end{align}
The last inequality follows because $M_{T}(\mu)$ and $I-M_{T}(\mu)$
are measurement operators, and because $\left\Vert \Phi_{\mu}(Q_{i})\right\Vert _{1}\leq\left\Vert Q_{i}\right\Vert _{1}$,
given that the trace norm is convex and unitarily invariant.

Let us now obtain an upper bound on the largest singular value of
the Hessian (i.e., its spectral norm). Consider that
\begin{align}
\left\Vert \nabla^{2}f_{T}(\mu)\right\Vert  & \overset{(a)}{\leq}\left\Vert \nabla^{2}f_{T}(\mu)\right\Vert _{1}\\
 & \overset{(b)}{=}-\Tr\!\left[\nabla^{2}f_{T}(\mu)\right]\\
 & \leq\sum_{i\in\left[c\right]}\left|\frac{\partial^{2}}{\partial\mu_{i}^{2}}f(\mu)\right|\\
 & \overset{(c)}{\leq}\frac{1}{T}\sum_{i\in\left[c\right]}\left\Vert Q_{i}\right\Vert _{1}\left\Vert Q_{i}\right\Vert .
\end{align}
The inequality (a) follows because the spectral norm does not exceed
the trace norm. The equality (b) follows because the Hessian matrix
$\nabla^{2}f_{T}(\mu)$ is negative semidefinite. The inequality (c)
follows from~\eqref{eq:hessian-elements-bound}.

\section{Proof of Proposition~\ref{prop:implementing-FD-measurements} (correctness
of Algorithm~\ref{alg:FD-thermal-alg} for simulating Fermi--Dirac
thermal measurements)}

\label{sec:implementing-FD-measurements}

Let $A$ be a $d\times d$ Hermitian matrix, and let $T_{1},T_{2}>0$.
The state vector $|\psi_{T}\rangle$ defined in~\eqref{eq:control-state-alg-FD}
is indeed a state (normalized) for all $T>0$ because
\begin{align}
\langle\psi_{T}|\psi_{T}\rangle & =\int_{-\infty}^{\infty}dp\,\frac{e^{p/T}}{T\left(e^{p/T}+1\right)^{2}}\\
 & =\int_{-\infty}^{\infty}dp\,\frac{d}{dp}\left(\frac{1}{e^{-p/T}+1}\right)\\
 & =\left.\frac{1}{e^{-p/T}+1}\right|_{-\infty}^{\infty}\\
 & =1-0\\
 & =1.
\end{align}

To prove Proposition~\ref{prop:implementing-FD-measurements}, let
us begin by supposing that the state of the data register is pure
and given by $|\varphi\rangle\!\langle\varphi|$, where $|\varphi\rangle$
is a state vector. Thus, the probability that Algorithm~\ref{alg:FD-thermal-alg}
outputs $0$ is equal to
\begin{equation}
\int_{0}^{\infty}dp\,\left\Vert \left(\langle p|\otimes I\right)e^{-i\hat{x}\otimes A/T_{2}}\left(|\psi_{T_{1}}\rangle\otimes|\varphi\rangle\right)\right\Vert ^{2}.
\end{equation}
Let a spectral decomposition of $A$ be given as
\begin{equation}
A=\sum_{i}a_{i}|i\rangle\!\langle i|.
\end{equation}
This implies that
\begin{equation}
e^{-i\hat{x}\otimes A/T_{2}}=\sum_{i}\int dx\,e^{-ixa_{i}/T_{2}}|x\rangle\!\langle x|\otimes|i\rangle\!\langle i|.
\end{equation}
Then
\begin{align}
 & \left(\langle p|\otimes I\right)e^{-i\hat{x}\otimes A/T_2}\left(|\psi_{T_{1}}\rangle\otimes I\right)\nonumber \\
 & =\left(\langle p|\otimes I\right)\left(\sum_{i}\int dx\,e^{-ixa_{i}/T_{2}}|x\rangle\!\langle x|\otimes|i\rangle\!\langle i|\right)\times\nonumber \\
 & \qquad\left(\int_{-\infty}^{\infty}dp'\,\sqrt{\frac{e^{p'/T_{1}}}{T_{1}\left(e^{p'/T_{1}}+1\right)^{2}}}|p'\rangle\otimes I\right)\\
 & =\sum_{i}\int dx\int_{-\infty}^{\infty}dp'\,e^{-ixa_{i}/T_{2}}\langle p|x\rangle\langle x|p'\rangle\times\nonumber \\
 & \qquad\sqrt{\frac{e^{p'/T_{1}}}{T_{1}\left(e^{p'/T_{1}}+1\right)^{2}}}|i\rangle\!\langle i|\\
 & =\sum_{i}\int_{-\infty}^{\infty}dp'\,\delta(p'-p-a_{i}/T_{2})\sqrt{\frac{e^{p'/T_{1}}}{T_{1}\left(e^{p'/T_{1}}+1\right)^{2}}}|i\rangle\!\langle i|\\
 & =\sum_{i}\sqrt{\frac{e^{\left(p+a_{i}/T_{2}\right)/T_{1}}}{T_{1}\left(e^{\left(p+a_{i}/T_{2}\right)/T_{1}}+1\right)^{2}}}|i\rangle\!\langle i|,
\end{align}
where the penultimate equality follows because
\begin{align}
 & \int_{-\infty}^{\infty}dx\,e^{-ixa_{i}/T_{2}}\langle p|x\rangle\langle x|p'\rangle\nonumber \\
 & =\frac{1}{2\pi}\int_{-\infty}^{\infty}dx\,e^{-ixa_{i}/T_{2}}e^{-ipx}e^{ip'x}\\
 & =\frac{1}{2\pi}\int dx\,e^{i\left(p'-p-a_{i}/T_{2}\right)x}\\
 & =\delta(p'-p-a_{i}/T_{2}).
\end{align}
So then
\begin{align}
 & \left\Vert \left(\langle p|\otimes I\right)e^{-i\hat{x}\otimes A/T_{2}}\left(|\psi_{T_{1}}\rangle\otimes|\varphi\rangle\right)\right\Vert ^{2}\nonumber \\
 & =\left\Vert \left(\sum_{i}\sqrt{\frac{e^{\left(p+a_{i}/T_{2}\right)/T_{1}}}{T_{1}\left(e^{\left(p+a_{i}/T_{2}\right)/T_{1}}+1\right)^{2}}}|i\rangle\!\langle i|\right)|\varphi\rangle\right\Vert ^{2}\\
 & =\langle\varphi|\left(\sum_{i}\frac{e^{\left(p+a_{i}/T_{2}\right)/T_{1}}}{T_{1}\left(e^{\left(p+a_{i}/T_{2}\right)/T_{1}}+1\right)^{2}}|i\rangle\!\langle i|\right)|\varphi\rangle,
\end{align}
which implies that
\begin{align}
 & \int_{0}^{\infty}dp\,\left\Vert \left(\langle p|\otimes I\right)e^{-i\hat{x}\otimes A/T_{2}}\left(|\psi_{T_{1}}\rangle\otimes|\varphi\rangle\right)\right\Vert ^{2}\nonumber \\
 & =\int_{0}^{\infty}dp\,\langle\varphi|\left(\sum_{i}\frac{e^{\left(p+a_{i}/T_{2}\right)/T_{1}}}{T_{1}\left(e^{\left(p+a_{i}/T_{2}\right)/T_{1}}+1\right)^{2}}|i\rangle\!\langle i|\right)|\varphi\rangle\\
 & =\langle\varphi|\left(\sum_{i}\left[\int_{0}^{\infty}dp\,\frac{e^{\left(p+a_{i}/T_{2}\right)/T_{1}}}{T_{1}\left(e^{\left(p+a_{i}/T_{2}\right)/T_{1}}+1\right)^{2}}\right]|i\rangle\!\langle i|\right)|\varphi\rangle\\
 & =\langle\varphi|\left(\sum_{i}\frac{1}{e^{a_{i}/\left(T_{1}T_{2}\right)}+1}|i\rangle\!\langle i|\right)|\varphi\rangle\\
 & =\langle\varphi|\left(e^{A/\left(T_{1}T_{2}\right)}+I\right)^{-1}|\varphi\rangle\\
 & =\Tr\!\left[M_{T_{1}T_{2}}(A)|\varphi\rangle\!\langle\varphi|\right].\label{eq:proof-FD-alg-pure-states}
\end{align}
The penultimate equality follows because, for all $a\in\mathbb{R}$
and $T>0$,
\begin{align}
 & \int_{0}^{\infty}dp\,\frac{e^{\left(p+a\right)/T}}{T\left(e^{\left(p+a\right)/T}+1\right)^{2}}\nonumber \\
 & =\int_{a}^{\infty}du\,\frac{e^{u/T}}{T\left(e^{u/T}+1\right)^{2}}\\
 & =\int_{a}^{\infty}du\,\frac{d}{du}\left(\frac{1}{e^{-u/T}+1}\right)\\
 & =\left.\frac{1}{e^{-u/T}+1}\right|_{a}^{\infty}\\
 & =1-\frac{1}{e^{-a/T}+1}\\
 & =\frac{e^{-a/T}}{e^{-a/T}+1}\\
 & =\frac{1}{e^{a/T}+1}.
\end{align}
The result generalizes to an arbitrary state $\rho$ because every
such state can be written as a convex combination of pure states as
\begin{equation}
\rho=\sum_{x}p(x)|\varphi_{x}\rangle\!\langle\varphi_{x}|,
\end{equation}
so that, defining $\left[p\right]\equiv|p\rangle\!\langle p|$, $\psi_{T_{1}}\equiv|\psi_{T_{1}}\rangle\!\langle\psi_{T_{1}}|$,
and $\varphi_{x}\equiv|\varphi_{x}\rangle\!\langle\varphi_{x}|$,
\begin{align}
 & \int_{0}^{\infty}dp\,\Tr\!\left[
 \begin{array}{c}
 \left(|p\rangle\!\langle p|\otimes I\right)e^{-i\hat{x}\otimes A/T_{2}}\times\\
\left(|\psi_{T_{1}}\rangle\!\langle\psi_{T_{1}}|\otimes\rho\right)e^{i\hat{x}\otimes A/T_{2}}
 \end{array}\right]\nonumber \\
 & =\sum_{x}p(x)\int_{0}^{\infty}dp\,\Tr\!\left[\begin{array}{c}
\left(\left[p\right]\otimes I\right)e^{-i\hat{x}\otimes A/T_{2}}\times\\
\left(\psi_{T_{1}}\otimes\varphi_{x}\right)e^{i\hat{x}\otimes A/T_{2}}
\end{array}\right]\\
 & =\sum_{x}p(x)\int_{0}^{\infty}dp \left\Vert \left(\langle p|\otimes I\right)e^{-i\hat{x}\otimes A/T_{2}}\left(|\psi_{T_{1}}\rangle|\varphi_{x}\rangle\right)\right\Vert ^{2}\\
 & =\sum_{x}p(x)\Tr\!\left[M_{T_{1}T_{2}}(A)|\varphi_{x}\rangle\!\langle\varphi_{x}|\right]\\
 & =\Tr\!\left[M_{T_{1}T_{2}}(A)\rho\right],
\end{align}
where the penultimate equality follows from~\eqref{eq:proof-FD-alg-pure-states}.
\end{document}